\documentclass[3p]{elsarticle}
\journal{Computer Methods in Applied Mechanics and Engineering}
\usepackage{graphicx}
\usepackage{amsmath}
\usepackage{amssymb}
\usepackage{amsthm}
\usepackage{latexsym}
\usepackage{bm}
\usepackage{array,supertabular}
\usepackage{color}
\usepackage{framed}
\usepackage{setspace}
\usepackage{fancyhdr}
\usepackage{stmaryrd}
\usepackage{mathrsfs}

\SetSymbolFont{stmry}{bold}{U}{stmry}{m}{n}

\newcommand{\bsfB}{\boldsymbol{\mathsf{B}}}
\newcommand{\bsfF}{\boldsymbol{\mathsf{F}}}
\newcommand{\bsfG}{\boldsymbol{\mathsf{G}}}
\newcommand{\bsfP}{\boldsymbol{\mathsf{P}}}
\newcommand{\bsfI}{\boldsymbol{\mathsf{I}}}
\newcommand{\bsfV}{\boldsymbol{\mathsf{V}}}
\newcommand{\bsfw}{\boldsymbol{\mathsf{w}}}
\newcommand{\bsfy}{\boldsymbol{\mathsf{y}}}

\newtheorem{proposition}{Proposition}
\newtheorem{remark}{Remark}

\numberwithin{equation}{section}

\begin{document}
\begin{frontmatter}
\title{A unified continuum and variational multiscale formulation for fluids, solids, and fluid-structure interaction}
\author{Ju Liu}
\ead{liuju@stanford.edu}

\author{Alison L. Marsden}
\ead{amarsden@stanford.edu}

\address{Department of Pediatrics (Cardiology), Department of Bioengineering, and Institute for Computational and Mathematical Engineering, Stanford University, Clark Center E1.3, 318 Campus Drive, Stanford, CA 94305, USA}

\begin{abstract}
We develop a unified continuum modeling framework using the Gibbs free energy as the thermodynamic potential. This framework naturally leads to a pressure primitive variable formulation for the continuum body, which is well-behaved in both compressible and incompressible regimes. Our derivation also provides a rational justification of the isochoric-volumetric additive split of free energies in nonlinear elasticity. The variational multiscale analysis is performed for the continuum model to construct a foundation for numerical discretization. We first consider the continuum body instantiated as a hyperelastic material and develop a variational multiscale formulation for the hyper-elastodynamic problem. The generalized-$\alpha$ method is applied for temporal discretization. A segregated algorithm for the nonlinear solver, based on the original idea introduced in \cite{Scovazzi2016}, is carefully analyzed. Second, we apply the new formulation to construct a novel unified formulation for fluid-solid coupled problems. The variational multiscale formulation is utilized for spatial discretization in both fluid and solid subdomains. The generalized-$\alpha$ method is applied for the whole continuum body, and optimal high-frequency dissipation is achieved in both fluid and solid subproblems. A new predictor multi-corrector algorithm is developed based on the segregated algorithm. The efficacy of the new formulations is examined in several benchmark problems. The results indicate that the proposed modeling and numerical methodologies constitute a promising technology for biomedical and engineering applications, particularly those necessitating incompressible models.
\end{abstract}

\end{frontmatter}

\section{Introduction}
Continued advancement in the variational multiscale (VMS) method for computational fluid dynamics (CFD) \cite{Bazilevs2007a,Hughes2004,Zeng2016}, multiscale boundary conditions to model the distal vasculature \cite{Vignon-Clementel2006}, numerical optimization methods \cite{Marsden2014,Yang2010}, novel coupling procedures for fluid-structure integration (FSI) \cite{Bazilevs2006,Figueroa2006,Moghadam2013}, and new solver technologies \cite{Moghadam2013a,Moghadam2015} have led to increasingly sophisticated simulation technologies of three-dimensional patient-specific cardiovascular problems \cite{Marsden2013,Marsden2015,Taylor2009a}. These simulation methods have been applied to investigate a wide range of cardiovascular problems with increasing clinical utility, such as coronary artery disease \cite{Ramachandra2016}, aneurysms \cite{Les2010}, and congenital heart disease \cite{Yang2012}. These advances have also led to new open problems and a call for new computational technologies in biomedical modeling \cite{Taylor2009}. In particular, there is a pressing need to accurately predict transmural stresses for nonlinear, anisotropic, nearly incompressible, viscoelastic materials in the setting of FSI. This technology will benefit the investigation of long-term vascular growth and remodeling driven by mechanical forces and mechanobiological response. This work represents our first step towards developing a robust, stable, accurate, and efficient finite element technology to address the aforementioned need. We focus on (1) constructing a unified modeling framework for solid and fluid dynamics, (2) developing a new numerical formulation to handle both compressible and incompressible isotropic hyperelastic materials in a consistent manner, and (3) formulating a new coupling procedure for FSI problems. It is worth mentioning that the technology proposed in this work, although motivated by biomedical problems, is also applicable to a wide range of more general engineering problems. In this section, we will first review traditional methods for nonlinear incompressible solid dynamics and then introduce our new formulation based on a new continuum basis. After that, we will discuss a new coupling approach for FSI problems. Lastly, we will provide an outline of the body of this work.

\subsection{Continuum mechanics and numerical methods for solid dynamics}
Incompressible solids refer to materials that undergo very small volume changes under mechanical loading. This property is ubiquitous in elastic-plastic materials, rubber-like materials, and biomaterials. From a numerical point of view, incompressibility poses a strict restriction on the choice of suitable methods. The classical treatment of incompressibility boils down to the use of a displacement-pressure element pair that satisfies the celebrated Ladyzhenskaya-Babu\v{s}ka-Brezzi (LBB) condition \cite{Brezzi1991}. An improper choice of elements may lead to mesh locking and spurious pressure modes \cite[Chapter 4]{Hughes1987}. In practical calculations, the auxiliary pressure variable in the mixed formulation is not favored due to the additional computational cost. The cost-effectiveness of an algorithm was emphasized during the early days of finite element method development. Consequently, there has been a strong motivation for developing displacement elements that circumvent the volumetric locking phenomenon in structural analysis. One possible solution is invoking higher-order finite element methods \cite{Heisserer2008}. Nevertheless, low-order elements still enjoy significant popularity in nonlinear problems for at least two reasons. First, they are simple to implement and to use in conjunction with adaptive mesh generation; second, they are robust in nonlinear dynamic analysis. One popular approach for the low-order displacement elements is the $\bar{\textup{B}}$ and $\bar{\textup{F}}$ projection methods \cite{SouzaNeto1996,Elguedj2008,Hughes1980}. The idea of these projection methods goes back to the mean-dilatation approach \cite{Nagtegaal1974} and the reduced and selective integration method \cite{Malkus1978}. In the $\bar{\textup{B}}$ projection method, the dilatational part of the strain is projected onto an appropriate function space to alleviate the volumetric constraint. The equivalence of this approach to the mixed finite element is now well-understood \cite{Hughes1977,Malkus1978}. The $\bar{\textup{F}}$ projection method is a generalization of the $\bar{\textup{B}}$ method for large deformation problems \cite{Hughes1975,Simo1985}. Without a doubt, the $Q_1/Q_0$ element, as the low-order quadrilateral/hexahedral element, is the most widely used element in structural analysis, with $Q_1$ for the displacement discretization space and $Q_0$ for the projected dilatational strain space. Interested readers are referred to \cite[Chapter 4]{Hughes1987} for a comprehensive review on this topic.

For problems with complex geometries, mesh generation poses an additional constraint on the choice of numerical methods. Automated hexahedral mesh generation still poses significant challenges for practical problems and remains a labor-intensive and time-consuming process \cite{Shepherd2008}. On the other hand, the algorithms for tetrahedral mesh generation are mature \cite{Geuzaine2009,Si2015} and are routinely used for complex problems \cite{Updegrove2017}. However, the aforementioned projection methods cannot be directly applied to low-order simplicial elements, since the $P_1/P_0$ element does not satisfy the LBB condition and will suffer from locking \cite{Hughes1969}. To overcome the above issues, particularly for problems with incompressible materials and complex geometries, it is imperative to enable low-order simplicial elements for nearly-incompressible materials. There are several approaches that have claimed success in this regard. The $\bar{\textup{F}}$ projection method is generalized by projecting the dilatational part of the deformation gradient onto a patch-wise constant space with predefined non-overlapping element patches \cite{SouzaNeto2004}. This method requires one to group elements into patches, and hence introduces additional complication in mesh generation as well as in the matrix assembly. In the mixed-enhanced approach, the compatible strain fields are augmented by an additional strain field. In doing so, the volumetric locking can be avoided with linear pressure interpolation \cite{Taylor2000}. This method bears some similarity with the MINI element and the stabilized method based on residual-free bubbles. However, those methods seem to be not completely satisfactory. It is still common to see that practitioners in the area of soft tissue modeling make compromise between the incompressibility material property and the geometrical complexity. 

Over the years, the stabilized methods have gained tremendous success in handling numerical instabilities in the finite element modeling of fluid flows \cite{Brooks1982,Hughes1986a}. Later, the VMS method was introduced to provide a rationale for the stabilized methods and a framework for designing new computational methods \cite{Bazilevs2007a,Hughes1995,Hughes1998,Hughes2004,Hughes2010,Scovazzi2012}. Among the stabilization techniques, the instability associated with the LBB condition was resolved by invoking a Petrov-Galerkin formulation and adding a perturbation term to the test function \cite{Hughes1986a}. It was proven that many implementationally convenient elements, such as the equal-order simplicial elements, are convergent for the Stokes problem by employing this stabilized formulation. Since the Stokes equations are formally identical to the displacement-pressure formulation for the incompressible isotropic elasticity equations, this stabilization technique can be directly applied to the incompressible elasticity equations. It is worth noting that the Petrov-Galerkin method has also been developed for the stress-displacement formulation of nearly-incompressible elasticity based on the classical Hellinger-Reissner principle \cite{Cervera2010,Franca1988}. In \cite{Klaas1999,Maniatty2001,Maniatty2002,Saracibar2006}, the stabilized method was applied to finite elasticity and plasticity by introducing a pressure stabilization term in an ad hoc way. In \cite{Masud2013}, the same static finite elasticity model was analyzed in the VMS framework. It should be pointed out that those formulations are restricted to static calculations and are not well-suited for dynamic fully-incompressible problems \cite{Scovazzi2016}.

In recent years, there is a growing interest in extending the stabilized and VMS methods to solid dynamics with particular interests in nearly incompressible and bending dominated scenarios. In \cite{Lee2014}, the linear momentum equation was augmented by a dynamic equation for the deformation gradient to improve the stress accuracy. The authors designed a Petrov-Galerkin finite element formulation for such system to provide stabilization for nearly incompressible materials. In \cite{Gil2014}, $J$, the determinant of the deformation gradient, was added as one additional independent variable to enhance the performance in the incompressible limit. In \cite{Bonet2015a}, the authors introduced a computational framework for polyconvex hyperelasticity by treating the deformation gradient, its cofactor, and its determinant as independent kinematic variables. This framework was numerically investigated in \cite{Bonet2015,Gil2016} by invoking the entropy variables (conjugate stresses) of these kinematic variables. The polyconvex hypothesis guarantees the existence of entropy variables and leads to a new Hellinger-Reissner type variational principle. In general, this computational framework enjoys second-order accuracy for the stress and strain and shows robust performance in compressible, nearly incompressible, and fully incompressible regimes. A drawback associated with this formulation is that there are nineteen independent variables for describing kinematics in three-dimensional problems, which makes the algorithm quite noneconomical. In parallel to the aforementioned development, a pressure-rate equation was proposed in \cite{Scovazzi2016}. It was shown that this rate-type equation for the pressure field significantly improves the calculations of nearly incompressible elastodynamic problems. In \cite{Rossi2016}, a VMS analysis was performed on this pressure-rate type formulation, and the authors coined it as ``dynamic variational multiscale method". This formulation shows robust performance in nearly incompressible and bending dominated scenarios and has been recently generalized to viscoelasticity \cite{Zeng2017}. However, due to the usage of $J$ as a variable in the model derivation, the authors had to treat the nearly incompressible and the fully incompressible cases separately \cite{Bonet2015a,Gil2016,Rossi2016,Scovazzi2016}. More importantly, considering several different mixed formulations have been proposed in the aforementioned works, we feel the role of those equations needs to be elucidated. In our opinion, the pressure rate equation \cite{Rossi2016,Scovazzi2016,Zeng2017} or the equation for $J$ \cite{Bonet2015,Gil2014} is a differential mass equation. Consequently, one should judiciously use the classical relation $\rho J = \rho_0$, and, in our opinion, it is redundant to solve $\rho J = \rho_0$ together with the differential mass equation in the discrete problem \cite{Scovazzi2016,Zeng2017}. This point will be clarified in Section \ref{sec:continuum_mechanics}.

In traditional approaches, the constitutive relations of finite elasticity are derived based on the Helmholtz free energy or the strain energy \cite{Holzapfel2000}. In non-equilibrium thermodynamics, one can alternatively derive equivalent formulations based on other thermodynamic potentials. These thermodynamic potentials are related through Legendre transformations. Among all the thermodynamic potentials, the Helmholtz free energy is favored for the discussion of constitutive modeling, likely because it is a function of temperature, specific volume, and the number of molecules. These variables are easy to measure and control in laboratories. For example, it is easy to measure the temperature while it is very hard, if not impossible, to measure its conjugate variable, the entropy. However, it should be pointed out that there are certain drawbacks of formulations based on the Helmholtz free energy. For example, in the incompressible limit, the Helmholtz free energy ceases to serve as a valid thermodynamic potential, since one of the state variables, the density, is constrained to be a constant \cite{Lowengrub1998}. Therefore, one has to perform a Legendre transformation with respect to the specific volume to transform the independent variable to the thermodynamic pressure. The resulting thermodynamic potential serves as a proper starting point for the discussion of incompressible materials. Deviating from traditional approaches, we propose that one should flexibly choose an appropriate thermodynamic potential for the discussion based on the particular problem. 

In this work, we consider the Gibbs free energy as the thermodynamic potential and derive the constitutive relations based on the Coleman-Noll approach \cite{Gurtin2009,Marsden1994,Truesdell1965}, with both elastic and viscous material responses taken into account. A by-product of this derivation is that the isochoric-volumetric additive split of the free energy, which was regarded as a postulate based on physical intuition, can be justified with solid arguments. The resulting system provides a unified framework for the viscous incompressible Navier-Stokes equations, the compressible Navier-Stokes equations, the compressible hyperelasticity, and the fully incompressible hyperelasticity. In particular, the resulting compressible Navier-Stokes equations derived from the Gibbs free energy are written in the pressure primitive variables. It is well-known that the Navier-Stokes equations are well-behaved in the pressure primitive variables, but not for variable sets where the density is an independent variable \cite{Hauke1998}. This fact partly corroborates our argument that the Helmholtz free energy degenerates in the incompressible limit. One attractive feature of a unified theory for fluids and solids is that this theory may serve as a bridge for generalizing the VMS method from computational fluid dynamics to computational solid dynamics. Indeed, the VMS method was originally proposed to provide ``a paradigm for computational mechanics" \cite{Hughes1998}. Yet, over the years, progress was made primarily in the area of CFD \cite{Bazilevs2008,Hughes2004,Hughes2010}. We hope that this work will spark further research on developing VMS formulations for computational mechanics and enable CFD experts to readily bring their expertise to bear on solid mechanics.


\subsection{Fluid-structure interaction}
Fluid-structure interaction problems refer to the coupling of the fluid and structural equations defined on non-overlapping domains with appropriate interface conditions. Over the past several decades, tremendous advancement has been made in computational FSI problems \cite{Bazilevs2008,Bueno2014,Kamensky2015,Long2012}. In terms of how the fluid-structure interface is treated, the numerical methods for FSI problems may be categorized into two major groups: the boundary-fitted approach and the immersed boundary approach. In the boundary-fitted approach, the fluid problem is typically posed in an arbitrary Lagrangian-Eulerian (ALE) coordinate system \cite{Hughes1981,Scovazzi2007}, and the fluid domain is represented by a mesh that is deformable with the Lagrangian solid mesh. The boundary-fitted approach enjoys the exact coupling conditions on the interface and accurate stress calculations near the interface, with the expense of moving the mesh. For problems involving very large deformations, it may be necessary to regenerate the mesh to maintain the mesh quality and the solution accuracy, which can be quite expensive for three-dimensional calculations \cite{Johnson1994}. 

The immersed boundary method was introduced as an alternative approach for FSI simulations \cite{Peskin1972}. This approach releases the requirement of the boundary-fitting condition and hence can be quite attractive for problems with very large deformations of the solid boundary, such as the heart valve dynamics \cite{Griffith2011,Kamensky2015}. One major shortcoming of the immersed boundary approach is the loss of accuracy near the interface. Hence, its applicability is limited when the interfacial physics is important. In this work, we adopt the boundary-fitted approach for our FSI formulation.

In this work, one critical feature that differs from traditional FSI formulations is that the solid problem is written in pressure primitive variables, and its governing equations constitute a set of first-order equations, including a differential mass balance equation, a set of linear momentum balance equations, and a set of kinematic equations. This formulation for the solid problem is formally similar to the ALE formulation of the incompressible Navier-Stokes equations, where the kinematic equations are replaced by the equations governing the mesh motion. Indeed, one can view our FSI formulation as a unified continuum body governed by the mass  and linear momentum balance equations, where, in the solid subdomain, the deviatoric part of the Cauchy stress is elastic and the problem is written in the Lagrangian reference frame; in the fluid subdomain, the deviatoric part of the Cauchy stress is viscous and the problem is written in the ALE coordinate system. This unified formulation naturally allows one to construct a uniform VMS formulation for spatial discretization, and the resulting numerical scheme provides a residual-based turbulence model for the fluid subproblem and a VMS scheme for the solid subproblem with the purpose of stabilizing the pressure instability arising from equal-order interpolations. 

Writing the FSI problem in a unified formulation also allows us to perform time integration in a uniform way. The generalized-$\alpha$ method has shown to be an accurate and robust temporal scheme for structural dynamics \cite{Chung1993}, fluid dynamics \cite{Jansen2000}, and FSI problems \cite{Bazilevs2008,Dettmer2006,Kuhl2003}. One issue associated with the traditional FSI formulation is that there is a mismatch for the choice of the parameters in the generalized-$\alpha$ scheme \cite[p.~120]{Bazilevs2013}. The structural dynamics problem is typically written in the pure displacement formulation and hence involves a second-order time derivative. In contrast, the fluid dynamics problem involves only a first-order time derivative. To achieve controllable high-frequency dissipation, the parameters in the generalized-$\alpha$ method are parametrized by the spectral radius of the amplification matrix at an infinitely large time step. This parametrization is different for first-order \cite{Chung1993} and second-order systems \cite{Jansen2000}. In traditional FSI formulations, this leads to a dilemma for the choice of parameters in the temporal scheme. In \cite{Bazilevs2008}, the parametrization was chosen based on the first-order system to provide optimal dissipation for the fluid problem. That choice sacrifices the dissipation for the solid problem. In \cite{Kuhl2003}, the optimal parametrization was chosen for the fluid and solid subdomains separately. Yet, that leads to a kinematic inconsistency on the fluid-solid interface. It has been observed that this kinematic inconsistency may lead to failure in FSI simulations, and an interpolation procedure along the fluid-solid interface was proposed to address this issue \cite{Joosten2010}. In this work, since the solid dynamics is written as a first-order system, the aforementioned numerical challenge is conveniently resolved. The generalized-$\alpha$ method is applied for the unified continuum problem using the parametrization for the first-order problem \cite{Chung1993}. One can get optimal numerical dissipation for both the fluid and solid subproblems.

Methods for solving the discretized FSI problem can be categorized into two families: staggered and monolithic methods. The latter can be further categorized into block-iterative, quasi-direct, and direct strategies \cite{Tezduyar2006}. A segregated algorithm for the solid dynamics, based on the original idea introduced in \cite{Scovazzi2016}, naturally leads to a new solution procedure for FSI problems. First, one solves the pressure and the velocity for the continuum body in the matrix problem. Second, the segregated algorithm for the solid is invoked to obtain the displacement field in the solid subdomain. Third, the solid displacement is applied as a boundary condition for the mesh motion problem in the fluid subdomain. This coupling procedure can be categorized as the quasi-direct solution strategy \cite[Chapter 6]{Bazilevs2013}. In comparison with traditional FSI formulations, the additional cost in this FSI formulation is mainly due to the introduction of the pressure variable in the solid problem. Considering the number of degrees of freedom in the solid subdomain is typically much smaller than that of the fluid subdomain, the new formulation in fact does not significantly increase the cost of the solution of the system. 

\subsection{Structure and content of the paper}
The body of this work is organized as follows. In Section \ref{sec:continuum_mechanics}, a unified continuum model is derived by choosing the Gibbs free energy as the thermodynamic potential. In Section \ref{sec:VMS}, we perform VMS analysis for the resulting continuum model. In Section \ref{sec:numerical_solid}, we derive a fully discrete scheme for hyper-elastodynamics. In Section \ref{sec:numerical_FSI}, the algorithm developed for solid dynamics is coupled with fluid dynamics and constitutes a novel FSI formulation. The coupling procedure and the implementation details for FSI problems are discussed. In Section \ref{sec:benchmark}, benchmark problems are studied to examine the numerical formulations. We draw conclusions in Section \ref{sec:conclusion}.

\section{Continuum Mechanics}
\label{sec:continuum_mechanics}
In this section, we begin by presenting the ALE formulation of the balance laws. Following that, we derive constitutive relations based on the Gibbs free energy using the Coleman-Noll type analysis. An interesting result is that the additive split of free energies can be justified in this derivation. We recover several familiar models within our modeling framework. In the last part, for various well-known volumetric energies, their Legendre transformations are derived and discussed.

\subsection{Continuum mechanics on moving domains}
In this section, we discuss the kinematics of a deformable body and present the balance equations defined in an ALE frame of reference. Let $\Omega_{\bm X}$, $\Omega_{\bm x}$, and $\Omega_{\bm \chi}$ be bounded open sets in $\mathbb R^{n_d}$, where $n_d$ represents the number of space dimensions. They represent the domain occupied by the continuum body in the material (Lagrangian), the current (Eulerian), and the referential frames, respectively. The Lagrangian-to-Eulerian map at time $t$ is a diffeomorphism defined as
\begin{align*}
\bm\varphi(\cdot, t) : \Omega_{\bm X} &\rightarrow \Omega_{\bm x} = \bm \varphi(\Omega_{\bm X}, t), \quad \forall t \geq 0, \\
\bm X &\mapsto \bm x = \bm \varphi(\bm X, t), \quad \forall \bm X \in \Omega_{\bm X}.
\end{align*}
This map satisfies that $\bm \varphi(\bm X, 0) = \bm X$, which implies that $\bm x$ is the current location of a material particle whose initial location is $\bm X$. The displacement and the velocity of the material particle are given by
\begin{align*}
\bm u &:= \bm \varphi(\bm X, t) - \bm \varphi(\bm X, 0) = \bm \varphi(\bm X, t) - \bm X, \\
\bm v &:= \left. \frac{\partial \bm \varphi}{\partial t}\right|_{\bm X}= \left. \frac{\partial \bm u}{\partial t}\right|_{\bm X} = \frac{d\bm u}{dt}.
\end{align*}
In this work, $d\left( \cdot \right)/dt$ designates a total time derivative. The deformation gradient and the Jacobian determinant are defined as
\begin{align*}
\bm F := \frac{\partial \bm \varphi}{\partial \bm X}, \qquad
J := \textup{det}\left(\bm F \right).
\end{align*}
The referential-to-Eulerian map at time $t$ is a diffeomorphism defined as
\begin{align*}
\hat{\bm \varphi}(\cdot, t) : \Omega_{\bm \chi} &\rightarrow \Omega_{\bm x} = \hat{\bm \varphi}(\Omega_{\bm \chi},t), \quad \forall t \geq 0, \\
\bm \chi &\mapsto \bm x = \hat{\bm \varphi}(\bm \chi, t), \quad \forall \bm \chi \in \Omega_{\bm \chi},
\end{align*}
and satisfies
$
\hat{\bm \varphi}(\bm \chi, 0) = \bm \chi.
$
The mesh displacement and the mesh velocity are defined as
\begin{align}
\label{eq:ale_mesh_displacement_def}
\hat{\bm u} &:= \hat{\bm \varphi}(\bm \chi, t) - \hat{\bm \varphi}(\bm \chi, 0) = \hat{\bm \varphi}(\bm \chi, t) - \bm \chi, \\
\label{eq:ale_mesh_velocity_def}
\hat{\bm v} &:= \left. \frac{\partial \hat{\bm \varphi}}{\partial t}\right|_{\bm \chi}= \left. \frac{\partial \hat{\bm u}}{\partial t}\right|_{\bm \chi}.
\end{align}
The mesh deformation gradient and the mesh Jacobian determinant are defined as
\begin{align*}
\hat{\bm F} := \frac{\partial \hat{\bm \varphi}}{\partial \bm \chi}, \qquad
\hat{J} := \textup{det}(\hat{\bm F} ).
\end{align*}
It also proves convenient to introduce the Lagrangian-to-referential mapping at time $t$ as
\begin{align*}
\tilde{\bm \varphi}(\cdot, t) : \Omega_{\bm X} &\rightarrow \Omega_{\bm \chi}^t = \tilde{\bm \varphi}(\Omega_{\bm X},t) \quad \forall t \geq 0, \displaybreak[2] \\
\bm X &\mapsto \bm \chi = \tilde{\bm \varphi}(\bm X, t), \quad \forall \bm X \in \Omega_{\bm X}.
\end{align*}
One can analogously define the displacement, velocity, and deformation gradient for the motion of the referential frame relative to the Lagrangian reference frame \cite{Scovazzi2007a}. The three mappings introduced above are related through an operator composition
\begin{align}
\label{eq:ALE_map_composition}
\bm \varphi = \hat{\bm \varphi} \circ \tilde{\bm \varphi}.
\end{align}
This operator composition is illustrated in Figure \ref{fig:ale_maps}. 
\begin{figure}
	\begin{center}
	\begin{tabular}{c}
\includegraphics[angle=0, trim=120 100 80 50, clip=true, scale = 0.4]{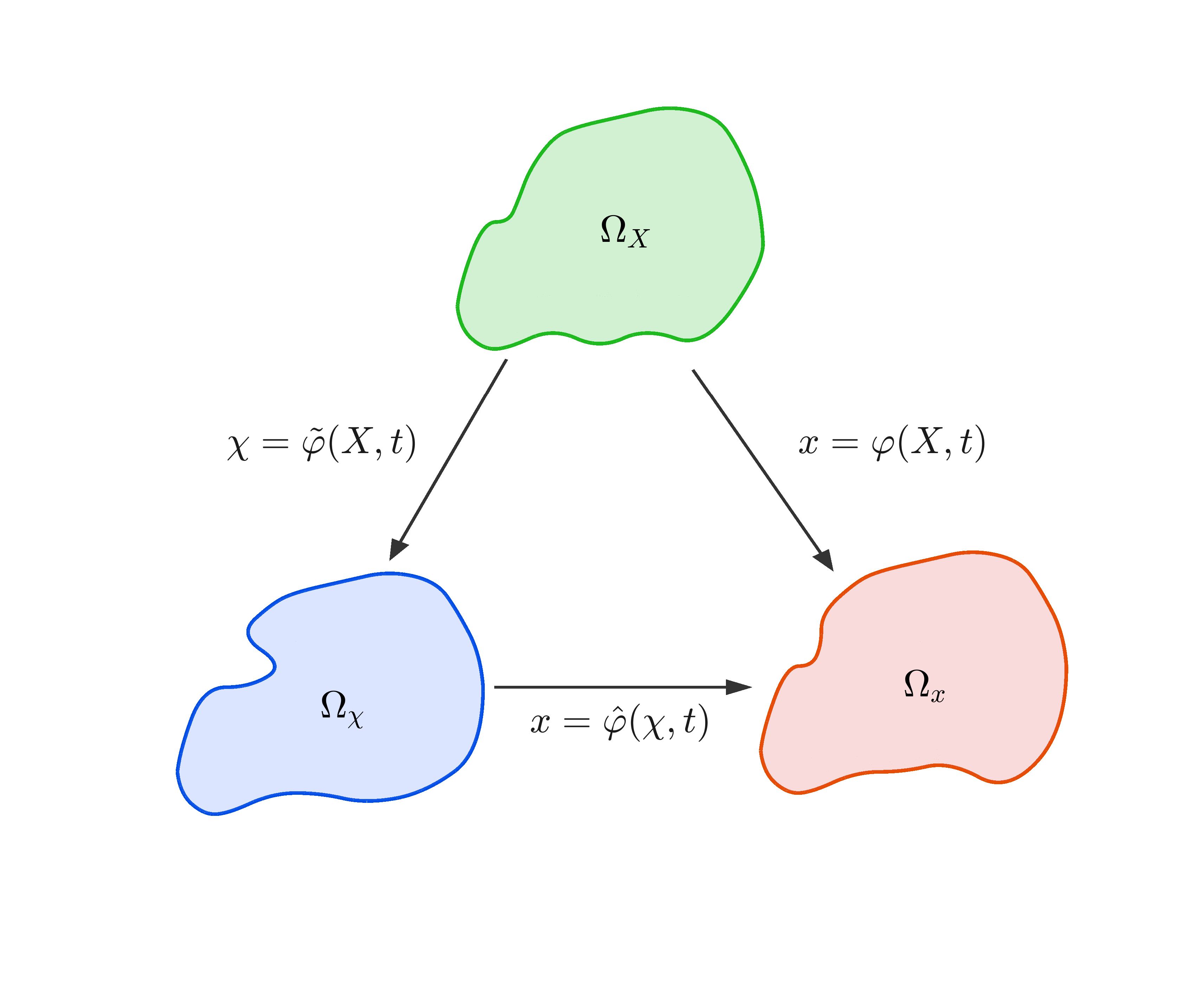}
\end{tabular}
\end{center}
\caption{Illustration of the diffeomorphisms $\bm \varphi$, $\hat{\bm \varphi}$, and $\tilde{\bm \varphi}$.} 
\label{fig:ale_maps}
\end{figure}
Let $\rho = \rho(\bm x, t)$ denote the density. The conservation of mass represented on the referential domain $\Omega_{\bm \chi}$ is
\begin{align}
\label{eq:referential_mass_banalce}
0 = \int_{\Omega_{\bm \chi}} \left. \frac{\partial (\hat{J}\rho)}{\partial t} \right|_{\bm \chi} + \nabla_{\bm \chi} \cdot \left(\hat{J}\hat{\bm F}^{-1}\rho \left(\bm v - \hat{\bm v} \right) \right)d\Omega_{\bm \chi}.
\end{align}
Let $\bm \sigma$ denote the Cauchy stress, and let $\bm b$ denote the body force per unit mass. The linear momentum balance equation on the referential domain $\Omega_{\bm \chi}$ is
\begin{align}
\label{eq:referential_linear_momentum_balance}
\bm 0 = \int_{\Omega_{\bm \chi}} \left. \frac{\partial (\hat{J}\rho \bm v)}{\partial t} \right|_{\bm \chi} + \nabla_{\bm \chi} \cdot \left( \left( \rho \bm v \otimes \left(\bm v - \hat{\bm v} \right) - \bm \sigma \right) \hat{J}\hat{\bm F}^{-T} \right) - \hat{J}\rho \bm bd\Omega_{\bm \chi}.
\end{align}
The balance of angular momentum is given by
\begin{align}
\label{eq:referential_angular_momentum_balance}
\bm 0 = \int_{\Omega_{\bm \chi}} \left. \frac{\partial (\hat{J} \bm x \times \rho \bm v)}{\partial t} \right|_{\bm \chi} + \nabla_{\bm \chi} \cdot \left( \left( \left( \bm x \times \rho \bm v\right) \otimes \left(\bm v - \hat{\bm v} \right) - \bm x \times \bm \sigma \right) \hat{J}\hat{\bm F}^{-T} \right) - \hat{J} \bm x \times \rho \bm b d\Omega_{\bm \chi}.
\end{align}
It can be shown that \eqref{eq:referential_angular_momentum_balance} is equivalent to the symmetry of the Cauchy stress,
\begin{align}
\bm \sigma = \bm \sigma^{T}.
\end{align}
Let $E:=\iota + \bm v \cdot \bm v / 2$ denote the total energy per unit mass, where $\iota$ is the internal energy per unit mass; let $\bm q$ denote the heat flux; let $r$ denote the heat source per unit mass. The balance of total energy is
\begin{align}
\label{eq:referential_total_energy_balance}
0 = \int_{\Omega_{\bm \chi}} \left. \frac{\partial (\hat{J}\rho E)}{\partial t} \right|_{\bm \chi} + \nabla_{\bm \chi} \cdot \left( \hat{J} \hat{\bm F}^{-1} \left(\rho E \left(\bm v - \hat{\bm v} \right) - \bm \sigma^{T}\bm v + \bm q \right) \right) - \hat{J}\rho \left( \bm b\cdot \bm v + r \right) d\Omega_{\bm \chi}.
\end{align}
The second law of thermodynamics can be stated as
\begin{align}
\label{eq:referential_2nd_law}
0 \leq \int_{\Omega_{\bm \chi}} \left. \frac{\partial (\hat{J}\rho s)}{\partial t} \right|_{\bm \chi} +  \nabla_{\bm \chi} \cdot \left( \hat{J}\hat{\bm F}^{-1} \left(\rho s \left( \bm v - \hat{\bm v}\right) + \frac{\bm q}{\theta} \right) \right) - \hat{J}\rho \frac{r}{\theta} d\Omega_{\bm \chi},
\end{align}
Using the Piola identity and the localization argument, one can derive the advective form of the balance equations and the second law of thermodynamics as follows,

\begin{align}
\label{eq:adv_form_mass_balance}
& 0 =  \left. \frac{\partial \rho}{\partial t} \right|_{\bm \chi} + \left(\bm v - \hat{\bm v} \right)\cdot \nabla_{\bm x} \rho + \rho \nabla_{\bm x} \cdot \bm v , \displaybreak[2] \\
\label{eq:adv_form_linear_momentum_balance}
& \bm 0 = \left. \rho \frac{\partial \bm v}{\partial t} \right|_{\bm \chi} + \rho \left( \nabla_{\bm x} \bm v \right) \left(\bm v - \hat{\bm v} \right) - \nabla_{\bm x} \cdot \bm \sigma - \rho \bm b, \displaybreak[2] \\
\label{eq:adv_form_angular_momentum_balance}
& \bm \sigma = \bm \sigma^T, \displaybreak[2] \\
\label{eq:adv_form_internal_energy_balance}
& 0 = \left. \rho \frac{\partial \iota}{\partial t} \right|_{\bm \chi} + \rho \left( \bm v - \hat{\bm v} \right) \cdot \nabla_{\bm x} \iota - \bm \sigma : \nabla_{\bm x} \bm v + \nabla_{\bm x} \cdot \bm q - \rho r , \displaybreak[2] \\
\label{eq:adv_form_second_law}
& 0 \leq  \left. \rho \frac{\partial s}{\partial t} \right|_{\bm \chi} + \rho \left(\bm v - \hat{\bm v} \right) \cdot \nabla_{\bm x} s +  \nabla_{\bm x} \cdot \left( \frac{\bm q}{\theta} \right) - \rho \frac{r}{\theta} .
\end{align}
Moreover, the displacement-velocity kinematic relation
\begin{align*}
\frac{d\bm u}{dt} = \bm v
\end{align*}
can be written in the advective form as
\begin{align}
\label{eq:adv_form_disp_velo_relation}
\left. \frac{\partial \bm u}{\partial t} \right|_{\bm \chi} + \left( \nabla_{\bm x} \bm u\right) \left( \bm v  - \hat{\bm v} \right)  = \bm v.
\end{align}
Details of the derivation of \eqref{eq:adv_form_mass_balance}-\eqref{eq:adv_form_second_law} can be found in \cite{Scovazzi2007}.

\begin{remark}
$\rho J = \rho_0$ is a widely used algebraic mass balance equation. However, we posit this is really a special case due to the Lagrangian description. In the general scenario, such as the ALE description, one has to adopt the differential equation \eqref{eq:adv_form_mass_balance} to describe mass conservation \cite{Liu1988}.
\end{remark}

\subsection{Constitutive relations}
There are multiple concepts of pressure in the literature. The mechanical pressure is defined as the dilatational (i.e., hydrostatic) part of the Cauchy stress. The thermodynamic pressure is an intensive state variable. Its negative value conjugates to the specific volume \cite{Callen1985}. Sometimes, mathematicians introduce pressure as a Lagrangian multiplier to enforce the incompressibility constraint. In this work, we adopt the thermodynamic definition for the pressure, and the derivation will reveal how these three concepts are related. In continuum mechanics, the Helmholtz free energy (or the strain energy) is often used as the thermodynamic potential to derive and describe the constitutive relations. Such a derivation often follows the classical Coleman-Noll approach \cite{Coleman1963,Holzapfel2000,Liu2015}. Among these relations, the thermodynamic pressure is often expressed as a function of the density and the temperature. However, this relation only remains valid for compressible materials. In the incompressible limit, the pressure-density curve becomes a vertical line of infinite slope since the density is constrained as a constant. We believe this pathological behavior of the pressure-density curve near the incompressible limit is the bane of incompressible solid solvers. A similar argument was made in \cite[p.~319]{Gurtin2009} based on a thought experiment, with the conclusion that ``it would seem unreasonable to allow a constitutive relation for an incompressible elastic body to involve the pressure." 

To remedy this degeneracy, one can derive the constitutive relations based on the Gibbs free energy \cite{Lowengrub1998}. The Gibbs free energy and the Helmholtz free energy are related by a Legendre transformation \cite{Callen1985}. For compressible materials, both the free energies are valid; for incompressible materials, the Helmholtz free energy degenerates, and the Gibbs free energy remains valid. In this section, we will derive the system of equations for a continuum mechanics model based on the Coleman-Noll approach. For simplicity, we take the referential frame of reference to be identical to the material frame of reference within this section. Then we directly have
\begin{align*}
& \frac{d\left(\cdot \right)}{d t} := \left. \frac{\partial \left( \cdot \right)}{\partial t}\right|_{\bm X} = \left. \frac{\partial \left( \cdot \right)}{\partial t}\right|_{\bm \chi} , \qquad \bm v = \hat{\bm v}.
\end{align*}
Consequently, the advective form of the balance equations and the second law of thermodynamics can be written as
\begin{align}
\label{eq:mass_balance}
& \frac{d\rho}{dt} + \rho \nabla_{\bm x} \cdot \bm v = 0, \displaybreak[2]\\
\label{eq:linear_momentum_balance}
& \rho \frac{d \bm{v} }{dt} = \nabla_{\bm x} \cdot \bm \sigma + \rho \bm b, \displaybreak[2] \\
\label{eq:angular_momentum_balance}
& \bm \sigma = \bm \sigma^T, \displaybreak[2] \\
\label{eq:internal_energy_balance}
& \rho \frac{d\iota}{dt} = \bm \sigma : \nabla_{\bm x} \bm v - \nabla_{\bm x} \cdot \bm q + \rho r, \displaybreak[2] \\
\label{eq:2nd_law_thermodynamics}
& \mathcal D := \rho \frac{ds}{dt} + \nabla_{\bm x} \cdot \left( \frac{\bm q}{\theta} \right) - \frac{\rho r}{\theta} \geq 0.
\end{align}
In the above, $\mathcal D$ represents the dissipation. The specific Gibbs free energy per unit mass is defined as 
\begin{align}
\label{eq:def_Gibbs_free_energy}
G := \iota - \theta s + \frac{p}{\rho}.
\end{align}
Taking material time derivatives on both sides of \eqref{eq:def_Gibbs_free_energy} results in
\begin{align*}
\rho \frac{dG}{dt} + \rho s \frac{d\theta}{dt} - \frac{dp}{dt} = \rho \frac{d\iota}{dt} - \rho \theta \frac{ds}{dt} - \frac{p}{\rho}\frac{d\rho}{dt}.
\end{align*}
Substituting the internal energy balance equation \eqref{eq:internal_energy_balance}, the second law of thermodyanmics \eqref{eq:2nd_law_thermodynamics}, and the mass balance equation \eqref{eq:mass_balance} into the above relation, one obtains
\begin{align}
\label{eq:gibbs_energy_balance_original}
\rho \frac{dG}{dt} = \bm \sigma : \nabla_{\bm x} \bm v - \bm q \cdot \nabla_{\bm x} \theta + p \nabla_{\bm x} \cdot \bm v - \theta \mathcal D - \rho s \frac{d\theta}{dt} + \frac{dp}{dt}.
\end{align}
To facilitate our discussion, we introduce the following notations. 
\begin{enumerate}
\item The right Cauchy-Green tensor $\bm C$ is defined as
$\bm C := \bm F^T \bm F$.
\item The deformation gradient $\bm F$ can be multiplicatively decomposed into dilatational and distortional parts \cite{Flory1961} as
\begin{align}
\label{eq:def_multiplicative_decomp_F}
\bm F = \left( J^{\frac{1}{3}} \bm I \right) \tilde{\bm F} = J^{\frac{1}{3}} \tilde{\bm F}.
\end{align} 
In the above relation, $\bm I$ represents the second-order identity tensor, $J^{\frac{1}{3}} \bm I$ represents the volume-changing (dilatational) part of the deformation, and $\tilde{\bm F} = J^{-\frac{1}{3}}\bm F$ is the volume-preserving (distortional) part of the deformation. Correspondingly, the right Cauchy-Green tensor can be decomposed as
\begin{align*}
\bm C = \left( J^{\frac{2}{3}} \bm I \right) \tilde{\bm C} = J^{\frac{2}{3}} \tilde{\bm C}.
\end{align*}
\item We can obtain the following differentiation relation
\begin{align*}
& \frac{\partial \tilde{\bm C}}{\partial \bm C} = J^{-\frac{2}{3}} \mathbb P^T,\quad \mathbb P = \mathbb I - \frac{1}{3} \bm C^{-1} \otimes \bm C,
\end{align*}
wherein $\mathbb I$ is the fourth-order identity tensor
\begin{align*}
\mathbb I_{IJKL} = \frac{1}{2} \left( \delta_{IK}\delta_{JL} + \delta_{IL} \delta_{JK} \right),
\end{align*}
and $\delta_{IJ}$ is the Kronecker delta. It is straightforward to show that $\mathbb P \mathbb P = \mathbb P$, which implies that $\mathbb P$ is a projection.
\item The Cauchy stress $\bm \sigma$ can be additively split into deviatoric and hydrostatic parts,
\begin{align*}
\bm \sigma &= \textup{dev}[\bm \sigma] + \frac{1}{3}\left(\textup{tr}[\bm \sigma] \right)\bm I.
\end{align*}
The second Piola-Kirchhoff stress $\bm S$ is defined by
$
\bm S := J \bm F^{-1} \bm \sigma \bm F^{-T},
$
and one can show that
$
\textup{dev}[\bm \sigma] = J^{-1} \bm F \left( \mathbb P : \bm S \right) \bm F^{T}.
$
\item The spatial velocity gradient can be additively split into the rate of deformation tensor $\bm d$ and the spin tensor $\bm w$ as
\begin{align*}
\nabla_{\bm x} \bm v &= \bm d + \bm w, \quad
\bm d := \frac{1}{2}\left( \nabla_{\bm x}\bm v + \nabla_{\bm x}\bm v^T\right), \quad
\bm w := \frac{1}{2}\left( \nabla_{\bm x}\bm v - \nabla_{\bm x}\bm v^T\right).
\end{align*}
Furthermore, we split $\bm d$ into deviatoric and hydrostatic parts as
\begin{align*}
\bm d &= \textup{dev}[\bm d] + \frac{1}{3} \nabla_{\bm x} \cdot \bm v \bm I.
\end{align*}
The inner product of $\bm \sigma$ and $\nabla_{\bm x} \bm v$ can be written as
\begin{align}
\label{eq:sigma_l_new_expression}
\bm \sigma : \nabla_{\bm x} \bm v = \bm \sigma : \bm d = \textup{dev}[\bm \sigma] : \textup{dev}[\bm d] + \frac{1}{3} \textup{tr}[\bm \sigma] \nabla_{\bm x} \cdot \bm v.
\end{align}
\item Algebraic manipulations can show that the time rate of $\bm C$ can be expressed in terms of $\bm d$ in the following relation 
\begin{align*}
\frac{d}{dt}\bm C = 2 \bm F^T \bm d \bm F.
\end{align*}
The time rate of $\tilde{\bm C}$ can be derived as
\begin{align*}
\frac{d}{dt} \tilde{\bm C} &= \frac{d}{dt}\left( J^{-\frac{2}{3}} \bm C \right) = J^{-\frac{2}{3}} \frac{d}{dt} \bm C - \frac{2}{3} J^{-\frac{5}{3}}\bm C\frac{d}{dt} J = 2 J^{-\frac{2}{3}} \bm F^T \textup{dev}[\bm d] \bm F.
\end{align*}
Hence, we have 
\begin{align*}
\textup{dev}[\bm d] = \frac{1}{2}J^{\frac{2}{3}} \bm F^{-T} \left(\frac{d}{dt}\tilde{\bm C} \right) \bm F^{-1},
\end{align*}
and
\begin{align}
\label{eq:dev_sigma_dev_d_new_expression}
\textup{dev}[\bm \sigma] : \textup{dev}[\bm d] = \frac{1}{2}J^{\frac{2}{3}} \bm F^{-1} \textup{dev}[\bm \sigma] \bm F^{-T} : \frac{d}{dt}\tilde{\bm C}.
\end{align}
\end{enumerate}
Using \eqref{eq:sigma_l_new_expression} and \eqref{eq:dev_sigma_dev_d_new_expression}, the relation \eqref{eq:gibbs_energy_balance_original} can be rewritten as
\begin{align}
\rho \frac{dG}{dt} &= \textup{dev}[\bm \sigma] : \textup{dev}[\bm d] +  \left(\frac{1}{3}\textup{tr}\left[ \bm \sigma \right] + p \right)\nabla_{\bm x}\cdot \bm v - \bm q \cdot \nabla_{\bm x} \theta - \rho s \frac{d\theta}{dt} + \frac{dp}{dt} - \theta \mathcal D \nonumber \displaybreak[2] \\
\label{eq:gibbs_energy_balance_version2}
&=  \frac{1}{2}J^{\frac{2}{3}} \bm F^{-1} \textup{dev}[\bm \sigma] \bm F^{-T} : \frac{d}{dt}\tilde{\bm C} +  \left(\frac{1}{3}\textup{tr}\left[ \bm \sigma \right] + p \right)\nabla_{\bm x}\cdot \bm v - \bm q \cdot \nabla_{\bm x} \theta  - \rho s \frac{d\theta}{dt} + \frac{dp}{dt} - \theta \mathcal D.
\end{align}
In the above relation, the time rate of $G$ is on the left-hand side, and the time rates of $\tilde{\bm C}$, $\theta$, and $p$ appear on the right-hand side. Invoking Truesdell's principle of equipresence \cite{Truesdell1965}, we demand that the Gibbs free energy is a function of $\tilde{\bm C}$, $p$, and $\theta$:
\begin{align*}
G = G\left( \tilde{\bm C}, p, \theta \right).
\end{align*}
Taking material time derivatives on both sides, we obtain the relation
\begin{align}
\label{eq:dG_dt_chain_rule}
\frac{dG}{dt} 
= \frac{1}{2} \tilde{\bm S} : \frac{d}{dt} \tilde{\bm C} + \frac{\partial G}{\partial p}\frac{dp}{dt} + \frac{\partial G}{\partial \theta} \frac{d\theta}{dt},
\end{align}
wherein
\begin{align*}
\tilde{\bm S} := 2 \frac{\partial G}{\partial \tilde{\bm C}}.
\end{align*}
Substituting \eqref{eq:dG_dt_chain_rule} into \eqref{eq:gibbs_energy_balance_version2} leads to 
\begin{align}
\label{eq:gibbs_energy_balance_version3}
\theta \mathcal D =& \left( \frac{1}{2}J^{\frac{2}{3}} \bm F^{-1} \textup{dev}[\bm \sigma] \bm F^{-T} - \frac{\rho}{2} \tilde{\bm S} \right) : \frac{d}{dt}\tilde{\bm C} +  \left(\frac{1}{3}\textup{tr}\left[ \bm \sigma \right] + p \right)\nabla_{\bm x}\cdot \bm v - \bm q \cdot \nabla_{\bm x} \theta \nonumber \\
& - \left( \rho s + \rho \frac{\partial G}{\partial \theta}\right) \frac{d\theta}{dt} + \left( 1 - \rho \frac{\partial G}{\partial p} \right)\frac{dp}{dt}.
\end{align}
We make the following choices for the Cauchy stress, the heat flux, the entropy, and the density.
\begin{align}
\label{eq:constitutive_dev_sigma}
\textup{dev}\left[ \bm \sigma \right] =& \rho \tilde{\bm F} \left( \mathbb P : \tilde{\bm S} \right) \tilde{\bm F}^T + 2\bar{\mu} \textup{dev}[\bm d], \displaybreak[2]\\
\label{eq:constitutive_tr_sigma}
\frac{1}{3}\textup{tr}\left[\bm \sigma \right] =& -p + \left(\frac{2}{3}\bar{\mu} + \bar{\lambda} \right) \nabla_{\bm x} \cdot \bm v, \displaybreak[2] \\
\label{eq:constitutive_heat_flux}
\bm q =& - \bar{\kappa} \nabla_{\bm x} \theta, \displaybreak[2] \\
\label{eq:constitutive_s}
s =& -\frac{\partial G}{\partial \theta}, \displaybreak[2] \\
\label{eq:constitutive_rho}
\rho =& \left( \frac{\partial G}{\partial p} \right)^{-1}.
\end{align}
In \eqref{eq:constitutive_dev_sigma}, $\bar{\mu}$ is the dynamic shear viscosity; in \eqref{eq:constitutive_tr_sigma}, $\bar{\lambda}$ is the second viscosity coefficient, and $\frac{2}{3}\bar{\mu} + \bar{\lambda}$ is the bulk viscosity; in \eqref{eq:constitutive_heat_flux}, $\bar{\kappa}$ is the termal conductivity. Combining the constitutive relations \eqref{eq:constitutive_dev_sigma} and \eqref{eq:constitutive_tr_sigma}, one obtains the Cauchy stress as
\begin{align}
\bm \sigma &= \rho \tilde{\bm F} \left( \mathbb P : \tilde{\bm S} \right) \tilde{\bm F}^T - p \bm I + 2\bar{\mu} \textup{dev}[\bm d] + \left(\frac{2}{3}\bar{\mu} + \bar{\lambda} \right) \nabla_{\bm x} \cdot \bm v \bm I \nonumber \\
\label{eq:constitutive_sigma}
&= J^{-1} \tilde{\bm F} \left( \mathbb P : 2 \frac{\partial \rho_0 G}{\partial \tilde{\bm C}} \right) \tilde{\bm F}^T - p \bm I + 2\bar{\mu} \textup{dev}[\bm d] + \left(\frac{2}{3}\bar{\mu} + \bar{\lambda} \right) \nabla_{\bm x} \cdot \bm v \bm I.
\end{align}
The first term in \eqref{eq:constitutive_sigma} represents the isochoric elastic stress \cite{Holzapfel2000}, the second term is the pressure, the third term gives the viscous shear stress, and the last term gives the bulk viscous stress. Observing that $\bm \sigma$ is symmetric, the angular momentum balance law \eqref{eq:adv_form_angular_momentum_balance} is automatically satisfied. The constitutive relation for the heat flux \eqref{eq:constitutive_heat_flux} is the Fourier's law. The constitutive relation for the entropy density $s$ \eqref{eq:constitutive_s} and the density $\rho$ \eqref{eq:constitutive_rho} coincides with the classical thermodynamic definitions \cite{Callen1985}. Invoking \eqref{eq:def_Gibbs_free_energy}, \eqref{eq:constitutive_s}, and \eqref{eq:constitutive_rho}, one can get the constitutive relation for the internal energy per unit mass as
\begin{align*}
\iota = G + s \theta - \frac{p}{\rho} = G - \frac{\partial G}{\partial \theta} \theta - \frac{\partial G}{\partial p} p.
\end{align*}
\begin{proposition}
\label{prop:dissipation_formula}
Given the constitutive relations \eqref{eq:constitutive_dev_sigma}-\eqref{eq:constitutive_rho}, the dissipation $\mathcal D$ defined in \eqref{eq:2nd_law_thermodynamics} takes the form
\begin{align}
\label{eq:dissipation_d_formula}
\mathcal D = \frac{2\bar{\mu}}{\theta} \textup{dev}[\bm d] : \textup{dev}[\bm d] + \frac{1}{\theta}\left(\frac{2}{3}\bar{\mu} + \bar{\lambda} \right) \left( \nabla_{\bm x} \cdot \bm v \right)^2 + \frac{\bar{\kappa}}{\theta} |\nabla_{\bm x} \theta|^2.
\end{align}
\end{proposition}
\begin{proof} 
The constitutive relations \eqref{eq:constitutive_s} and \eqref{eq:constitutive_rho} make the last two terms in \eqref{eq:gibbs_energy_balance_version3} vanish. Therefore, one has
\begin{align}
\label{eq:gibbs_energy_balance_version3_reduced}
\theta \mathcal D = \left( \frac{1}{2}J^{\frac{2}{3}} \bm F^{-1} \textup{dev}[\bm \sigma] \bm F^{-T} - \frac{\rho}{2} \tilde{\bm S} \right) : \frac{d}{dt}\tilde{\bm C} +  \left(\frac{1}{3}\textup{tr}\left[ \bm \sigma \right] + p \right)\nabla_{\bm x}\cdot \bm v - \bm q \cdot \nabla_{\bm x} \theta.
\end{align}
The constitutive relation for the heat flux \eqref{eq:constitutive_heat_flux} leads to
\begin{align*}
-\bm q \cdot \nabla_{\bm x} \theta = \bar{\kappa} \nabla_{\bm x} \theta \cdot \nabla_{\bm x} \theta = \bar{\kappa} |\nabla_{\bm x} \theta|^2.
\end{align*}
The constitutive relation \eqref{eq:constitutive_tr_sigma} leads to
\begin{align*}
\left(\frac{1}{3}\textup{tr}\left[ \bm \sigma \right] + p \right)\nabla_{\bm x}\cdot \bm v = \left( \frac{2}{3}\bar{\mu} + \bar{\lambda} \right) \left( \nabla_{\bm x}\cdot \bm v\right)^2.
\end{align*}
Using the constitutive relation \eqref{eq:constitutive_dev_sigma}, the first term in \eqref{eq:gibbs_energy_balance_version3_reduced} can be simplified as
\begin{align*}
& \left( \frac{1}{2}J^{\frac{2}{3}} \bm F^{-1} \textup{dev}[\bm \sigma] \bm F^{-T} - \frac{\rho}{2} \tilde{\bm S} \right) : \frac{d}{dt}\tilde{\bm C} = 2\bar{\mu} \textup{dev}[\bm d] : \textup{dev}[\bm d].
\end{align*}
In summary, one has
\begin{align*}
\theta \mathcal D = 2\bar{\mu} \textup{dev}[\bm d] : \textup{dev}[\bm d] + \left(\frac{2}{3}\bar{\mu} + \bar{\lambda} \right) \left( \nabla_{\bm x} \cdot \bm v \right)^2 + \bar{\kappa} |\nabla_{\bm x} \theta|^2,
\end{align*}
which completes the proof.
\end{proof}
The relation \eqref{eq:dissipation_d_formula} suggests that the dissipation $\mathcal D$ is guaranteed to be non-negative if the dynamic shear viscosity, the bulk viscosity, and the thermal conductivity are non-negative. 
\begin{proposition}
The constitutive relations \eqref{eq:constitutive_dev_sigma}-\eqref{eq:constitutive_rho} satisfy the principle of material frame indifference.
\end{proposition}
\begin{proof}
One only needs to verify that \eqref{eq:constitutive_dev_sigma} is material frame indifferent. Considering a proper orthogonal tensor $\bm Q$, one has $\textup{det}(\bm Q \bm F) = \textup{det}(\bm Q) \textup{det}(\bm F) = \textup{det}(\bm F) = J$. Therefore, $\widetilde{\bm Q \bm F} = \bm Q \tilde{\bm F}$. It is known that under the rigid-body motion described by $\bm Q$, $\textup{dev}[\bm d]$ transforms to $\bm Q \textup{dev}[\bm d] \bm Q^T$. Consequently, the right-hand side of \eqref{eq:constitutive_dev_sigma} transforms to
\begin{align*}
&\rho \bm Q \tilde{\bm F} \left( \mathbb P : \tilde{\bm S} \right) \tilde{\bm F}^T \bm Q^T + 2\bar{\mu} \bm Q \textup{dev}[\bm d] \bm Q^T = \bm Q \textup{dev}[\bm \sigma] \bm Q^T. \qedhere
\end{align*}
\end{proof}
\begin{proposition}
The Gibbs free energy takes the following additive decoupled form,
\begin{align}
\label{eq:gibbs_free_energy_additive_split}
G\left(\tilde{\bm C}, p, \theta \right) = G_{iso}\left(\tilde{\bm C}, \theta \right) + G_{vol}\left(p, \theta \right).
\end{align}
\end{proposition}
\begin{proof}
The density $\rho$ is independent of $\tilde{\bm C}$ since $\tilde{\bm C}$ is volume-preserving ($\textup{det}(\tilde{\bm C}) = 1$). Consequently, the constitutive relation \eqref{eq:constitutive_rho} leads to
\begin{align*}
\frac{\partial G\left(\tilde{\bm C}, p, \theta \right)}{\partial p} = \rho^{-1}(p,\theta).
\end{align*}
Integrating the above partial derivative gives \eqref{eq:gibbs_free_energy_additive_split}, where
$G_{vol}(p,\theta) = \int \rho^{-1} dp$.
\end{proof}
It is worth pointing out that the free energy adopted in the above discussion is the \textit{specific} free energy, which means that it is the energy per unit mass. It is sometimes useful to introduce the free energy per unit volume in the material configuration \cite{Gurtin2009,Holzapfel2000}, which is denoted as $G^R$. The two energies are linked by the relation $G^R = \rho_0 G$. Correspondingly, we denote $G^R_{iso} := \rho_0 G_{iso}$ and $G^R_{vol} := \rho_0 G_{vol}$. Let $H^R$ denote the Helmholtz free energy per unit volume in the material configuration. It can be obtained through a Legendre transformation of $-G^R$ with respect to $p$, namely,
\begin{align*}
H^R(\tilde{\bm C}, J, \theta) := \sup_{p}\left( -pJ+ G^R(\tilde{\bm C}, p, \theta) \right). 
\end{align*}
Invoking the additive split of the Gibbs free energy \eqref{eq:gibbs_free_energy_additive_split}, one has
\begin{align*}
H^R(\tilde{\bm C}, J, \theta) = G^R_{iso}\left(\tilde{\bm C}, \theta \right) + \sup_{p}\left( -pJ + G^R_{vol}\left(p, \theta \right) \right).
\end{align*}
This relation can be summarized as the following proposition.
\begin{proposition}
\label{prop:helmholtz_additive_split}
The Helmholtz free energy $H^R$ admits the additive decoupled form
\begin{align}
\label{eq:helmholtz_additive_split}
H^R(\tilde{\bm C}, J, \theta) = H^R_{iso}(\tilde{\bm C}, \theta) + H^R_{vol}(J, \theta),
\end{align}
wherein
\begin{align*}
H^R_{iso}(\tilde{\bm C}, \theta) =& G^R_{iso}(\tilde{\bm C}, \theta), \quad
H^R_{vol}(J, \theta) = \sup_{p}\left( -pJ + G^R_{vol}\left(p, \theta \right) \right).
\end{align*}
\end{proposition}
\begin{remark}
The additive split of the Helmholtz free energy was introduced as a postulate based on the multiplicative decomposition of the deformation gradient \cite{Flory1961}. Above, we have given a justification of the additive split \eqref{eq:helmholtz_additive_split} for hyperelastic materials. Extension of the arguments to inelastic materials is under assessment.
\end{remark}
With the energy split \eqref{eq:gibbs_free_energy_additive_split}, the constitutive relations \eqref{eq:constitutive_rho} and \eqref{eq:constitutive_sigma} can be rewritten as
\begin{align}
\label{eq:constitutive_rho_ghat}
\rho &= \rho(p,\theta) = \left( \frac{\partial G_{vol}(p,\theta)}{\partial p} \right)^{-1}, \\
\label{eq:constitutive_sigma_gtilde}
\bm \sigma &= J^{-1} \tilde{\bm F} \left( \mathbb P : 2 \frac{\partial G^R_{iso}}{\partial \tilde{\bm C}} \right) \tilde{\bm F}^T - p \bm I + 2\bar{\mu} \textup{dev}[\bm d] + \left(\frac{2}{3}\bar{\mu} + \bar{\lambda} \right) \nabla_{\bm x} \cdot \bm v \bm I.
\end{align}
Here, we introduce the isobaric thermal expansion coefficient $\alpha_p$ and the isothermal compressibility coefficient $\beta_{\theta}$ \cite{Callen1985}. They are defined as
\begin{align*}
\alpha_p := -\frac{1}{\rho}\frac{\partial \rho}{\partial \theta}, \qquad
\beta_{\theta} := \frac{1}{\rho} \frac{\partial \rho}{\partial p}.
\end{align*}
Making use of the constitutive relation \eqref{eq:constitutive_rho_ghat}, they can be expressed explicitly as
\begin{align}
\label{eq:continuum_model_def_beta}
\alpha_p = \frac{\partial^2 G_{vol}}{\partial p \partial \theta} / \frac{\partial G_{vol}}{\partial p} , \quad
\beta_{\theta}  = -\frac{\partial^2 G_{vol}}{\partial p^2} / \frac{\partial G_{vol}}{\partial p}.
\end{align}

\subsection{Examples of closed systems of equations}
\label{sec:examples_of_equations}
Before proceeding further, we notice that the time derivative of $\rho$ can be expanded as
\begin{align*}
\left. \frac{\partial \rho}{\partial t} \right|_{\bm \chi} = -\alpha_{p} \rho \left. \frac{\partial \theta}{\partial t} \right|_{\bm \chi} + \beta_{\theta} \rho \left. \frac{\partial p}{\partial t} \right|_{\bm \chi}.
\end{align*}
Hence, the governing equations in the ALE frame of reference become
\begin{align}
\label{eq:examples_governing_eqn_disp_velo}
& \bm 0 = \left. \frac{\partial \bm u}{\partial t} \right|_{\bm \chi} + \left( \bm v  - \hat{\bm v} \right) \cdot \nabla_{\bm x} \bm u - \bm v, \displaybreak[2] \\
\label{eq:examples_governing_eqn_mass}
& 0 = -\alpha_{p} \rho \left. \frac{\partial \theta}{\partial t} \right|_{\bm \chi} + \beta_{\theta} \rho \left. \frac{\partial p}{\partial t} \right|_{\bm \chi} + \left(\bm v - \hat{\bm v} \right)\cdot \nabla_{\bm x} \rho + \rho \nabla_{\bm x} \cdot \bm v , \displaybreak[2] \\
\label{eq:examples_governing_eqn_momentum}
& \bm 0 = \left. \rho \frac{\partial \bm v}{\partial t} \right|_{\bm \chi} + \rho \left( \nabla_{\bm x} \bm v \right) \left(\bm v - \hat{\bm v} \right) - \nabla_{\bm x} \cdot \bm \sigma - \rho \bm b, \displaybreak[2] \\
\label{eq:examples_governing_eqn_energy}
& 0 = \left. \rho \frac{\partial \iota}{\partial t} \right|_{\bm \chi} + \rho \left( \bm v - \hat{\bm v} \right) \cdot \nabla_{\bm x} \iota - \bm \sigma : \nabla_{\bm x} \bm v + \nabla_{\bm x} \cdot \bm q - \rho r .
\end{align}

\paragraph{Ideal gas model}
The Gibbs free energy for the ideal gas is
\begin{align*}
G^{\textup{pg}}(\tilde{\bm C}, p, \theta) = \mathit R \theta \ln \left( \frac{p\theta_{\textup{ref}}}{\theta p_{\textup{ref}}} \right) - \mathit C_v \theta \ln \left( \frac{\theta}{\theta_{\textup{ref}}} \right) + \left(\mathit C_v + \mathit R \right)\theta,
\end{align*}
wherein $\mathit R$ is the specific gas constant, $\mathit C_v$ is the specific heat at constant volume, $p_{\textup{ref}}$ and $\theta_{\textup{ref}}$ are the reference values of the pressure and temperature. With this choice, the constitutive relations are
\begin{align*}
\bm \sigma &= -p \bm I + \bar{\mu} \left(\nabla_{\bm x} \bm v + \nabla_{\bm x} \bm v^T \right) + \bar{\lambda} \nabla_{\bm x} \cdot \bm v \bm I, \quad \rho = \frac{p}{\mathit R \theta}, \\
s &= -R\ln \left( \frac{p\theta_{\textup{ref}}}{\theta p_{\textup{ref}}} \right) + \mathit C_v \ln \left( \frac{\theta}{\theta_{\textup{ref}}} \right), \quad
\iota = \mathit C_v \theta.
\end{align*}
The balance equations \eqref{eq:examples_governing_eqn_mass}-\eqref{eq:examples_governing_eqn_energy} together with the above constitutive equations constitute the pressure primitive variable formulation for the compressible Navier-Stokes equations \cite{Hauke1998}. It was shown that the set of pressure primitive variables is among the sets of variables that are well-behaved for both compressible and incompressible flows \cite{Hauke1994,Hauke1998,Hughes1986,Liu2013}; on the contrary, any set of variables involving density (e.g. conservation variables or density primitive variables) becomes ill-defined in the incompressible limit \cite{Hauke1994}. That observation, in part, justifies our motivation of deriving a continuum mechanics model based on the Gibbs free energy.

\paragraph{Incompressible viscous flow} We introduce the Gibbs free energy for incompressible viscous fluid flow as
\begin{align*}
G^{if}(\tilde{\bm C}, p, \theta) = G_{iso}^{if}(\theta) + \frac{p}{\rho_0}.
\end{align*}
With this choice, the constitutive relations are
\begin{align*}
\bm \sigma &= -p \bm I + \bar{\mu} \left(\nabla_{\bm x} \bm v + \nabla_{\bm x} \bm v^T \right) + \bar{\lambda} \nabla_{\bm x} \cdot \bm v \bm I, \quad
\rho = \rho_0, \displaybreak[2] \\
s &= - \frac{\partial G_{iso}^{if}}{\partial \theta}, \quad
\iota = G_{iso}^{if} - \theta \frac{\partial G_{iso}^{if}}{\partial \theta}.
\end{align*}

\paragraph{Compressible hyperelastic model} The Gibbs free energy for compressible hyperelastic materials takes the following general form.
\begin{align*}
G^{ch}(\tilde{\bm C}, p, \theta) = G^{ch}_{iso}( \tilde{\bm C},\theta ) + G_{vol}^{ch}(p,\theta).
\end{align*}
If we ignore the viscous effect, the constitutive relations are 
\begin{align*}
& \bm \sigma = J^{-1} \tilde{\bm F} \left( \mathbb P : 2 \rho_0 \frac{\partial G^{ch}_{iso}}{\partial \tilde{\bm C}} \right) \tilde{\bm F}^T - p \bm I, \quad
\rho = \left(\frac{\partial G_{vol}^{ch}}{\partial p} \right)^{-1}, \displaybreak[2] \\
& s = - \frac{\partial G^{ch}}{\partial \theta}, \quad \iota = G^{ch} - \theta \frac{\partial G^{ch}}{\partial \theta} - p \frac{\partial G_{vol}^{ch}}{\partial p}.
\end{align*}

\paragraph{Incompressible hyperelastic model}
Noticing that incompressibility implies $\tilde{\bm C} = \bm C$, the Gibbs free energy for incompressible hyperelastic materials takes the following form.
\begin{align*}
G^{ih}(\tilde{\bm C}, p, \theta) = G^{ih}(\bm C, p, \theta) = G^{ih}_{iso}(\bm C, \theta) + \frac{p}{\rho_0}.
\end{align*}
Ignoring the viscous effect, the constitutive relations can be written as
\begin{align*}
& \bm \sigma = J^{-1} \tilde{\bm F} \left( \mathbb P : 2 \rho_0 \frac{\partial G^{ih}_{iso}}{\partial \bm C} \right) \tilde{\bm F}^T - p \bm I, \quad
\rho = \rho_0, \displaybreak[2] \\
& s = - \frac{\partial G^{ih}_{iso}}{\partial \theta}, \quad
\iota = G^{ih}_{iso} - \theta \frac{\partial G^{ih}_{iso}}{\partial \theta}.
\end{align*}

\begin{remark}
It is worth pointing out that the displacement-velocity relation \eqref{eq:adv_form_disp_velo_relation} is not the unique choice for describing kinematics. The deformation gradient transport relation \cite[p. 24]{Scovazzi2007},
\begin{align*}
\frac{d}{dt} \bm F = \nabla_{\bm x} \bm v \bm F,
\end{align*}
can be utilized to calculate the strain and stress as well \cite{Duddu2012,Scovazzi2017}. In fact, this strategy is expected to give second-order spatial accuracy in the calculation of the strain and the stress with linear elements. A trade-off is that $n_d^2$ additional differential equations for $\bm F$ need to be solved. Based on the polyconvexity hypothesis, it seems natural to introduce kinematic relations for $\bm F$, $J \bm F^{-T}$, and $J$ \cite{Bonet2015,Bonet2015a,Gil2014}. In that approach, the kinematic equations involve $2n_d^2+1$ degrees of freedom. In this work, we choose to solve the simple displacement-velocity equation \eqref{eq:adv_form_disp_velo_relation}. Its simple structure leads to an additional benefit in the design of the algorithm.
\end{remark}

\subsection{Legendre transformation of volumetric energies}
Here, we want to point out that, in the literature, there exist derivations of constitutive relations based on a different Gibbs free energy \cite{Rajagopal2013,Surana2013}. In those works, the Gibbs free energy is taken as a function of the stress, which implicitly requires the Helmholtz free energy to be convex with respect to the strain. Convexity is a strong requirement and rules out many important nonlinear materials. This is a key difference with the present work, which does not impose the requirement. Most hyperelastic materials are postulated to be polyconvex \cite{Ball1976,Marsden1994}. A direct consequence of polyconvexity is that, with the isochoric-volumetric split shown in Proposition \ref{prop:helmholtz_additive_split}, the volumetric energy is convex. Thus, it is indeed legitimate to perform a Legendre transformation on the volumetric part of the energy. The derivations in \cite{Rajagopal2013,Surana2013} correspond to the Hellinger-Reissner variational principle; while our derivation corresponds to the Herrmann variational principle \cite{Herrmann1965,Key1969}. In the following, we present a few examples of the Lengendre transformations for several classical volumetric energies, and we use $\kappa$ to denote the bulk modulus.

\paragraph{Quadratic volumetric energy}
The first example is the quadratic energy
\begin{align}
\label{eq:psi_vol_quadratic}
H_{vol}^R(J) = \frac{\kappa}{2}\left(J - 1\right)^2.
\end{align}
The corresponding specific energy is
\begin{align*}
H_{vol}(v) :=\frac{1}{\rho_0} H_{vol}^R(\rho_0 v)= \frac{\kappa}{2\rho_0} \left( \rho_0 v - 1 \right)^2.
\end{align*}
The conjugate function of the above $H_{vol}(v)$ is
\begin{align*}
G_{vol}(p) := \sup_{v} \left( pv + H_{vol}(v) \right) = \frac{p}{\rho_0} - \frac{p^2}{2\kappa\rho_0}. 
\end{align*}
According to \eqref{eq:constitutive_rho_ghat} and \eqref{eq:continuum_model_def_beta}, we have
\begin{align*}
\rho = \frac{\rho_0}{1 - \frac{p}{\kappa}}, \qquad
\beta_{\theta} = \frac{1}{\kappa - p}.
\end{align*}
%

\paragraph{ST91 volumetric energy}
The quadratic volumetric energy \eqref{eq:psi_vol_quadratic} has been under criticism because it approaches a finite value as the volume goes to zero. This may lead to numerical instabilities \cite{Simo1982}. To circumvent this issue, the following volumetric energy was proposed in \cite{Simo1991} and is now widely used.
\begin{align}
\label{eq:psi_vol_ST91}
H_{vol}^R(J) = \frac{\kappa}{4}\left( J^2 - 1 - 2 \ln(J) \right).
\end{align}
The specific energy is
\begin{align*}
H_{vol}(v) := \frac{1}{\rho_0}H_{vol}^R(\rho_0 v) = \frac{\kappa}{4\rho_0}\left( \rho_0^2 v^2 - 1 - 2 \ln(\rho_0 v) \right).
\end{align*}
The conjugate function to $H_{vol}(v)$ is
\begin{align*}
G_{vol}(p) := \sup_{v} \left( pv + H_{vol}(v) \right) = \frac{-p^2 + p \sqrt{p^2+\kappa^2}}{2\kappa \rho_0} - \frac{\kappa}{2\rho_0}\ln \left( \frac{\sqrt{p^2+\kappa^2}-p}{\kappa} \right).
\end{align*}
The Taylor expansion of the above $G_{vol}$ is 
\begin{align*}
G_{vol}(p) = \frac{p}{\rho_0} - \frac{p^2}{2\rho_0 \kappa} + \frac{p^3}{6\rho_0 \kappa^2} - \frac{1}{40}\frac{p^5}{\rho_0 \kappa^4} + \mathcal O(\frac{1}{\kappa^5}).
\end{align*}
Clearly, the $G_{vol}(p)$ for the ST91 volumetric energy \eqref{eq:psi_vol_ST91} can be viewed as a high-order modification of the one associated with the quadratic energy. From the formula of $G_{vol}(p)$, one obtains
\begin{align*}
\rho = \frac{\rho_0}{\kappa} \left( \sqrt{p^2 + \kappa^2} + p \right), \qquad
\beta_{\theta} =  \frac{1}{\sqrt{p^2+\kappa^2}}.
\end{align*}

\paragraph{M94 volumetric energy}
The volumetric free energy proposed in \cite{Miehe1994} and its corresponding specific free energy are
\begin{align}
\label{eq:psi_vol_M94}
& H^R_{vol}(J) = \kappa \left(J - \ln(J) -1 \right), \\
& H_{vol}(v) := \frac{1}{\rho_0} H^{R}_{vol}(\rho_0 v) = \frac{\kappa}{\rho_0} \left(\rho_0 v - \ln(\rho_0 v) -1 \right). \nonumber
\end{align}
The conjugate function to $H_{vol}(v)$ is
\begin{align*}
G_{vol}(p) := \sup_{v} \left( pv + H_{vol}(v) \right) = -\frac{\kappa}{\rho_0} \ln \left(\frac{\kappa}{p+\kappa} \right).
\end{align*}
Its Taylor expansion is
\begin{align*}
G_{vol}(p) = \frac{p}{\rho_0} - \frac{p^2}{2 \rho_0 \kappa} + \frac{p^3}{3\rho_0\kappa^2} - \frac{p^4}{4\rho_0 \kappa^3} + \frac{p^5}{5\rho_0 \kappa^4} + \mathcal{O}(\frac{1}{\kappa^5}).
\end{align*}
Based on the above formula, one arrives at
\begin{align*}
\rho = \rho_0 \left( 1 + \frac{p}{\kappa} \right), \qquad
\beta_{\theta} = \frac{1}{p+\kappa}.
\end{align*}
Here, the constitutive relation for the density is the linear barotropic relation, which is usually used to describe small density variations near a reference value \cite{Greenshields2005}.

\paragraph{L94 volumetric energy}
The volumetric free energy proposed in \cite{Liu1994} and the corresponding specific free energy are
\begin{align}
\label{eq:psi_vol_l94}
& H_{vol}^{R}(J) = \kappa \left( J \ln(J) - J + 1 \right), \\
& H_{vol}(v) := \frac{1}{\rho_0} H^{R}_{vol}(\rho_0 v) = \frac{\kappa}{\rho_0} \left( \rho_0 v \ln(\rho_0 v) - \rho_0 v + 1 \right).
\end{align}
The conjugate function to $H_{vol}(v)$ is
\begin{align*}
G_{vol}(p) = \frac{\kappa}{\rho_0} \left( 1 - e^{-\frac{p}{\kappa}} \right).
\end{align*}
Its Taylor expansion is
\begin{align*}
G_{vol}(p) = \frac{p}{\rho_0} - \frac{p^2}{2\kappa \rho_0} + \frac{p^3}{6\kappa^2 \rho_0} - \frac{p^4}{24 \kappa^3 \rho_0} + \frac{p^5}{120 \rho_0 \kappa^4} + \mathcal{O}\left(\frac{1}{\kappa^5} \right),
\end{align*}
and
\begin{align*}
\rho = \rho_0 e^{\frac{p}{\kappa}}, \qquad \beta_{\theta} = \frac{1}{\kappa}.
\end{align*}
Interestingly, this volumetric free energy gives a constant isothermal compressibility coefficient.
\begin{remark}
For compressible materials, one can derive pressure by taking derivative of the volumetric energy with respect to $J$ as
\begin{align}
\label{eq:legendre_sec_p_J_eos}
p = -\frac{dH^R_{vol}}{dJ}.
\end{align}
It can be verified that \eqref{eq:legendre_sec_p_J_eos} is compatible with the constitutive equation \eqref{eq:constitutive_rho_ghat}. At this stage, one can clearly see that the equation \eqref{eq:legendre_sec_p_J_eos} is an equation of state, just like $p=\rho R \theta$ for the ideal gas model. In CFD, it is very rare to see one put $p=\rho R \theta$ into a weak form and solve it by finite element or finite volume methods. By analogy, we feel it is worth raising a question on the validity of using \eqref{eq:legendre_sec_p_J_eos} in the mixed formulation for finite elasticity, although such an approach is regarded as well-established \cite{Saracibar2006,Holzapfel2000,Klaas1999,Maniatty2002}.
\end{remark}

\begin{remark}
\label{remark:scovazzi_2016_discussion}
In \cite{Scovazzi2016}, it has been observed that, even for small-strain elasticity, it is not advisable to use $0=p/\kappa + \nabla_{\bm x} \cdot \bm u$ in the VMS formulation for transient analysis. This fact confirms our doubts about using \eqref{eq:legendre_sec_p_J_eos} in the VMS formulation for finite elastodynamics.
\end{remark}

\begin{remark}
In the incompressible limit, the bulk modulus approaches infinity, and we have
$G_{vol}(p) \rightarrow p/\rho_0$, $\rho \rightarrow \rho_0$, $\beta_{\theta} \rightarrow 0$. Therefore, the above constitutive relations are well-defined in both compressible and incompressible regimes. In contrast, the constitutive relation \eqref{eq:legendre_sec_p_J_eos} based on the Helmholtz free energy $H^R_{vol}$ will blow up in the incompressible limit.
\end{remark}

\begin{remark}
It is interesting to notice that, although the volumetric energies $H^R_{vol}$ are different, their conjugate counterparts $G_{vol}$ are very similar. The first two terms of their Taylor expansions are identical. Notice that, besides the convexity condition, physical intuition also suggests that the volumetric energy $H^R_{vol}$ achieve its minimum value at $J=1$ and blows up to infinity as $J\rightarrow 0$ and $J\rightarrow \infty$. We feel that these conditions may imply some mathematical properties for $G_{vol}$. These properties may help design constitutive relations directly based on $G_{vol}$. We believe future work in this area will be useful both theoretically and practically.
\end{remark}

\section{Variational Multiscale Analysis}
\label{sec:VMS}
The variational multiscale method was introduced as a general framework for subgrid-scale modeling in computational mechanics \cite{Hughes1995,Hughes1998}. As a generalization of the stabilized methods, it improves the stability bound for singularly perturbed problems and overcomes the inf-sup condition for saddle-point problems. Interested readers are referred to \cite{Ahmed2017,Codina2017,Hughes2004} for comprehensive reviews. In this section, we invoke the residual-based VMS method \cite{Bazilevs2007a,Oberai2016} to construct a formulation for the continuum problem derived in Section \ref{sec:continuum_mechanics}. Here, and in what follows, we restrict our discussion to the isothermal condition\footnote{Strictly speaking, the isothermal condition is another constraint condition in thermodynamics. The Gibbs free energy degenerates since the relation \eqref{eq:constitutive_s} becomes invalid for a fixed temperature. One may choose the enthalpy as the  thermodynamic potential to derive a complete theory for an isothermal system \cite{Callen1985}. However, it can be shown that that the mechanical part of that system is identical to \eqref{eq:isothermal_governing_disp_velo}-\eqref{eq:isothermal_governing_eqn_momentum}. Hence we do not provide that derivation for a tautological system in this work.}. The system of equations \eqref{eq:examples_governing_eqn_disp_velo}-\eqref{eq:examples_governing_eqn_energy} \footnote{In this work, we choose to discuss the VMS formulation based on the advective form \eqref{eq:adv_form_mass_balance}-\eqref{eq:adv_form_disp_velo_relation} to simplify the derivation. However, for some cases, it is convenient to start with a conservative form \cite{Bazilevs2007a,Zeng2016}.} are simplified as
\begin{align}
\label{eq:isothermal_governing_disp_velo}
& \bm 0 = \left. \frac{\partial \bm u}{\partial t} \right|_{\bm \chi} + \left( \nabla_{\bm x} \bm u \right) \left( \bm v  - \hat{\bm v} \right) - \bm v, \\
\label{eq:isothermal_governing_eqn_mass}
& 0 = \beta_{\theta} \left. \frac{\partial p}{\partial t} \right|_{\bm \chi} + \beta_{\theta} \left(\bm v - \hat{\bm v} \right)\cdot \nabla_{\bm x} p + \nabla_{\bm x} \cdot \bm v ,\\
\label{eq:isothermal_governing_eqn_momentum}
& \bm 0 = \left. \rho \frac{\partial \bm v}{\partial t} \right|_{\bm \chi} + \rho \left( \nabla_{\bm x} \bm v \right) \left(\bm v - \hat{\bm v} \right) - \nabla_{\bm x} \cdot \bm \sigma - \rho \bm b.
\end{align}
The constitutive relations for the density $\rho$ and the isothermal compressibility $\beta_{\theta}$ can be simplified as univariate functions of the pressure:
\begin{align}
\label{eq:isothermal_constitutive_rho_ghat}
\rho &= \rho(p) = \left( \frac{d G_{vol}(p)}{d p} \right)^{-1}, \quad
\beta_{\theta} = \beta_{\theta}(p) = - \frac{d^2 G_{vol}(p)}{d p^2} / \frac{d G_{vol}(p)}{d p}.
\end{align}
In this section, we consider the strong-form problem endowed with proper initial conditions and periodic boundary conditions. To simplify the notation, we choose $\bsfV$ to denote both the trial and the test function spaces, which are assumed to be identical in this section. Let $\left(\cdot , \cdot \right)_{\Omega_{\bm x}}$ denote the $\mathcal L^2$ inner product over the domain $\Omega_{\bm x}$. The variational formulation for the equations \eqref{eq:isothermal_governing_disp_velo}-\eqref{eq:isothermal_governing_eqn_momentum} can be stated as follows. Find $\bsfy = \left\lbrace \bm u, p, \bm v \right\rbrace^T \in \bsfV$ such that for $\forall \hspace{1mm} \bsfw = \left\lbrace \bm w_{\bm u}, w_{p}, \bm w_{\bm v} \right\rbrace^T \in \bsfV$,
\begin{align}
\label{eq:vms_original_weak_form_problem}
&\bsfB\left( \bsfw, \bsfy \right) = \bsfF\left(\bsfw \right), \displaybreak[2] \\
&\bsfB\left( \bsfw, \bsfy \right) = \left( \bm w_{\bm u}, \left. \frac{\partial \bm u}{\partial t} \right|_{\bm \chi} + \left( \nabla_{\bm x} \bm u \right) \left( \bm v  - \hat{\bm v} \right) - \bm v \right)_{\Omega_{\bm x}}  + \left(w_{p}, \beta_{\theta} \left. \frac{\partial p}{\partial t} \right|_{\bm \chi} + \beta_{\theta} \left(\bm v - \hat{\bm v} \right)\cdot \nabla_{\bm x} p + \nabla_{\bm x} \cdot \bm v \right)_{\Omega_{\bm x}} \displaybreak[2] \nonumber \\
& \hspace{1.6cm} + \left(\bm w_{\bm v},  \left. \rho \frac{\partial \bm v}{\partial t} \right|_{\bm \chi} + \rho \left( \nabla_{\bm x} \bm v \right) \left(\bm v - \hat{\bm v} \right) \right)_{\Omega_{\bm x}} + \left( \nabla_{\bm x} \bm w_{\bm v}, \bm \sigma \right)_{\Omega_{\bm x}},  \displaybreak[2]  \nonumber \\
& \bsfF\left( \bsfw \right) = \left(\bm w_{\bm v}, \rho \bm b \right)_{\Omega_{\bm x}}. \nonumber
\end{align}
Now we introduce a projection operator $\bsfP : \bsfV \rightarrow \overline{\bsfV}$, wherein $\overline{\bsfV}$ is a computable finite-dimensional subspace of $\bsfV$. With the aid of the projection operator, we have a well-defined direct-sum decomposition of the function space $\bsfV$ as
\begin{align*}
& \bsfV = \overline{\bsfV} \oplus \bsfV^{\prime}, \quad \overline{\bsfV} = \bsfP \bsfV, \quad \bsfV^{\prime} = \left( \bsfI - \bsfP \right) \bsfV,
\end{align*}
where $\bsfI$ is the identity operator. Here, $\bsfV^{\prime}$ represents the unresolved fine  scales.  With this space decomposition, we can decompose the trial solution $\bsfy$ and the test function $\bsfw$ as
\begin{align*}
& \bsfy = \bar{\bsfy} + \bsfy^{\prime}, \quad \bar{\bsfy} = \bsfP \bsfy, \quad \bsfy^{\prime} = \left( \bsfI - \bsfP \right) \bsfy, \displaybreak[2] \\
& \bsfw = \bar{\bsfw} + \bsfw^{\prime}, \quad \bar{\bsfw} = \bsfP \bsfw, \quad \bsfw^{\prime} = \left( \bsfI - \bsfP \right) \bsfw.
\end{align*}
With the decomposition of $\bsfw$ and by virtue of the linear dependency of the variational formulation $\bsfB(\bsfw, \bsfy)$ in $\bsfw$, we can decompose the original variational formulation into a coupled system as
\begin{align}
\label{eq:VMS_coarse_eqn}
& \bsfB\left(\bar{\bsfw}, \bar{\bsfy} + \bsfy^{\prime} \right) = \bsfF\left(\bar{\bsfw}\right), \\
\label{eq:VMS_fine_eqn}
& \bsfB\left(\bsfw^{\prime}, \bar{\bsfy} + \bsfy^{\prime} \right) = \bsfF\left(\bsfw^{\prime}\right).
\end{align}
The above two equations are usually referred to as the coarse-scale and the fine-scale equations \cite{Bazilevs2007a}. We assume that $\bsfB$ is Fr\'echet differentiable with respect to $\bsfy \in \bsfV$ up to the $n$-th derivative. Using the Taylor's Formula in the Banach space \cite{Arbogast2008}, the left-hand side of \eqref{eq:VMS_fine_eqn} can be expanded as
\begin{align*}
\bsfB\left(\bsfw^{\prime}, \bar{\bsfy} + \bsfy^{\prime}\right) = \bsfB\left(\bsfw^{\prime}, \bar{\bsfy}\right) + D_{\bsfy}\bsfB\left(\bsfw^{\prime}, \bar{\bsfy} \right)[\bsfy^{\prime}] + \cdots + \frac{1}{n!}D^n_{\bsfy}\left(\bsfw^{\prime}, \bar{\bsfy} \right)[\underbrace{\bsfy^{\prime},\cdots , \bsfy^{\prime}}_\text{n copies}] + o\left( \|\bsfy^{\prime}\|_{\bsfV}^n \right).
\end{align*}
In the above, $D_{\bsfy}^k\bsfB$ represents the $k$-th derivative of $\bsfB$ in terms of the second argument $\bsfy$. It is a $k$-linear functional on $\bsfV^{\prime} \times \cdots \times \bsfV^{\prime}$.
Moving $\bsfB\left(\bsfw^{\prime}, \bar{\bsfy}\right)$ to the right-hand side, we have
\begin{align}
\label{eq:VMS_fine_eqn_taylor}
D_{\bsfy}\bsfB\left(\bsfw^{\prime}, \bar{\bsfy} \right)[\bsfy^{\prime}] + \cdots + \frac{1}{n!}D^n_{\bsfy}\left(\bsfw^{\prime}, \bar{\bsfy} \right)[\bsfy^{\prime},\cdots , \bsfy^{\prime}] + o\left( \|\bsfy^{\prime}\|_{\bsfV}^n \right) = \bsfF\left(\bsfw^{\prime} \right) - \bsfB\left( \bsfw^{\prime}, \bar{\bsfy} \right).
\end{align}
One can represent the right-hand side of \eqref{eq:VMS_fine_eqn_taylor} as $\textbf{Res}\left(\bar{\bsfy}\right)[\bsfw^{\prime}]$, wherein $\textbf{Res}\left(\bar{\bsfy}\right)$ is, formally, the residual of the coarse-scale lifted to $\bsfV^{\prime *}$, the dual of $\bsfV^{\prime}$. Based on \eqref{eq:VMS_fine_eqn_taylor}, one may observe that $\bsfy^{\prime}$ depends on $\textbf{Res}\left(\bar{\bsfy}\right)$ and $\bar{\bsfy}$. Hence, $\bsfy^{\prime}$ can be represented by an abstract mapping $\mathscr F^{\prime}$,
\begin{align}
\label{eq:formal_representation_y_prime}
\bsfy^{\prime} = \mathscr F^{\prime}\left(\bar{\bsfy}, \textbf{Res}\left(\bar{\bsfy} \right) \right).
\end{align}
Inserting \eqref{eq:formal_representation_y_prime} into \eqref{eq:VMS_coarse_eqn}, a closed, finite-dimensional system can be obtained for $\bar{\bsfy}$,
\begin{align}
\label{eq:formal_representation_y_bar}
\bsfB\left( \bar{\bsfw} , \bar{\bsfy} +   \mathscr F^{\prime}\left(\bar{\bsfy}, \textbf{Res}\left(\bar{\bsfy} \right) \right)\right) = \bsfF\left( \bsfw^{\prime} \right).
\end{align}
Given the analytic form of the mapping $\mathscr F^{\prime}$, one may obtain $\bar{\bsfy}$ from \eqref{eq:formal_representation_y_bar} and $\bsfy^{\prime}$ from \eqref{eq:formal_representation_y_prime}. The resulting $\bsfy = \bar{\bsfy} + \bsfy^{\prime}$ is the exact solution of the original problem \eqref{eq:vms_original_weak_form_problem}. However, obtaining an analytic form for $\mathscr F^{\prime}$ is as hard as solving the original problem analytically, if not harder. A practical approach is to systematically design an approximated mapping $\tilde{\mathscr F}^{\prime}$. Replacing $\mathscr F^{\prime}$ in \eqref{eq:formal_representation_y_prime}-\eqref{eq:formal_representation_y_bar} by the approximated mapping, one may obtain a suite of computable formulations for the fine- and coarse-scale components.

Similar to the residual-based VMS modeling approach for turbulence \cite{Bazilevs2007a}, we introduce a perturbation series to represent $\bsfy^{\prime}$ and derive a detailed pathway to construct $\tilde{\mathscr F}^{\prime}$. The difference between our approach and the one adopted in \cite{Bazilevs2007a} is that, in addition to the approximation of the fine-scale Green's operator and the truncation of the perturbation series, we introduce one additional approximation procedure, i.e. the truncation of the Taylor expansion formula in \eqref{eq:VMS_fine_eqn_taylor}. This additional step is due to the general nonlinear term that may appear in finite elasticity. Since the derivation of the fine-scale approximation goes deeper into functional analysis, we give the detailed derivation in \ref{appd:VMS_derivation_perturbation_series}. In our model, the fine-scale component is approximated as
\begin{align}
\label{eq:VMS_algebraic_static_y_prime}
\bsfy^{\prime} \approx \tilde{\mathscr F}^{\prime}\left(\bar{\bsfy}, \textbf{Res}\left(\bar{\bsfy} \right) \right) = - \bm \tau \textbf{Res}\left( \bar{\bsfy} \right),
\end{align}
wherein
\begin{align*}
& \bm \tau = 
\begin{bmatrix}
\bm\tau_K & \bm 0 & \bm 0 \\
\bm 0 & \tau_C & \bm 0 \\
\bm 0 & \bm 0 & \bm \tau_M
\end{bmatrix}, \qquad \textbf{Res}\left(\bar{\bsfy} \right) = 
\begin{Bmatrix}
\bm r_K(\bar{\bsfy}) \\
r_C(\bar{\bsfy}) \\
\bm r_M(\bar{\bsfy})
\end{Bmatrix}, \displaybreak[2] \\
& \bm r_K(\bar{\bsfy}) = \left. \frac{\partial \bar{\bm u}}{\partial t} \right|_{\bm \chi} + \left( \nabla_{\bm x} \bar{\bm u} \right) \left( \bar{\bm v}  - \hat{\bm v} \right) - \bar{\bm v},  \displaybreak[2]  \\
& r_C(\bar{\bsfy}) = \bar{\beta}_{\theta} \left. \frac{\partial \bar{p}}{\partial t} \right|_{\bm \chi} + \bar{\beta}_{\theta} \left(\bar{\bm v} - \hat{\bm v} \right)\cdot \nabla_{\bm x} \bar{p} + \nabla_{\bm x} \cdot \bar{\bm v},  \displaybreak[2]  \\
& \bm r_M(\bar{\bsfy}) = \left. \bar{\rho} \frac{\partial \bar{\bm v}}{\partial t} \right|_{\bm \chi} + \bar{\rho} \left( \nabla_{\bm x} \bar{\bm v} \right) \left(\bar{\bm v}  - \hat{\bm v} \right) - \nabla_{\bm x} \cdot \bar{\bm \sigma} - \bar{\rho} \bm b.
\end{align*}
In the above, the choice $\bm \tau = \textup{diag}(\bm \tau_K, \tau_C, \bm \tau_M)$ implies that the fine-scales are postulated to be decoupled. The precise formulas for the stabilization parameters $\tau_K$, $\tau_C$, and $\tau_M$ depend on the specific problem considered. For simple linear problems, $\bm \tau$ can be computed as a local mean-value of the fine-scale Green's function \cite{Hughes1995,Hughes2007}; sometimes, error estimates provide a guidance for the design of $\bm \tau$ \cite{Franca1992,Hughes1988}; for complex problems, scaling arguments are usually made for the design of $\bm \tau$ \cite{Tezduyar2000}. The detailed formula of $\bm \tau$ for solid and fluid dynamics will be given in the subsequent sections. With \eqref{eq:VMS_algebraic_static_y_prime}, we can complete our VMS formulation as follows,
\begin{align}
\label{eq:VMS_final_abstract_framework}
\bsfB \left( \bar{\bsfw} , \bar{\bsfy} - \bm \tau \textbf{Res}\left( \bar{\bsfy} \right) \right) &= \bsfF(\bar{\bsfw}).
\end{align}
This formulation provides a basis for our finite element formulations in the subsequent sections.

\section{Formulation for solid dynamics}
\label{sec:numerical_solid}
In this section, we restrict our discussion to hyper-elastodynamics. Within the general VMS framework developed in Section \ref{sec:VMS}, the problem is spatially discretized using the VMS formulation. The generalized-$\alpha$ method is utilized for temporal discretization. As is shown in \cite{Scovazzi2016}, a block decomposition of the tangent matrix reveals that the problem can be solved in a segregated manner without losing consistency in the nonlinear solver. Lastly, we discuss the choice of the stabilization parameters.

\subsection{Initial-boundary value problem}
We consider the hyper-elastodynamic problem written in the Lagrangian reference frame,
\begin{align}
\label{eq:stabilized_solids_strong_form_kinematic}
& \bm 0 = \frac{d\bm u}{dt} - \bm v, && \mbox{ in } \Omega_{\bm x}, \displaybreak[2] \\
\label{eq:stabilized_solids_strong_form_pressure}
& 0 = \beta_{\theta}(p) \frac{dp}{dt} + \nabla_{\bm x} \cdot \bm v && \mbox{ in } \Omega_{\bm x},  \displaybreak[2] \\
\label{eq:stabilized_solids_strong_form_momentum}
& \bm 0 = \rho(p) \frac{d\bm v}{dt} - \nabla_{\bm x} \cdot \bm \sigma_{dev} + \nabla_{\bm x} p - \rho(p) \bm b, && \mbox{ in } \Omega_{\bm x}.
\end{align}
The time interval of interest is denoted as $(0,T)$, with $T>0$. The boundary $\Gamma_{\bm x} = \partial \Omega_{\bm x}$ can be partitioned into two non-overlapping subdivisions:
$
\Gamma_{\bm x} = \Gamma_{\bm x}^g \cup \Gamma_{\bm x}^{h},
$ 
wherein, $\Gamma_{\bm x}^g$ represents the Dirichlet part of the boundary, and $\Gamma_{\bm x}^h$ represents the Neumann part of the boundary. Boundary conditions for this problem are imposed as
\begin{align}
\label{eq:stabilized_solids_strong_form_dirichlet_u}
& \bm u = \bm g, && \mbox{ on } \Gamma_{\bm x}^{g}, \displaybreak[2] \\
\label{eq:stabilized_solids_strong_form_dirichlet_v}
& \bm v = \frac{d\bm g}{dt}, && \mbox{ on } \Gamma_{\bm x}^{g}, \displaybreak[2] \\
& (\bm \sigma_{dev} - p\bm I) \bm n = \bm h, && \mbox{ on } \Gamma_{\bm x}^{h}.
\end{align}
Given the initial data $\bm u_0$, $p_0$, and $\bm v_0$, the initial conditions for the strong-form problem \eqref{eq:stabilized_solids_strong_form_kinematic}-\eqref{eq:stabilized_solids_strong_form_momentum} can be stated as
\begin{align}
& \bm u(\bm x, 0) = \bm u_0(\bm x),  \qquad p(\bm x, 0) = p_0(\bm x),  \qquad
\bm v(\bm x, 0) = \bm v_0(\bm x).
\end{align}
Here we only consider hyperelastic materials, and the deviatoric part of the Cauchy stress takes the following specific form
\begin{align*}
\bm \sigma_{dev} = J^{-1} \tilde{\bm F} \left( \mathbb P : \tilde{\bm S} \right) \tilde{\bm F}^T, \qquad \tilde{\bm S} = 2 \frac{\partial G^R_{iso}(\tilde{\bm C})}{\partial \tilde{\bm C}}.
\end{align*}
The constitutive relations $\rho = \rho(p)$ and $\beta_{\theta} = \beta_{\theta}(p)$ are given by the relations \eqref{eq:isothermal_constitutive_rho_ghat}. To simplify our subsequent discussion, we introduce the first Piola-Kirchhoff stress $\bm P := J \bm \sigma \bm F^{-T}$, and $\hat{\bm P} := J \bm \sigma_{dev} \bm F^{-T}.$

\begin{remark}
Although upon first glance, the mass equation \eqref{eq:stabilized_solids_strong_form_pressure} may look similar to the pressure rate equation with artificial compressibility proposed by A.J. Chorin in \cite{Chorin1967}, the concept is different. The isothermal compressibility coefficient $\beta_{\theta}$ here is a real physical quantity, not a numerical artifact. It is zero for fully incompressible materials.
\end{remark}

\subsection{Variational multiscale formulation}
\label{subsec:vms_solids}
Based on the VMS formulation derived in Section \ref{sec:VMS}, the formulation for the strong-form problem \eqref{eq:stabilized_solids_strong_form_kinematic}-\eqref{eq:stabilized_solids_strong_form_momentum} can be constructed conveniently. Consider a discretization of the current domain into finite elements. The union of element interiors is denoted by $\Omega_{\bm x}^{\prime}$. Let us denote the finite dimensional trial solution spaces for the solid displacement, pressure, and velocity in the current domain as $\mathcal S_{\bm u_h}$, $\mathcal S_{p_h}$, and $\mathcal S_{\bm v_h}$, respectively. We assume that functions in the trial spaces satisfy the Dirichlet boundary conditions \eqref{eq:stabilized_solids_strong_form_dirichlet_u}-\eqref{eq:stabilized_solids_strong_form_dirichlet_v} on $\Gamma_{\bm x}^g$. Let $\mathcal V_{\bm u_h}$, $\mathcal V_{p_h}$, and $\mathcal V_{\bm v_h}$ denote the corresponding test function spaces. The VMS formulation can be stated as follows. Find $\bm y_h(t) := \left\lbrace  \bm u_h(t), p_h(t), \bm v_h(t)\right\rbrace^T \in \mathcal S_{\bm u_h} \times \mathcal S_{p_h} \times \mathcal S_{\bm v_h}$ such that for $t\in [0, T)$,
\begin{align}
\label{eq:vms_solids_kinematics_current}
& 0 = \mathbf B_k\left( \bm w_{\bm u_h}; \dot{\bm y}_h, \bm y_h  \right) := \int_{\Omega_{\bm x}} \bm w_{\bm u_h} \cdot \left( \frac{d\bm u_h}{dt} - \bm v_h \right) d\Omega_{\bm x}, \displaybreak[2]\\
\label{eq:vms_solids_mass_current}
& 0 = \mathbf B_p\left( w_{p_h}; \dot{\bm y}_h, \bm y_h  \right) := \int_{\Omega_{\bm x}} w_{p_h} \beta_{\theta}(p_h) \frac{dp_h}{dt} + w_{p_h} \nabla_{\bm x} \cdot \bm v_h d\Omega_{\bm x} - \int_{\Omega_{\bm x}^{\prime}} \nabla_{\bm x} w_{p_h} \cdot \bm v^{\prime}  d\Omega_{\bm x}, \displaybreak[2] \\
\label{eq:vms_solids_momentum_current}
& 0 = \mathbf B_m\left( \bm w_{\bm v_h}; \dot{\bm y}_h, \bm y_h  \right) := \int_{\Omega_{\bm x}} \bm w_{\bm v_h} \cdot \rho(p_h) \frac{d\bm v_h}{dt} + \nabla_{\bm x} \bm w_{\bm v_h} : \bm \sigma_{dev}(\bm u_h) - \nabla_{\bm x} \cdot \bm w_{\bm v_h} p_h - \bm w_{\bm v_h} \cdot \rho(p_h)  \bm b d\Omega_{\bm x}  \displaybreak[2] \\
& \hspace{3.5cm} - \int_{\Gamma_{\bm x}^{h}} \bm w_{\bm v_h} \cdot \bm h d\Gamma_{\bm x} - \int_{\Omega_{\bm x}^{\prime}} \nabla_{\bm x} \cdot \bm w_{\bm v_h} p^{\prime} d\Omega_{\bm x}, \displaybreak[2] \nonumber \\
& \bm v^{\prime} := -\bm \tau_M \left( \rho(p_h) \frac{d\bm v_h}{dt} - \nabla_{\bm x} \cdot \bm \sigma_{dev}(\bm u_h) + \nabla_{\bm x} p_h - \rho(p_h) \bm b \right), \displaybreak[2]  \\
\label{eq:vms_solids_p_prime}
& p^{\prime} := -\tau_{C} \left( \beta_{\theta}(p_h) \frac{dp_h}{d t} + \nabla_{\bm x} \cdot \bm v_h \right), 
\end{align}
for $\forall \left\lbrace  \bm w_{\bm u_h} , w_{p_h}, \bm w_{\bm v_h}\right\rbrace \in \mathcal V_{\bm u_h} \times \mathcal V_{p_h} \times \mathcal V_{\bm v_h}$, with $\bm y_h(0) = \left\lbrace \bm u_{h0}, p_{h0}, \bm v_{h0} \right\rbrace^T$. Here $\bm u_{h0}$, $p_{h0}$, and $\bm v_{h0}$ are the $\mathcal L^2$ projections of the initial data onto the finite dimensional spaces $\mathcal S_{\bm u_h}$, $\mathcal S_{p_h}$, and $\mathcal S_{\bm v_h}$, respectively. The subscript $h$ denotes a mesh parameter. In the above and henceforth, the formulations for the kinematic equations, the mass equation, and the linear momentum equations are indicated by the subscripts $k$, $p$ and $m$, respectively. 
\begin{remark}
The VMS formulation \eqref{eq:vms_solids_kinematics_current}-\eqref{eq:vms_solids_p_prime} is derived directly from \eqref{eq:VMS_final_abstract_framework} by adopting the finite dimensional trial solution spaces as $\overline{\bsfV}$ and taking the following assumptions, 
\begin{align*}
&(1) \hspace{5mm} \left(\bm w_{\bm u_h}, \bm v^{\prime} \right)_{\Omega_{\bm x}} = 0; \displaybreak[2] \\
&(2) \hspace{5mm} \left(\bm w_{\bm u_h} , d\bm u^{\prime} /dt \right)_{\Omega_{\bm x}} = \left(w_{p_h} , \beta_{\theta}(p_h+p^{\prime})dp^{\prime} /dt \right)_{\Omega_{\bm x}} = \left(\bm w_{\bm v_h} , \rho(p_h+p^{\prime})d\bm v^{\prime} /dt \right)_{\Omega_{\bm x}} = 0;  \displaybreak[2]  \\
&(3) \hspace{5mm} \bm u^{\prime} = \bm v^{\prime} = \bm 0 \mbox{ on the boundary }; \displaybreak[2] \\
&(4) \hspace{5mm} \left(w_{p_h}, \beta_{\theta}(p_h+p^{\prime})dp_h/dt \right)_{\Omega_{\bm x}} = \left(w_{p_h}, \beta_{\theta}(p_h)dp_h/dt \right)_{\Omega_{\bm x}}; \displaybreak[2] \\
&(5) \hspace{5mm} \left( \bm w_{\bm v_h}, \rho(p_h+p^{\prime})d\bm v_h/dt \right)_{\Omega_{\bm x}} = \left( \bm w_{\bm v_h}, \rho(p_h)d\bm v_h/dt \right)_{\Omega_{\bm x}}; \displaybreak[2] \\
&(6) \hspace{5mm} \left( \nabla_{\bm x} \bm w_{\bm v_h}, \bm \sigma_{dev}\left(\bm u_h + \bm u^{\prime} \right) \right)_{\Omega_{\bm x}} = \left( \nabla_{\bm x} \bm w_{\bm v_h}, \bm \sigma_{dev}\left(\bm u_h \right) \right)_{\Omega_{\bm x}}.
\end{align*}
The assumptions are adopted to simplify the numerical model. Similar assumptions have been made in the residual-based VMS modeling for turbulence \cite{Bazilevs2007a}.
\end{remark}

The VMS formulation \eqref{eq:vms_solids_kinematics_current}-\eqref{eq:vms_solids_p_prime} can be pulled back to the material frame of reference. To obtain that, we define the test functions defined in the material frame of reference as $\bm W_{\bm U_h}(\bm X, t) := \bm w_{\bm u_h}(\bm \varphi(\bm X,t),t)$, $W_{P_h}(\bm X,t) := w_{p_h}(\bm \varphi(\bm X,t), t)$, $\bm W_{\bm V_h}(\bm X, t) := \bm w_{\bm v_h}(\bm \varphi(\bm X,t),t)$,
and we define
\begin{align*}
& \bm U_h(\bm X, t) := \bm u_h( \bm \varphi(\bm X,t), t ), && P_h(\bm X, t) := p_h(\bm \varphi(\bm X,t),t), && \bm V_h(\bm X, t) := \bm v_h( \bm \varphi(\bm X,t), t ), \\
& \bm B(\bm X, t) := \bm b(\bm \varphi(\bm X,t),t), && \bm H(\bm X, t) := \bm h(\bm \varphi(\bm X,t),t), && \bm G(\bm X, t) := \bm g(\bm \varphi(\bm X,t),t).
\end{align*}
The corresponding trial solution spaces on $\Omega_{\bm X}$ are denoted as $\mathcal S_{\bm U_h}$, $\mathcal S_{P_h}$, and $\mathcal S_{\bm V_h}$; the corresponding test function spaces are denoted as $\mathcal V_{\bm U_h}$, $\mathcal V_{P_h}$, and $\mathcal V_{\bm V_h}$. The VMS formulation in the material frame of reference can be stated as follows. Find $\bm Y_h(t) := \left\lbrace \bm U_h(t), P_h(t), \bm V_h(t) \right\rbrace^T \in \mathcal S_{\bm U_h} \times \mathcal S_{P_h} \times \mathcal S_{\bm V_h}$, such that for $t\in [0, T)$,
\begin{align}
\label{eq:vms_solids_kinematics_material}
& 0 =  \mathbf B_k\left( \bm W_{\bm U_h} ; \dot{\bm Y}_h,  \bm Y_h \right) := \int_{\Omega_{\bm X}} J_h \bm W_{\bm U_h} \cdot \left( \frac{d\bm U_h}{dt} - \bm V_h \right) d\Omega_{\bm X}, \displaybreak[2] \\
\label{eq:vms_solids_mass_material}
& 0 =  \mathbf B_p\left( W_{P_h} ; \dot{\bm Y}_h,  \bm Y_h \right) := \int_{\Omega_{\bm X}} J_h W_{P_h} \beta_{\theta}(P_h)\frac{dP}{dt} + W_{P_h} \nabla_{\bm X}\bm V_h : \left(J_h \bm F_h^{-T} \right) d\Omega_{\bm X}  \nonumber \\
& \hspace{0.3cm} + \int_{\Omega^{\prime}_{\bm X}} \left(\nabla_{\bm X} W_{P_h} \bm F_h^{-1} \right) \cdot \bm \tau_M \left( J_h \rho(P_h) \frac{d\bm V_h}{dt} - \nabla_{\bm X} \cdot \hat{\bm P}(\bm U_h) + \left( J_h \bm F_h^{-T} \right) \nabla_{\bm X} P_h  - J_h \rho(P_h) \bm B \right) d\Omega_{\bm X}, \displaybreak[2]  \\
\label{eq:vms_solids_momentum_material}
& 0 =  \mathbf B_m\left( \bm W_{\bm V_h} ; \dot{\bm Y}_h,  \bm Y_h \right) \nonumber \\
& \hspace{2mm} := \int_{\Omega_{\bm X}} \bm W_{\bm V_h} \cdot \left( J_h \rho(P_h) \frac{d\bm V_h}{dt} \right)  + \nabla_{\bm X}\bm W_{\bm V_h} : \hat{\bm P}(\bm U_h) - \nabla_{\bm X} \bm W_{\bm V_h} : \left(J_h \bm F_h^{-T} \right) P_h - \bm W_{\bm V_h} \cdot \left(J_h \rho(P_h) \bm B \right) d\Omega_{\bm X}  \nonumber \\
& \hspace{0.5cm}  - \int_{\Gamma_{\bm X}^H} \bm W_{\bm V_h} \cdot \bm H d\Gamma_{\bm X} + \int_{\Omega^{\prime}_{\bm X}} \left( \nabla_{\bm X} \bm W_{\bm V_h} : \bm F_h^{-T} \right) \tau_C \left( J_h \beta_{\theta}(P_h) \frac{dP_h}{dt} + \nabla_{\bm X} \bm V_h : \left( J_h \bm F_h^{-T} \right) \right) d\Omega_{\bm X},
\end{align}
for $\forall \left\lbrace \bm W_{\bm U_h}, W_{P_h}, \bm W_{\bm V_h} \right\rbrace \in \mathcal V_{\bm U_h} \times \mathcal V_{P_h} \times \mathcal V_{\bm V_h}$, with $J_h = \textup{det}\left( \bm F_h \right)$, $\bm F_h = \nabla_{\bm X} \bm U_h + \bm I_{n_d}$, $\bm I_{n_d}$ being the $n_d \times n_d$ identity matrix, $\dot{\bm Y}_h(t) := \left\lbrace d\bm U_h/dt, dP_h/dt, d\bm V_h/dt \right\rbrace^T$, and $\bm Y_h(0) = \left\lbrace \bm U_{h0}, P_{h0}, \bm V_{h0} \right\rbrace^T$. Here $\bm U_{h0}$, $P_{h0}$, and $\bm V_{h0}$ are the $\mathcal L^2$ projections of the initial data onto the finite dimensional spaces $\mathcal S_{\bm U_h}$, $\mathcal S_{P_h}$, and $\mathcal S_{\bm V_h}$, respectively.

\begin{remark}
If piecewise linear elements are used for the spatial discretization, the term,
\begin{align*}
\int_{\Omega^{\prime}_{\bm x}} \nabla_{\bm x} w_{p_h} \cdot \tau_m \nabla_{\bm x} \cdot \bm \sigma_{dev}(\bm u_h) d\Omega_{\bm x},
\end{align*}
in \eqref{eq:vms_solids_mass_current} and the term
\begin{align*}
\int_{\Omega^{\prime}_{\bm X}} \left(\nabla_{\bm X} W_{P_h} \bm F_h^{-1} \right) \cdot \tau_m \left( \nabla_{\bm X} \cdot \hat{\bm P}(\bm U^h)  \right) d\Omega_{\bm X},
\end{align*}
in \eqref{eq:vms_solids_mass_material} vanish.
\end{remark}

\subsection{Temporal discretization}
\label{sec:time_discretization}
A desirable time integration algorithm for structural dynamics should satisfy at least three requirements: unconditional stability, second-order accuracy, and dissipation on high-frequency modes. The generalized-$\alpha$ method is one type of time integration schemes that satisfy these requirements \cite{Chung1993,Jansen2000,Kadapa2017}. For this reason, we adopt this method for time integration in this work. The time interval $[0,T)$ is divided into a set of $n_{ts}$ subintervals of size $\Delta t_n := t_{n+1} - t_n$ delimited by a discrete time vector $\left\lbrace t_n \right\rbrace_{n=0}^{n_{ts}}$. The solution vector and its first-order time derivative evaluated at the time step $t_n$ are denoted as $\bm Y_n$ and $\dot{\bm Y}_n$; the basis function for the discrete function spaces is denoted as $\hat{N}_A$. With those, the residual vectors can be represented as
\begin{align*}
\boldsymbol{\mathrm R}_k\left(\dot{\bm Y}_{n}, \bm Y_{n} \right) &: = \left\lbrace \mathbf B_k\left( \hat{N}_A \bm e_i ;  \dot{\bm Y}_{n}, \bm Y_{n} \right) \right\rbrace , \displaybreak[2] \\
\boldsymbol{\mathrm R}_p\left(\dot{\bm Y}_{n}, \bm Y_{n} \right) &: =\left\lbrace \mathbf B_p\left( \hat{N}_A ;  \dot{\bm Y}_{n}, \bm Y_{n} \right) \right\rbrace, \displaybreak[2] \\
\boldsymbol{\mathrm R}_m\left(\dot{\bm Y}_{n}, \bm Y_{n} \right) &: =\left\lbrace \mathbf B_m\left( \hat{N}_A \bm e_i ;  \dot{\bm Y}_{n}, \bm Y_{n} \right) \right\rbrace.
\end{align*}
The fully discrete problem can be stated as follows. At time step $t_n$, given $\dot{\bm Y}_n$, $\bm Y_n$, the time step $\Delta t_n$, and the parameters $\alpha_m$, $\alpha_f$, and $\gamma$, find $\dot{\bm Y}_{n+1}$ and $\bm Y_{n+1}$ such that 
\begin{align}
\label{eq:gen_alpha_fully_discrete_solid_kinematic_eqn}
& \boldsymbol{\mathrm R}_k(\dot{\bm Y}_{n+\alpha_m}, \bm Y_{n+\alpha_f}) = \bm 0, \displaybreak[2] \\
\label{eq:gen_alpha_fully_discrete_solid_pressure_rate}
& \boldsymbol{\mathrm R}_p(\dot{\bm Y}_{n+\alpha_m}, \bm Y_{n+\alpha_f}) = \bm 0, \displaybreak[2] \\
\label{eq:gen_alpha_fully_discrete_solid_linear_momentum}
& \boldsymbol{\mathrm R}_m(\dot{\bm Y}_{n+\alpha_m}, \bm Y_{n+\alpha_f}) = \bm 0, \displaybreak[2] \\
\label{eq:gen_alpha_def_Y_n_plus_1}
& \bm Y_{n+1} = \bm Y_{n} + \Delta t_n \dot{\bm Y}_n, + \gamma \Delta t_n \left( \dot{\bm Y}_{n+1} - \dot{\bm Y}_{n}\right), \displaybreak[2] \\
\label{eq:gen_alpha_def_Y_n_alpha_m}
& \dot{\bm Y}_{n+\alpha_m} = \dot{\bm Y}_{n} + \alpha_m \left(\dot{\bm Y}_{n+1} - \dot{\bm Y}_{n} \right), \displaybreak[2] \\
\label{eq:gen_alpha_def_Y_n_alpha_f}
& \bm Y_{n+\alpha_f} = \bm Y_{n} + \alpha_f \left( \bm Y_{n+1} - \bm Y_{n} \right).
\end{align}
The parameters $\alpha_m$, $\alpha_f$, and $\gamma$ define the time integration scheme. It has been shown in \cite{Jansen2000} that, for linear problems, second-order accuracy in time can be achieved, provided
\begin{align*}
\gamma = \frac{1}{2} + \alpha_m - \alpha_f,
\end{align*}
and unconditional stability can be attained if
\begin{align*}
\alpha_m \geq \alpha_f \geq \frac12.
\end{align*}
The generalized-$\alpha$ method allows one to control the damping effect on the high frequency modes while maintain the accuracy of the temporal scheme. This desired property is achieved for first-order linear equations by choosing the parameters as
\begin{align}
\label{eq:gen_alpha_1st_order_ODE_parametrization_in_varrho}
\alpha_m = \frac{1}{2}\left( \frac{3-\varrho_{\infty}}{1+\varrho_{\infty}} \right), \quad \alpha_f = \frac{1}{1+\varrho_{\infty}}, \quad \gamma = \frac{1}{1+\varrho_{\infty}},
\end{align}
wherein $\varrho_{\infty}$ denotes the spectral radius of the amplification matrix at the highest mode \cite{Chung1993,Jansen2000}. In long-time nonlinear calculations, it is advisable to choose $\varrho_{\infty}$ strictly less than one to avoid detrimental effects from the high frequency modes \cite{Bazilevs2008,Jansen2000,Liu2013a}. Unless otherwise specified, we choose $\varrho_{\infty} = 0.5$ in this work.
\begin{remark}
Using the generalized-$\alpha$ method for the first-order structural dynamic problem has recently been shown to enjoy improved numerical properties in comparison with the generalized-$\alpha$ method applied for the second-order structural dynamic systems \cite{Kadapa2017}.
\end{remark}

\subsection{A segregated algorithm}
\label{sec:solid_segregated_algorithm}
It can be shown that the fully discrete kinematic equation \eqref{eq:gen_alpha_fully_discrete_solid_kinematic_eqn} is equivalent to 
\begin{align}
\label{eq:gen_alpha_fully_discrete_solid_kinematics_new}
\bar{\boldsymbol{\mathrm R}}_k(\dot{\bm Y}_{n+\alpha_m}, \bm Y_{n+\alpha_f}) := \dot{\bm U}_{n+\alpha_m} - \bm V_{n+\alpha_f} = \bm 0.
\end{align}
Invoking the relations \eqref{eq:gen_alpha_def_Y_n_plus_1}-\eqref{eq:gen_alpha_def_Y_n_alpha_f}, the left-hand side of above equation can be written explicitly as
\begin{align*}
\bar{\boldsymbol{\mathrm R}}_k(\dot{\bm Y}_{n+\alpha_m}, \bm Y_{n+\alpha_f}) = \frac{\alpha_m}{\gamma \Delta t_n} \left( \bm U_{n+1} - \bm U_n \right) + \left( 1 - \frac{\alpha_m}{\gamma} \right) \dot{\bm U}_n - \alpha_f \bm V_{n+1} - (1-\alpha_f) \bm V_n.
\end{align*}
Equations \eqref{eq:gen_alpha_fully_discrete_solid_pressure_rate}, \eqref{eq:gen_alpha_fully_discrete_solid_linear_momentum}, and \eqref{eq:gen_alpha_fully_discrete_solid_kinematics_new} constitute a system of nonlinear algebraic equations and can be solved by the Newton-Raphson method. The solution vector $\bm Y_{n+1}$ at the time step $t_{n+1}$ with $n=0, \cdots, n_{ts}-1$ is solved iteratively by means of a predictor multi-corrector scheme. We define 
\begin{align*}
\bm Y_{n+1,(i)} := \left\lbrace \bm U_{n+1,(i)}, P_{n+1,(i)}, \bm V_{n+1,(i)} \right\rbrace^T
\end{align*}
as the solution vector at the Newton-Raphson iteration step $i=0, \cdots, i_{max}$. We denote that the residual vectors evaluated at the iteration stage $i$ as
\begin{align}
\label{eq:definition_residual_R_at_intermediate_i}
\boldsymbol{\mathrm R}_{(i)} &:= \left\lbrace \bar{\boldsymbol{\mathrm R}}_{k,(i)}, \boldsymbol{\mathrm R}_{p,(i)}, \boldsymbol{\mathrm R}_{m,(i)} \right\rbrace^T, \\
\bar{\boldsymbol{\mathrm R}}_{k,(i)} &:= \bar{\boldsymbol{\mathrm R}}_k\left( \dot{\bm Y}_{n+\alpha_m, (i)}, \bm Y_{n+\alpha_f, (i)} \right), \\
\boldsymbol{\mathrm R}_{p,(i)} &:= \boldsymbol{\mathrm R}_p\left( \dot{\bm Y}_{n+\alpha_m, (i)}, \bm Y_{n+\alpha_f, (i)} \right), \\
\label{eq:definition_residual_R_m_at_intermediate_i}
\boldsymbol{\mathrm R}_{m,(i)} &:= \boldsymbol{\mathrm R}_m\left( \dot{\bm Y}_{n+\alpha_m, (i)}, \bm Y_{n+\alpha_f, (i)} \right).
\end{align}
The tangent matrix of the above nonlinear system is denoted as
\begin{align*}
\boldsymbol{\mathrm K}_{(i)} =
\begin{bmatrix}
\boldsymbol{\mathrm K}_{k,(i), \dot{\bm U}} & \boldsymbol{\mathrm K}_{k,(i), \dot{P}} & \boldsymbol{\mathrm K}_{k,(i), \dot{\bm V}} \\[0.3mm]
\boldsymbol{\mathrm K}_{p,(i), \dot{\bm U}} & \boldsymbol{\mathrm K}_{p,(i), \dot{P}} & \boldsymbol{\mathrm K}_{p,(i), \dot{\bm V}} \\[0.3mm]
\boldsymbol{\mathrm K}_{m,(i), \dot{\bm U}} & \boldsymbol{\mathrm K}_{m,(i), \dot{P}} & \boldsymbol{\mathrm K}_{m,(i), \dot{\bm V}}
\end{bmatrix},
\end{align*}
wherein
\begin{align}
\label{eq:definition_Kk_U_block}
& \boldsymbol{\mathrm K}_{k,(i), \dot{\bm U}} := \alpha_m \frac{\partial \bar{\boldsymbol{\mathrm R}}_{k,(i)}\left( \dot{\bm Y}_{n+\alpha_m, (i)}, \bm Y_{n+\alpha_f, (i)} \right)}{\partial \dot{\bm U}_{n+\alpha_m}} = \alpha_m \boldsymbol{\mathrm I}, \displaybreak[2] \\
& \boldsymbol{\mathrm K}_{k,(i), \dot{P}} := \bm 0, \displaybreak[2] \\
\label{eq:definition_Kk_V_block}
& \boldsymbol{\mathrm K}_{k,(i), \dot{\bm V}} := \alpha_f \gamma \Delta t_n \frac{\partial \bar{\boldsymbol{\mathrm R}}_{k,(i)}\left( \dot{\bm Y}_{n+\alpha_m, (i)}, \bm Y_{n+\alpha_f, (i)} \right)}{\partial \bm V_{n+\alpha_f}} = -\alpha_f \gamma \Delta t_n \boldsymbol{\mathrm I}, \displaybreak[2] \\
\label{eq:definition_Kp_U_block}
& \boldsymbol{\mathrm K}_{p,(i), \dot{\bm U}} := \alpha_f \gamma \Delta t_n \frac{\partial \boldsymbol{\mathrm R}_{p,(i)}\left( \dot{\bm Y}_{n+\alpha_m, (i)}, \bm Y_{n+\alpha_f, (i)} \right)}{\partial \bm U_{n+\alpha_f}}, \displaybreak[2] \\
& \boldsymbol{\mathrm K}_{p,(i), \dot{P}} := \alpha_m \frac{\partial \boldsymbol{\mathrm R}_{p,(i)}\left( \dot{\bm Y}_{n+\alpha_m, (i)}, \bm Y_{n+\alpha_f, (i)} \right)}{\partial \dot{P}_{n+\alpha_m}} + \alpha_f \gamma \Delta t_n \frac{\partial \boldsymbol{\mathrm R}_{p,(i)}\left( \dot{\bm Y}_{n+\alpha_m, (i)}, \bm Y_{n+\alpha_f, (i)} \right)}{\partial P_{n+\alpha_f}}, \displaybreak[2] \\
& \boldsymbol{\mathrm K}_{p,(i), \dot{\bm V}} := \alpha_m \frac{\partial \boldsymbol{\mathrm R}_{p,(i)}\left( \dot{\bm Y}_{n+\alpha_m, (i)}, \bm Y_{n+\alpha_f, (i)} \right)}{\partial \dot{\bm V}_{n+\alpha_m}} + \alpha_f \gamma \Delta t_n \frac{\partial \boldsymbol{\mathrm R}_{p,(i)}\left( \dot{\bm Y}_{n+\alpha_m, (i)}, \bm Y_{n+\alpha_f, (i)} \right)}{\partial \bm V_{n+\alpha_f}}, \displaybreak[2] \\
& \boldsymbol{\mathrm K}_{m,(i), \dot{\bm U}} := \alpha_f \gamma \Delta t_n \frac{\partial \boldsymbol{\mathrm R}_{m,(i)}\left( \dot{\bm Y}_{n+\alpha_m, (i)}, \bm Y_{n+\alpha_f, (i)} \right)}{\partial \bm U_{n+\alpha_f}}, \displaybreak[2] \\
& \boldsymbol{\mathrm K}_{m,(i), \dot{P}} := \alpha_m \frac{\partial \boldsymbol{\mathrm R}_{m,(i)}\left( \dot{\bm Y}_{n+\alpha_m, (i)}, \bm Y_{n+\alpha_f, (i)} \right)}{\partial \dot{P}_{n+\alpha_m}} + \alpha_f \gamma \Delta t_n \frac{\partial \boldsymbol{\mathrm R}_{m,(i)}\left( \dot{\bm Y}_{n+\alpha_m, (i)}, \bm Y_{n+\alpha_f, (i)} \right)}{\partial P_{n+\alpha_f}}, \displaybreak[2] \\
\label{eq:definition_Km_V_block}
& \boldsymbol{\mathrm K}_{m,(i), \dot{\bm V}} := \alpha_m \frac{\partial \boldsymbol{\mathrm R}_{m,(i)}\left( \dot{\bm Y}_{n+\alpha_m, (i)}, \bm Y_{n+\alpha_f, (i)} \right)}{\partial \dot{\bm V}_{n+\alpha_m}} + \alpha_f \gamma \Delta t_n \frac{\partial \boldsymbol{\mathrm R}_{m,(i)}\left( \dot{\bm Y}_{n+\alpha_m, (i)}, \bm Y_{n+\alpha_f, (i)} \right)}{\partial \bm V_{n+\alpha_f}}.
\end{align}
Notice that \eqref{eq:definition_Kk_U_block} and \eqref{eq:definition_Kk_V_block} are diagonal matrices. This special structure of the two sub-matrices can be exploited to perform a block decomposition of the tangent matrix $\boldsymbol{\mathrm K}$:
\begin{align*}
\boldsymbol{\mathrm K}_{(i)} =
\begin{bmatrix}
\boldsymbol{\mathrm I} & \bm 0 & \bm 0 \\[0.3mm]
\frac{1}{\alpha_m}\boldsymbol{\mathrm K}_{p,(i), \dot{\bm U}} & \boldsymbol{\mathrm K}_{p,(i), \dot{P}} & \boldsymbol{\mathrm K}_{p,(i), \dot{\bm V}} + \frac{\alpha_f \gamma \Delta t_n}{\alpha_m} \boldsymbol{\mathrm K}_{p,(i), \dot{\bm U}} \\[0.3mm]
\frac{1}{\alpha_m}\boldsymbol{\mathrm K}_{m,(i), \dot{\bm U}} & \boldsymbol{\mathrm K}_{m,(i), \dot{P}} & \boldsymbol{\mathrm K}_{m,(i), \dot{\bm V}} +  \frac{\alpha_f \gamma \Delta t_n}{\alpha_m} \boldsymbol{\mathrm K}_{m,(i), \dot{\bm U}}
\end{bmatrix}
\begin{bmatrix}
\alpha_m \boldsymbol{\mathrm I} & \bm 0 & -\alpha_f \gamma \Delta t_n \boldsymbol{\mathrm I} \\[0.3mm]
\bm 0 & \boldsymbol{\mathrm I} & \bm 0 \\[0.3mm]
\bm 0 & \bm 0 & \boldsymbol{\mathrm I}
\end{bmatrix}.
\end{align*}
With the above decomposition, the original linear system of equations for the Newton-Raphson iteration,
\begin{align*}
\boldsymbol{\mathrm K}_{(i)} \Delta \dot{\bm Y}_{n+1,(i)} = - \boldsymbol{\mathrm R}_{(i)}.
\end{align*}
can be solved in a two-stage segregated algorithm. In stage one, the following linear system of equations are solved to obtain the intermediate unknowns $\big[\Delta \dot{\bm U}^{*}_{n+1,(i)},  \Delta \dot{P}^{*}_{n+1,(i)}, \Delta \dot{\bm V}^{*}_{n+1,(i)} \big]^T$.
\begin{align}
\label{eq:newton_linear_system_separated_eqn_1}
\begin{bmatrix}
\boldsymbol{\mathrm I} & \bm 0 & \bm 0 \\[0.3mm]
\frac{1}{\alpha_m}\boldsymbol{\mathrm K}_{p,(i), \dot{\bm U}} & \boldsymbol{\mathrm K}_{p,(i), \dot{P}} & \boldsymbol{\mathrm K}_{p,(i), \dot{\bm V}} + \frac{\alpha_f \gamma \Delta t_n}{\alpha_m} \boldsymbol{\mathrm K}_{p,(i), \dot{\bm U}} \\[0.3mm]
\frac{1}{\alpha_m}\boldsymbol{\mathrm K}_{m,(i), \dot{\bm U}} & \boldsymbol{\mathrm K}_{m,(i), \dot{P}} & \boldsymbol{\mathrm K}_{m,(i), \dot{\bm V}} +  \frac{\alpha_f \gamma \Delta t_n}{\alpha_m} \boldsymbol{\mathrm K}_{m,(i), \dot{\bm U}}
\end{bmatrix}
\begin{bmatrix}
\Delta \dot{\bm U}^{*}_{n+1,(i)}  \\[0.3mm]
\Delta \dot{P}^{*}_{n+1,(i)} \\[0.3mm]
\Delta \dot{\bm V}^{*}_{n+1,(i)}
\end{bmatrix}
= -
\begin{bmatrix}
\bar{\boldsymbol{\mathrm R}}_{k,(i)} \\[0.3mm] 
\boldsymbol{\mathrm R}_{p,(i)} \\[0.3mm]
\boldsymbol{\mathrm R}_{m,(i)}
\end{bmatrix}.
\end{align}
In stage two, one solves the upper triangular matrix problem
\begin{align}
\label{eq:newton_linear_system_separated_eqn_2}
\begin{bmatrix}
\alpha_m \boldsymbol{\mathrm I} & \bm 0 & -\alpha_f \gamma \Delta t_n \boldsymbol{\mathrm I} \\[0.3mm]
\bm 0 & \boldsymbol{\mathrm I} & \bm 0 \\[0.3mm]
\bm 0 & \bm 0 & \boldsymbol{\mathrm I}
\end{bmatrix}
\begin{bmatrix}
\Delta \dot{\bm U}_{n+1,(i)}  \\[0.3mm]
\Delta \dot{P}_{n+1,(i)} \\[0.3mm]
\Delta \dot{\bm V}_{n+1,(i)}
\end{bmatrix}
= 
\begin{bmatrix}
\Delta \dot{\bm U}^{*}_{n+1,(i)}  \\[0.3mm]
\Delta \dot{P}^{*}_{n+1,(i)} \\[0.3mm]
\Delta \dot{\bm V}^{*}_{n+1,(i)}
\end{bmatrix}
\end{align}
to get the increments $\big[\Delta \dot{\bm U}_{n+1,(i)},  \Delta \dot{P}_{n+1,(i)}, \Delta \dot{\bm V}_{n+1,(i)} \big]^T$ for the iteration step $i$. From \eqref{eq:newton_linear_system_separated_eqn_1} and \eqref{eq:newton_linear_system_separated_eqn_2}, we have the following observations,
\begin{align*}
& \alpha_m \Delta \dot{\bm U}_{n+1,(i)} - \alpha_f \gamma \Delta t_n \Delta \dot{\bm V}_{n+1,(i)} = \Delta \dot{\bm U}^{*}_{n+1,(i)} = - \bar{\boldsymbol{\mathrm R}}^k_{(i)}, \\
& \Delta \dot{P}_{n+1,(i)} = \Delta \dot{P}^{*}_{n+1,(i)}, \\
& \Delta \dot{\bm V}_{n+1,(i)} = \Delta \dot{\bm V}^{*}_{n+1,(i)}.
\end{align*}
With the above relations, the linear system \eqref{eq:newton_linear_system_separated_eqn_1} can be reduced to a smaller problem,
\begin{align}
\label{eq:newton_linear_system_smaller_eqn_1}
&
\begin{bmatrix}
\boldsymbol{\mathrm K}_{p,(i), \dot{P}} & \boldsymbol{\mathrm K}_{p,(i), \dot{\bm V}} + \frac{\alpha_f \gamma \Delta t_n}{\alpha_m} \boldsymbol{\mathrm K}_{p,(i), \dot{\bm U}} \\[0.3mm]
\boldsymbol{\mathrm K}_{m,(i), \dot{P}} & \boldsymbol{\mathrm K}_{m,(i), \dot{\bm V}} +  \frac{\alpha_f \gamma \Delta t_n}{\alpha_m} \boldsymbol{\mathrm K}_{m,(i), \dot{\bm U}}
\end{bmatrix}
\begin{bmatrix}
\Delta \dot{P}_{n+1,(i)} \\[0.3mm]
\Delta \dot{\bm V}_{n+1,(i)}
\end{bmatrix}
= -
\begin{bmatrix}
\boldsymbol{\mathrm R}_{p,(i)} - \frac{1}{\alpha_m}  \boldsymbol{\mathrm K}_{p,(i), \dot{\bm U}}\bar{\boldsymbol{\mathrm R}}_{k,(i)} \\[0.3mm]
\boldsymbol{\mathrm R}_{m,(i)}- \frac{1}{\alpha_m}  \boldsymbol{\mathrm K}_{m,(i), \dot{\bm U}}\bar{\boldsymbol{\mathrm R}}_{k,(i)}
\end{bmatrix}.
\end{align}
The linear system \eqref{eq:newton_linear_system_separated_eqn_2} can be reduced to the relation
\begin{align}
\label{eq:newton_linear_system_smaller_eqn_2}
& \Delta \dot{\bm U}_{n+1,(i)} = \frac{\alpha_f \gamma \Delta t_n}{\alpha_m} \Delta \dot{\bm V}_{n+1,(i)} - \frac{1}{\alpha_m}\bar{\boldsymbol{\mathrm R}}^k_{(i)}.
\end{align}
Therefore, the solution procedure can be consistently rewritten into two smaller problems. One first solves the equation \eqref{eq:newton_linear_system_smaller_eqn_1} to obtain $\big[\Delta \dot{P}_{n+1,(i)}, \Delta \dot{\bm V}_{n+1,(i)} \big]^T$. Then the displacement increment $\Delta \dot{\bm U}_{n+1,(i)}$ can be obtained through the relation \eqref{eq:newton_linear_system_smaller_eqn_2}. We can summarize the above discussion as the following segregated predictor multi-corrector algorithm.

\noindent \textbf{Predictor stage}: Set:
\begin{align*}
\bm Y_{n+1, (0)} &= \bm Y_{n}, \\
\dot{\bm Y}_{n+1, (0)} &= \frac{\gamma - 1}{\gamma} \dot{\bm Y}_{n}.
\end{align*}

\noindent \textbf{Multi-corrector stage}:
Repeat the following steps \(i =1, \dots, i_{max}\):
\begin{enumerate}
\item Evaluate the solution vectors at the intermediate stages:
\begin{align*}
\dot{\bm Y}_{n+\alpha_m, (i)} &= \dot{\bm Y}_n + \alpha_m \left( \dot{\bm Y}_{n+1,(i-1)} - \dot{\bm Y}_n \right), \\
\bm Y_{n+\alpha_f, (i)} &= \bm Y_n + \alpha_f \left( \bm Y_{n+1,(i-1)} - \bm Y_n \right).
\end{align*}

\item Assemble the residual vectors $\boldsymbol{\mathrm R}_{(i)}$ based on \eqref{eq:definition_residual_R_at_intermediate_i}-\eqref{eq:definition_residual_R_m_at_intermediate_i} using the solution evaluated at the intermediate stages.

\item Let $\|\boldsymbol{\mathrm R}_{(i)}\|_{l^2}$ denote the $l^2$-norm of the residual vector. If either one of the following stopping criteria 
\begin{align*}
& \frac{\|\boldsymbol{\mathrm R}_{(i)}\|_{l^2}}{\|\boldsymbol{\mathrm R}_{(0)}\|_{l^2}} \leq tol_R, \qquad \|\boldsymbol{\mathrm R}_{(i)}\|_{l^2} \leq tol_A,
\end{align*}
is satisfied for prescribed tolerances $tol_R$, $tol_A$, set the solution vector at time step $t_{n+1}$ as $\dot{\bm Y}_{n+1} = \dot{\bm Y}_{n+1, (i-1)}$ and $\bm Y_{n+1} = \bm Y_{n+1, (i-1)}$, and exit the multi-corrector stage; otherwise, continue to step 4.

\item Assemble the tangent matrices \eqref{eq:definition_Kp_U_block}-\eqref{eq:definition_Km_V_block}.

\item Solve the linear system of equations \eqref{eq:newton_linear_system_smaller_eqn_1} for $\Delta \dot{P}_{n+1,(i)}$ and $\Delta \dot{\bm V}_{n+1,(i)}$.

\item Obtain $\Delta \dot{\bm U}_{n+1,(i)}$ from the solution $\Delta \dot{\bm V}_{n+1,(i)}$ according to the relation \eqref{eq:newton_linear_system_smaller_eqn_2}.

\item Update the solution vectors as
\begin{align*}
\dot{\bm Y}_{n+1,(i)} &= \dot{\bm Y}_{n+1,(i)} + \Delta \dot{\bm Y}_{n+1,(i)}, \\
\bm Y_{n+1,(i)} &= \bm Y_{n+1,(i)} + \gamma \Delta t_n \Delta \dot{\bm Y}_{n+1,(i)}.
\end{align*}
and return to step 1.
\end{enumerate}
Since the kinematic equation is linear, it is reasonable to expect that this equation can be solved in one Newton-Raphson iteration. Indeed, we have the following proposition.
\begin{proposition}
\label{prop:Rk_2_zero}
In the above algorithm, $\bar{\boldsymbol{\mathrm R}}_{k,(i)} = \bm 0$ for $i \geq 2$.
\end{proposition}
\noindent The proof of this proposition is given in \ref{appd:proof_prop_Rk_2_zero}. Due to this fact, the right-hand side of \eqref{eq:newton_linear_system_smaller_eqn_1} becomes $-\left[\boldsymbol{\mathrm R}_{p,(i)}, \boldsymbol{\mathrm R}_{m,(i)} \right]^T$, and the relation \eqref{eq:newton_linear_system_smaller_eqn_2} reduces to
\begin{align*}
\Delta \dot{\bm U}_{n+1,(i)} = \frac{\alpha_f \gamma \Delta t_n}{\alpha_m} \Delta \dot{\bm V}_{n+1,(i)}, \mbox{ for } i \geq 2.
\end{align*}
\begin{remark}
In general, $\bar{\boldsymbol{\mathrm R}}^k_{(1)} \neq \bm 0$. In the algorithm we presented, the same-$\boldsymbol{\mathrm Y}$ predictor \cite{Jansen2000} is adopted, and it is straightforward to show that, with this predictor,
\begin{align*}
\bar{\boldsymbol{\mathrm R}}^k_{(1)} &= \frac{\alpha_m}{\gamma \Delta t_n} \left( \bm U_{n+1,(0)} - \bm U_n \right) + \left( 1 - \frac{\alpha_m}{\gamma} \right) \dot{\bm U}_n - \alpha_f \bm V_{n+1,(0)} - (1-\alpha_f) \bm V_n \\
&= \left( 1 - \frac{\alpha_m}{\gamma} \right) \dot{\bm U}_n - \bm V_n.
\end{align*}
If the time-stepping parameters are $\alpha_m = \alpha_f = \gamma = 1$ and the predictors are chosen as
\begin{align*}
\bm U_{n+1,(0)} = \bm U_n, \qquad
\bm V_{n+1,(0)} = \bm 0,
\end{align*}
one can get $\bar{\boldsymbol{\mathrm R}}^k_{(1)} = \bm 0$. Setting $\alpha_m = \alpha_f = \gamma = 1$ corresponds to the backward Euler method.
In \cite{Rossi2016}, a non-traditional predictor is chosen to enforce $\bar{\boldsymbol{\mathrm R}}^k_{(1)} = \bm 0$.
\end{remark}
\begin{remark}
This segregated algorithm based on block decomposition was first proposed by G. Scovazzi and co-authors in \cite{Scovazzi2016}. This solution procedure \eqref{eq:newton_linear_system_smaller_eqn_1}-\eqref{eq:newton_linear_system_smaller_eqn_2} significantly reduces the size of the linear system for the implicit solver. In comparison with the second-order pure-displacement formulation, writing the system into a first-order system does not lead to a significant increase in computational cost. As will be revealed in Section \ref{sec:fsi_temporal}, the first-order formulation is very appealing for FSI problems.
\end{remark}

\subsection{Stabilization parameters}
The design of the stabilization parameter is of critical importance for the behavior of the stabilized finite element formulation. In this work, a practical setting of the stabilization parameters is
\begin{align*}
\bm \tau_M = \tau_M \bm I_{n_d}, \quad \tau_M = c_m \frac{\Delta x}{c \rho}, \quad
\tau_C = c_c c \Delta x \rho,
\end{align*}
wherein $\Delta x$ is the diameter of the circumscribing sphere of the tetrahedral element, $c_m$ and $c_c$ are two non-dimensional parameters, $c$ is the maximum wave speed in the solid body. In this work, the formula for $c$ is chosen based on a small-strain isotropic linear elastic material \cite{Hughes1988,Scovazzi2016}. For compressible materials, $c$ is given by the bulk wave speed
\begin{align*}
c = \sqrt{\frac{\lambda + 2 \mu}{\rho_0}},
\end{align*}
and for incompressible materials, $c$ is given by the shear wave speed
\begin{align*}
c = \sqrt{\frac{\mu}{\rho_0}}.
\end{align*}

\begin{remark}
The stabilization parameter is the crux of the design of the VMS formulation. In \cite{Hughes1988}, the authors proposed two choices of the stabilization parameter for the small-strain elastodynamic problem: $0.5 \Delta x / c$ and $0.5 \Delta t$. In \cite{Gil2016, Lee2014}, the stabilization parameter is designed based on $\Delta t$. In \cite{Scovazzi2016}, the authors proposed a set of stabilization parameters as $c_{s15} \Delta t / 2 c_{CFL}$ or $c_{s15} \Delta x / 2c$, wherein $c_{s15}$ is a non-dimensional parameter, and $c_{CFL}$ is the global CFL number. In \cite{Rossi2016}, the stabilization parameter for compressible materials is designed as
\begin{align*}
\frac{1}{2}\max \left[\frac{\Delta x}{100 c}, \min\left( \Delta t, \frac{\Delta x}{c} \right) \right],
\end{align*}
and
\begin{align*}
\frac{ c_{r16} }{2} \max \left[ \frac{\Delta x}{100 c}, \min\left( \Delta t, \frac{\Delta x}{c} \right) \right]
\end{align*}
for incompressible materials, with $c_{r16}$ in the range $[0.01, 0.03]$. The purpose of this design is to enhance the robustness of the algorithm when the time step $\Delta t$ is too large or too small. We favor the design based on $\Delta x / c$ to make the solution independent of the time step size.
\end{remark}

\begin{remark}
In the current design, the elastic wave speed $c$ is based on the small-strain elastic theory. By exploiting the eigenvalue structure of the large deformation problem, one can obtain the elastic wave speed for nonlinear materials \cite{Bonet2015,Gil2014}. This will surely lead to a better design of the stabilization parameters.
\end{remark}

\section{Formulation for fluid dynamics and fluid-structure interaction}
\label{sec:numerical_FSI}
In this section, we present a new framework for FSI problems based on the unified continuum model derived in Section \ref{sec:continuum_mechanics}. The formulation for solid dynamics is directly adopted from the one developed in Section \ref{sec:numerical_solid}; the formulation for fluid dynamics is constructed based on the general VMS formulation given in Section \ref{sec:VMS}. The finite element formulation, time integration, and a new solution procedure are discussed in this section. In all discussions related to FSI problems, we use a superscript $f$ to indicate quantities related to fluids, a superscript $m$ to indicate quantities related to the mesh motion in the fluid subdomain, and a superscript $s$ to indicate quantities related to solids.

\subsection{Initial-boundary value problem}
\begin{figure}
	\begin{center}
	\begin{tabular}{c}
\includegraphics[angle=0, trim=120 160 120 130, clip=true, scale = 0.45]{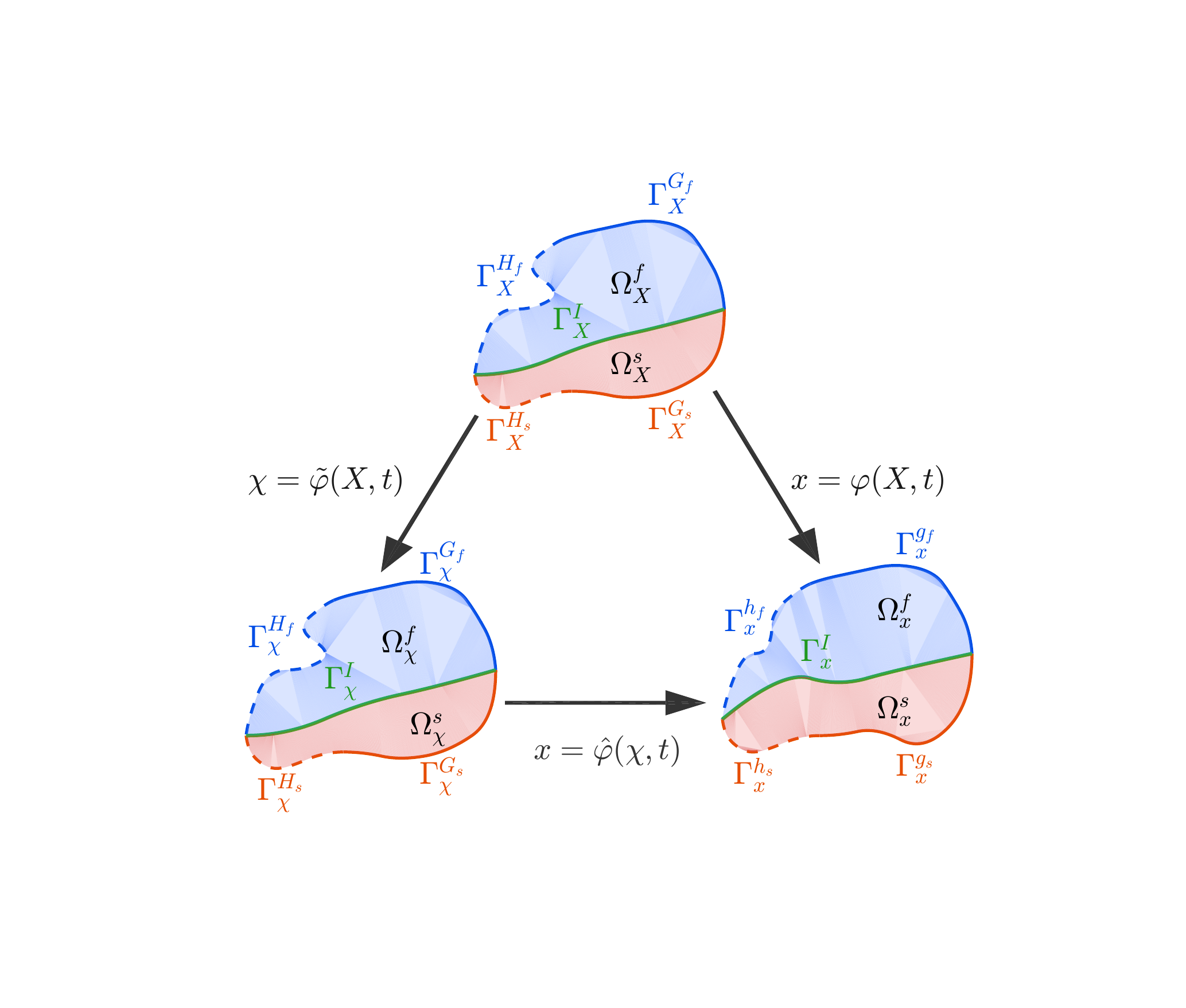}
\end{tabular}
\end{center}
\caption{Illustration of the fluid-structural interaction problem setting.} 
\label{fig:fsi_domain_bc_setting}
\end{figure}
In this section, we present the initial-boundary value problem for the FSI problem. The time interval of interest is denoted as $(0, T)$, with $T>0$. The domain occupied by the continuum body in the referential frame $\Omega_{\bm \chi}$  admits the decomposition
\begin{align*}
\Omega_{\bm \chi} = \overline{\Omega_{\bm \chi}^{f} \cup \Omega_{\bm \chi}^{s} },  \qquad \emptyset = \Omega_{\bm \chi}^{f} \cap \Omega_{\bm \chi}^{s}.
\end{align*}
$\Omega_{\bm \chi}^f$ and $\Omega_{\bm \chi}^s$ are the two subdomains occupied by the fluid and solid respectively. The fluid-solid interface is a $\mathbb R^{n_d -1}$-dimensional manifold and is denoted as $\Gamma_{\bm \chi}^{I}$. The current domain $\Omega_{\bm x}$ can be decomposed correspondingly as
\begin{align*}
\Omega_{\bm x} = \overline{\Omega_{\bm x}^{f} \cup \Omega_{\bm x}^{s} }, \qquad \emptyset = \Omega_{\bm x}^{f} \cap \Omega_{\bm x}^{s}.
\end{align*}
The boundary $\Gamma_{\bm x} = \partial \Omega_{\bm x}$ can be partitioned into four non-overlapping subdivisions:
\begin{align*}
\Gamma_{\bm x} = \Gamma_{\bm x}^{g_s} \cup \Gamma_{\bm x}^{g_f} \cup \Gamma_{\bm x}^{h_s} \cup \Gamma_{\bm x}^{h_f}.
\end{align*}
In the above decomposition, $\Gamma_{\bm x}^{g_s}$ represents the Dirichlet part of the solid boundary; $\Gamma_{\bm x}^{g_f}$ represents the Dirichlet part of the fluid boundary; $\Gamma_{\bm x}^{h_s}$ represents the Neumann part of the solid boundary; $\Gamma_{\bm x}^{h_f}$ represents the Neumann part of the fluid boundary. The unit outward normal vector to $\Gamma_{\bm x}$ is denoted as $\bm n$. The fluid-solid interface in the current configuration is denoted as $\Gamma_{\bm x}^{I}$; $\bm n^s$ and $\bm n^f$ represent the unit outward normal vector to $\Omega_{\bm x}^s$ and $\Omega_{\bm x}^f$ on the interface $\Gamma_{\bm x}^I$, respectively. The configurations and the boundary decomposition are illustrated in Figure \ref{fig:fsi_domain_bc_setting}. The referential configuration can be defined by the Lagrangian-to-referential map
$
\bm \chi = \tilde{\bm \varphi}\left(\bm X, t \right).
$
We define the map $\tilde{\bm \varphi}$ in the solid subdomain $\Omega_{\bm X}^{s}$ as the identity map
\begin{align}
\label{eq:fsi_solid_subdomain_lag_to_ref_map}
\bm \chi = \tilde{\bm \varphi}\left(\bm X, t \right) = \bm X, \quad \mbox{ for } \bm X \in \Omega_{\bm X}^{s} \mbox{ and } t \geq 0.
\end{align}
Due to the composition relation \eqref{eq:ALE_map_composition}, one has $\bm \varphi = \hat{\bm \varphi} \circ \bm{id_{\Omega_{\bm X}^{s}}}$. This implies that the material description is adopted in the solid subdomain. In the fluid subdomain, the referential configuration can be determined by the referential-to-Eulerian map $\hat{\bm \varphi}$. Typically, one determines $\hat{\bm u}^m$ in $\Omega^{f}_{\bm \chi}$ first, and the referential-to-Eulerian map in the fluid subdomain is given by the definition \eqref{eq:ale_mesh_displacement_def}. The motion of the referential domain needs to maintain the regularity for functions defined on $\Omega_{\bm \chi}$ \cite{Nobile2001}. Oftentimes, the construction of this mapping is tailored to specific problems. In this work, we consider two options: the harmonic extension algorithm and the pseudo-linear-elasticity algorithm.
\paragraph{The harmonic extension algorithm}
One simple and effective way of constructing the map $\hat{\bm \varphi}$ is by solving for $\hat{\bm u}^m$ as a harmonic extension of the trace of $\bm U^s$ on $\Gamma^{I}_{\bm \chi}$ \cite{Zeng2016}, viz.
\begin{align}
\label{eq:fsi_geometry_laplace}
-\nabla_{\bm \chi} \cdot \left( \nabla_{\bm \chi} \hat{\bm u}^m \right) &= \bm 0, && \mbox{ in } \Omega_{\bm \chi}^{f}, \displaybreak[2] \\
\hat{\bm u}^m &= \bm U^s, && \mbox{ on } \Gamma^I_{\bm \chi}.
\end{align}
Once $\hat{\bm u}^m$ is obtained, the mesh velocity $\hat{\bm v}^m$ in the fluid subdomain can be calculated as
\begin{align}
\hat{\bm v}^m = \left. \frac{\partial \hat{\bm \varphi}}{\partial t} \right|_{\bm \chi} = \left. \frac{\partial \hat{\bm u}^m }{\partial t} \right|_{\bm \chi} \hspace{3mm} \mbox{ in } \Omega_{\bm \chi}^{f}.
\end{align}
In all, we define the mesh displacement $\hat{\bm u}$ and the mesh velocity $\hat{\bm v}$ as
\begin{align*}
\hat{\bm u} = \left\{
  \begin{array}{ll}
    \bm U^s &  \mbox{ in } \Omega^s_{\bm \chi} = \Omega^s_{\bm X}  \\
    \hat{\bm u}^m & \mbox{ in } \Omega^f_{\bm \chi}
  \end{array}
\right.
,
\hspace{8mm}
\hat{\bm v} = \left\{
  \begin{array}{ll}
    \bm V^s &  \mbox{ in } \Omega^s_{\bm \chi} = \Omega^s_{\bm X} \\
    \hat{\bm v}^m & \mbox{ in } \Omega^f_{\bm \chi}
  \end{array}
\right.
.
\end{align*}

\paragraph{The pseudo-linear-elasticity algorithm}
One may also model the motion of the fluid subdomain by solving a succession of pesudo-linear-elastostatic equations \cite{Johnson1994,Zeng2016}. We consider a time instant $\tilde{t}<t$ such that $\tilde{t}$ is close to $t$. In numerical computations, $\tilde{t}$ is often conveniently chosen as the previous time step. The mesh displacement at time $\tilde{t} < t$ is given by
$
\hat{\bm u}\left( \bm \chi, \tilde{t} \right) :=  \tilde{\bm x} - \bm \chi = \hat{\bm \varphi}\left( \bm \chi, \tilde{t} \right) - \bm \chi.
$
Consequently, one has the following identity,
\begin{align*}
\hat{\bm \varphi}\left( \bm \chi, t \right) = \hat{\bm \varphi}\left( \bm \chi, \tilde{t} \right) + \hat{\bm u}\left( \bm \chi, t \right) - \hat{\bm u}\left( \bm \chi, \tilde{t}\right).
\end{align*}
Freezing the time instants $\tilde{t}$ and $t$, we define $\tilde{\bm u}^m(\tilde{\bm x})$ as $
\tilde{\bm u}^m \circ \hat{\bm \varphi}(\bm \chi, \tilde{t}) := \hat{\bm u}\left( \bm \chi, t \right) - \hat{\bm u}\left( \bm \chi, \tilde{t}\right).
$
Given the fictitious Lam\'e parameters $\lambda^m$ and $\mu^m$, $\hat{\bm u}^m$ is determined by solving the following linear elastostatic equations:
\begin{align}
\label{eq:fsi_geometry_elasticity}
& \nabla_{\tilde{\bm x}}  \cdot \left( \mu^m \left( \nabla_{\tilde{\bm x}} \tilde{\bm u}^m + \left(\nabla_{\tilde{\bm x}} \tilde{\bm u}^m\right)^T  \right) + \lambda^m  \nabla_{\tilde{\bm x}} \cdot \tilde{\bm u}^m \bm I \right) = \bm 0, && \mbox{ in } \Omega^f_{\tilde{\bm x}}, \\
& \tilde{\bm u}^m(\tilde{\bm x}) = \left( \bm u^s_t - \bm u^s_{\tilde{t}} \right) \circ \hat{\bm \varphi}^{-1}\left( \tilde{\bm x}, \tilde{t} \right), && \mbox{ on } \Gamma_{\tilde{\bm x}}^I.
\end{align}
In the above, $\bm u^s_t$ and $\bm u^s_{\tilde{t}}$ are the solid displacements in the current configurations at time $t$ and $\tilde{t}$, respectively. In our implementation, the fictitious Lam\'e parameters are divided by the Jacobian determinant of the element mapping \cite{Johnson1994}. This effectively increases the robustness of the mesh moving algorithm.

The fluid-solid body can be viewed as a single continuum body governed by the following balance equations.
\begin{align}
\label{eq:unified_fsi_continuum_mass}
0 &= \left. \beta_{\theta}(p)\frac{\partial p}{\partial t}\right|_{\bm \chi} + \beta_{\theta}(p) \left( \bm v - \hat{\bm v} \right) \cdot \nabla_{\bm x} p + \nabla_{\bm x} \cdot \bm v, && \mbox{ in } \Omega_{\bm x}, \\
\label{eq:unified_fsi_continuum_linear_momentum}
\bm 0 &= \rho(p)  \left.  \frac{\partial \bm v}{\partial t}\right|_{\bm \chi}  + \rho(p) \left( \nabla_{\bm x} \bm v \right) \left( \bm v - \hat{\bm v} \right)  - \nabla_{\bm x} \cdot \bm \sigma_{dev} + \nabla_{\bm x} p - \rho(p) \bm b, && \mbox{ in } \Omega_{\bm x}.
\end{align}
In the fluid subdomain, we consider viscous incompressible fluid flow, with the constitutive relations given in Section \ref{sec:examples_of_equations}. Due to the fact that $\beta_{\theta}^f(p) = 0$, the strong-form problem \eqref{eq:unified_fsi_continuum_mass}-\eqref{eq:unified_fsi_continuum_linear_momentum} in the fluid subdomain can be further simplified as 
\begin{align}
\label{eq:fsi_fluid_strong_form_mass_eqn}
& 0 = \nabla_{\bm x} \cdot \bm v^f,  &&  \mbox{ in } \Omega_{\bm x}^f, \displaybreak[2] \\
\label{eq:fsi_fluid_strong_form_linear_mom_eqn}
& \bm 0 = \left. \rho_0^f \frac{\partial \bm v^f}{\partial t} \right|_{\bm \chi} + \rho_0^f \left( \nabla_{\bm x} \bm v^f \right) \left(\bm v^f - \hat{\bm v} \right) - \nabla_{\bm x} \cdot \bm \sigma^f_{dev} + \nabla_{\bm x} p^f - \rho_0^f \bm b,  && \mbox{ in } \Omega_{\bm x}^f.
\end{align}
In the solid subdomain, we have $\hat{\bm v} = \bm v$ due to the Lagrangian description \eqref{eq:fsi_solid_subdomain_lag_to_ref_map} adopted in the solid subproblem. Consequently, the strong-form problem \eqref{eq:unified_fsi_continuum_mass}-\eqref{eq:unified_fsi_continuum_linear_momentum} for the solid can be explicitly written as
\begin{align}
\label{eq:fsi_solids_strong_form_kinematic}
& \bm 0 = \frac{d\bm u^s}{dt} - \bm v^s, && \mbox{ in } \Omega_{\bm x}^s, \displaybreak[2] \\
\label{eq:fsi_solids_strong_form_pressure}
& 0 = \beta^{s}_{\theta}(p^s) \frac{dp^s}{dt} + \nabla_{\bm x} \cdot \bm v^s && \mbox{ in } \Omega_{\bm x}^s,  \displaybreak[2] \\
& \bm 0 = \rho^s(p^s) \frac{d\bm v^s}{dt} - \nabla_{\bm x} \cdot \bm \sigma^{s}_{dev} + \nabla_{\bm x} p^s - \rho^s(p^s) \bm b, && \mbox{ in } \Omega_{\bm x}^s.
\end{align}
Given the Dirichlet data $\bm g^s$ and $\bm g^f$ and the boundary tractions $\bm h^s$ and $\bm h^f$, the boundary conditions can be stated as
\begin{align*}
& \bm u^s = \bm g^s, && \mbox{ on } \Gamma_{\bm x}^{g_s}, \\
& \bm v^s = \frac{d \bm g^s}{dt}, && \mbox{ on } \Gamma_{\bm x}^{g_s}, \\
& \bm v^f = \bm g^f, && \mbox{ on } \Gamma_{\bm x}^{g_f}, \\
& (\bm \sigma^s_{dev} - p^s\bm I) \bm n = \bm h^s, && \mbox{ on } \Gamma_{\bm x}^{h_s}, \\
& (\bm \sigma^f_{dev} - p^f\bm I) \bm n = \bm h^f, && \mbox{ on } \Gamma_{\bm x}^{h_f}.
\end{align*}
In the fluid subdomain, the initial conditions are given by a divergence-free velocity field $\bm v^f_0$,
\begin{align}
\bm v^f\left( \bm x, 0 \right) = \bm v^f_0\left(\bm x \right). 
\end{align}
In the solid subdomain, the initial conditions are specified as
\begin{align}
\bm u^s\left(\bm x, 0 \right) = \bm u^s_0\left( \bm x \right), \quad p^s\left( \bm x, 0 \right) = p^s_0\left(\bm x \right), \quad \bm v^s\left(\bm x, 0\right) = \bm v^s_0\left( \bm x \right).
\end{align}

\subsection{Variational multiscale formulation}
In this section, we apply the VMS formulation developed for the general continuum problem in Section \ref{sec:VMS} to the FSI problem. The VMS formulation for the solid problem has been derived in Section \ref{subsec:vms_solids}. For completeness of the FSI formulation, it is stated as follows. Let $\mathcal S_{\bm u_h}^s$, $\mathcal S_{p_h}^s$, and $\mathcal S_{\bm v_h}^s$ denote the finite dimensional trial solution spaces for the solid displacement, pressure, and velocity in the current domain, and let $\mathcal V_{\bm u_h}^s$, $\mathcal V_{p_h}^s$, and $\mathcal V_{\bm v_h}^s$ denote the corresponding test function spaces. We assume that the Dirichlet boundary conditions are built into the definitions of the trial solution spaces. The variational multiscale formulation is stated as follows. Find $\bm y_h^s(t) := \left\lbrace  \bm u_h^s(t), p_h^s(t), \bm v_h^s(t)\right\rbrace^T \in \mathcal S_{\bm u_h}^s \times \mathcal S_{p_h}^s \times \mathcal S_{\bm v_h}^s$ such that
\begin{align}
& \mathbf B^s_{k}\left( \bm w^s_{\bm u_h}; \dot{\bm y}_h^s, \bm y_h^s \right) = \bm 0, && \forall \bm w^s_{\bm u_h} \in \mathcal V_{\bm u_h}^s, \\
& \mathbf B^s_{p}\left( w^s_{p_h}; \dot{\bm y}_h^s, \bm y_h^s \right) = 0, && \forall w^s_{p_h} \in \mathcal V_{p_h}^s, \\
& \mathbf B^s_{m}\left( \bm w^s_{\bm v_h}; \dot{\bm y}_h^s, \bm y_h^s \right) = \bm 0, && \forall \bm w^s_{\bm v_h} \in \mathcal V_{\bm v_h}^s,
\end{align}
where
\begin{align}
& \mathbf B^s_{k}\left( \bm w^s_{\bm u_h}; \dot{\bm y}_h^s, \bm y_h^s \right) := \int_{\Omega_{\bm x}^s} \bm w^s_{\bm u_h} \cdot \left( \frac{d\bm u_h^s}{dt} - \bm v_h^s \right) d\Omega_{\bm x}, \displaybreak[2] \\
& \mathbf B^s_{p}\left( w^s_{p_h}; \dot{\bm y}_h^s, \bm y_h^s \right) := \int_{\Omega_{\bm x}^s} w^s_{p_h} \left( \beta_{\theta}^s(p^s_h) \frac{dp_h^s}{dt} + \nabla_{\bm x} \cdot \bm v_h^s \right) d\Omega_{\bm x} \nonumber \displaybreak[2]  \\
& \hspace{15mm} + \int_{\Omega_{\bm x}^{s\prime}} \nabla_{\bm x} w^s_{p_h} \cdot \bm \tau_M^s \left( \rho^s(p_h^s) \frac{d\bm v_h^s}{dt} - \nabla_{\bm x} \cdot \bm \sigma^s_{dev}(\bm u^s_h) + \nabla_{\bm x} p^s_h - \rho^s(p^s_h) \bm b \right) d\Omega_{\bm x}, \displaybreak[2]  \\
& \mathbf B^s_{m}\left( \bm w^s_{\bm v_h}; \dot{\bm y}_h^s, \bm y_h^s \right) := \int_{\Omega_{\bm x}^s} \bm w^s_{\bm v_h} \cdot \rho^s(p^s_h) \frac{d\bm v_h^s}{dt} + \nabla_{\bm x} \bm w^s_{\bm v_h} : \bm \sigma^s_{dev}(\bm u^s_h) - \nabla_{\bm x} \cdot \bm w^s_{\bm v_h} p_h^s - \bm w^s_{\bm v_h}\cdot \rho^s(p^s_h) \bm b d\Omega_{\bm x} \nonumber \displaybreak[2]  \\
& \hspace{15mm} -\int_{\Gamma_{\bm x}^{h_s}}\bm w_{\bm v_h}^s \cdot \bm h^s d\Gamma_{\bm x} + \int_{\Omega_{\bm x}^{s\prime}} \nabla_{\bm x} \cdot \bm w^s_{\bm v_h} \tau_C^s \left(\beta_{\theta}^s(p^s_h) \frac{dp_h^s}{dt} + \nabla_{\bm x} \cdot \bm v_h^s \right) d\Omega_{\bm x}.
\end{align} 
In the above formulation, $\Omega_{\bm x}^{s\prime}$ represents the union of solid element interiors. The above formulation can be conveniently pulled back to the material frame of reference, as was done in Section \ref{subsec:vms_solids}.

Next, we give the formulation for the motion of the fluid subdomain. We only present the variational formulation for the harmonic extension algorithm. The formulation based on the pseudo-linear-elastic algorithm can be similarly constructed. Let $\mathcal S^m_{\hat{\bm u}_h}$ denote the trial solution space of the mesh displacement $\hat{\bm u}^m_h$, and let $\mathcal V^m_{\hat{\bm u}_h}$ denote the corresponding test function space. The kinematic boundary condition $\hat{\bm u}_h^m = \bm U^s_h$ on $\Gamma_{\bm \chi}^I$ are built into the definition of the space $\mathcal S^m_{\hat{\bm u}_h}$. Find $\hat{\bm u}^m_h \in \mathcal S^m_{\hat{\bm u}_h}$ such that
\begin{align}
\mathbf B^m\left( \bm w^m_h; \hat{\bm u}^m_h \right) = 0, \qquad \forall \bm w^m_h \in \mathcal V^m_{\hat{\bm u}_h},
\end{align}
where
\begin{align}
\mathbf B^m\left( \bm w^m_h; \hat{\bm u}^m_h \right) := \int_{\Omega_{\bm \chi}^f} \nabla_{\bm \chi} \bm w_h^m \cdot \nabla_{\bm \chi}\hat{\bm u}^m_h d\Omega_{\bm \chi}.
\end{align}
With $\hat{\bm u}^m_h$, the mesh velocity in the fluid subdomain can be obtained as
\begin{align*}
\hat{\bm v}_h^m := \left. \frac{\partial \hat{\bm u}^m_h}{\partial t} \right|_{\bm \chi}.
\end{align*}

Lastly, we present the VMS formulation for the fluid problem. Let $\mathcal S^f_{p_h}$ and $\mathcal S^f_{\bm v_h}$ denote the trial solution space of the fluid pressure and velocity, and let $\mathcal V^f_{p_h}$ and $\mathcal V^f_{\bm v_h}$ denote the test function spaces. We assume that functions in the space $\mathcal S^f_{\bm v_h}$ satisfy the Dirichlet boundary condition on $\Gamma_{\bm x}^{g_f}$. The VMS formulation is stated as follows. Find $\bm y_h^f(t):= \left\lbrace p_h^f(t), \bm v_h^f(t) \right\rbrace \in \mathcal S^f_{p_h} \times \mathcal S^f_{\bm v_h}$ such that
\begin{align}
& \mathbf B^f_p\left( w^f_{p_h}; \dot{\bm y}^f_h, \bm y_h^f \right) = 0, && \forall w^f_{p_h} \in \mathcal V^f_{p_h}, \displaybreak[2] \\
& \mathbf B^f_m \left( \bm w^f_{\bm v_h} ;  \dot{\bm y}_h^f, \bm y_h^f \right) = 0, && \forall \bm w^f_{\bm v_h} \in \mathcal V^f_{\bm v_h},
\end{align}
where
\begin{align}
\label{eq:fsi_fluid_residual_based_vms_mass}
& \mathbf B^f_p\left( w^f_{p_h}; \dot{\bm y}_h^f, \bm y_h^f \right) := \int_{\Omega_{\bm x}^f} w^f_{p_h} \nabla_{\bm x} \cdot \bm v_h^f d\Omega_{\bm x}  - \int_{\Omega_{\bm x}^{f}} \nabla_{\bm x} w^f_{p_h} \cdot  \bm v^{f\prime} d\Omega_{\bm x}, \displaybreak[2] \\
\label{eq:fsi_fluid_residual_based_vms_momentum}
& \mathbf B^f_m \left( \bm w^f_{\bm v_h} ;  \dot{\bm y}_h^f, \bm y_h^f ; \hat{\bm v}_h \right) := \int_{\Omega_{\bm x}^f} \bm w^f_{\bm v_h} \cdot \left( \left. \rho_0^f \frac{\partial \bm v_h^f}{\partial t} \right|_{\bm \chi} + \rho_0^f \left( \nabla_{\bm x} \bm v_h^f \right) \left(\bm v_h^f - \hat{\bm v}^m_h \right) \right) d\Omega_{\bm x} \nonumber \displaybreak[2] \\
& \hspace{3mm} - \int_{\Omega_{\bm x}^f} \nabla_{\bm x} \cdot \bm w^f_{\bm v_h} p_h^f d\Omega_{\bm x} + \int_{\Omega_{\bm x}^f} \nabla_{\bm x} \bm w^f_{\bm v_h} : \bm \sigma_{dev}^f(\bm v^f_h) d\Omega_{\bm x} - \int_{\Omega_{\bm x}^f} \bm w^f_{\bm v_h} \cdot \rho^f_0 \bm b d\Omega_{\bm x} - \int_{\Gamma^{h_f}_{\bm x}} \bm w^f_{\bm v_h} \cdot \bm h^f d\Gamma_{\bm x} \nonumber \\
& \hspace{3mm} - \int_{\Omega_{\bm x}^{f\prime}} \nabla_{\bm x} \bm w^f_{\bm v_h} : \left( \rho^f_0 \bm v^{f\prime} \otimes \left( \bm v_h^f - \hat{\bm v}^m_h \right) \right) d\Omega_{\bm x} - \int_{\Omega_{\bm x}^{f\prime}} \nabla_{\bm x} \bm w^f_{\bm v_h} : \left( \rho^f_0 \bm v_h^f \otimes \bm v^{f\prime} \right) d\Omega_{\bm x} \nonumber \displaybreak[2] \\
& \hspace{3mm} - \int_{\Omega_{\bm x}^{f\prime}} \nabla_{\bm x} \bm w^f_{\bm v_h} : \left( \rho^f_0 \bm v^{f\prime} \otimes \bm v^{f\prime} \right) d\Omega_{\bm x} - \int_{\Omega_{\bm x}^{f\prime}} \nabla_{\bm x} \cdot \bm w^f_{\bm v_h} p^{f\prime} d\Omega_{\bm x}, \displaybreak[2] \\
& \bm v^{f\prime} := -\bm \tau_{M}^f \left( \left. \rho_0^f \frac{\partial \bm v_h^f}{\partial t} \right|_{\bm \chi} + \rho_0^f \left( \nabla_{\bm x} \bm v_h^f \right) \left(\bm v_h^f - \hat{\bm v}^m_h \right) + \nabla_{\bm x} p_h^f - \nabla_{\bm x} \cdot \bm \sigma^f_{dev}(\bm v^f_h) - \rho_0^f \bm b \right), \displaybreak[2] \\
& p^{f\prime} := -\tau^f_C \nabla_{\bm x} \cdot \bm v_h^f, \displaybreak[2] \\
& \bm \tau^f_M := \tau^f_M \bm I_{n_d}, \displaybreak[2] \\
\label{eq:fsi_fluid_def_tau_m}
& \tau^f_M := \frac{1}{\rho_0^f}\left( \frac{C_T}{\Delta t^2} + \left(\bm v_h^f - \hat{\bm v}^m_h \right) \cdot \bm G \left(\bm v_h^f - \hat{\bm v}^m_h \right) + C_I \left( \frac{\bar{\mu}}{\rho^f_0} \right)^2 \bm G : \bm G \right)^{-\frac12}, \displaybreak[2] \\
& \tau^f_C := \frac{1}{\tau_M \bm g \cdot \bm g}, \displaybreak[2] \\
& G_{ij} := \sum_{k=1}^{n_d} \frac{\partial \xi_k}{\partial x_i}\frac{\partial \xi_k}{\partial x_j}, \displaybreak[2] \\
& \bm G : \bm G := \sum_{i,j=1}^{n_d} G_{ij} G_{ij}, \displaybreak[2] \\
& g_i := \sum_{j=1}^{n_d} \frac{\partial \xi_j}{\partial x_i}, \displaybreak[2] \\
& \bm g \cdot \bm g := \sum_{i=1}^{n_d} g_i g_i. 
\end{align}
In the above, $\bm \xi = \left\lbrace \xi_i \right\rbrace_{i=1}^{n_d}$ are the coordinates of an element in the parent domain, and $C_I$ is a positive constant derived from an element-wise inverse estimate \cite{Franca1992}. Unless otherwise specified, the value of $C_T$ is taken to be $4$. The value of $C_I$ is independent of the mesh size but relies on the polynomial order of the interpolation basis functions. For linear interpolations, $C_I$ is suggested to be $36$ \cite[p.~65]{Figuero2006}.

\subsection{Formulation for the coupled problem}
Based on the individual subproblems given in the previous section, the semi-discrete FSI problem can be stated as follows. Find $\bm y^s_h(t) := \left\lbrace  \bm u^s_h(t), p^s_h(t), \bm v^s_h(t)\right\rbrace \in \mathcal S_{\bm u_h}^s \times \mathcal S_{p_h}^s \times \mathcal S_{\bm v_h}^s$, $\hat{\bm u}^m_h \in \mathcal S^m_{\hat{\bm u}_h}$, and $\bm y^f_h(t):= \left\lbrace p^f_h(t), \bm v^f_h(t) \right\rbrace \in \mathcal S^f_{p_h} \times \mathcal S^f_{\bm v_h}$ such that
\begin{align}
\label{eq:fsi_coupled_semi_discrete}
& \mathbf B^s_k\left( \bm w^s_{\bm u_h}; \dot{\bm y}^s_h, \bm y^s_h \right) + \mathbf B^s_p\left( w^s_{p_h}; \dot{\bm y}^s_h, \bm y^s_h \right) + \mathbf B^s_m\left( \bm w^s_{\bm v_h}; \dot{\bm y}^s_h, \bm y^s_h \right) + \mathbf B^m\left( \bm w^m_h ; \hat{\bm u}^m_h \right) \nonumber \\
& \hspace{8mm} + \mathbf B^f_p\left( w^f_{p_h}; \dot{\bm y}^f_h, \bm y^f_h \right) + \mathbf B^f_m\left( \bm w^f_{\bm v_h}; \dot{\bm y}^f_h, \bm y^f_h \right) + \mathbf B^f_{Stab}\left( \bm w^f_{\bm v_h}; \dot{\bm y}^f_h, \bm y^f_h \right)  = 0,
\end{align}
wherein the last term is an additional stabilization term defined as
\begin{align}
\label{eq:fsi_coupled_semi_discrete_additional_stabilization}
\mathbf B^f_{Stab}\left( \bm w^f_{\bm v_h}; \dot{\bm y}^f_h, \bm y^f_h \right) :=& \int_{\Omega_{\bm x}^f} \left( \nabla_{\bm x}\bm w^f_{\bm v_h} \bm v^{f\prime} \right) \cdot  \bar{\tau} \left( \nabla_{\bm x}\bm v^f \bm v^{f\prime} \right), \\
\bar{\tau} := \rho_0^f \left( \bm v^{f\prime} \cdot \bm G \bm v^{f\prime} \right)^{-\frac12}. \nonumber
\end{align}
This stabilization term \eqref{eq:fsi_coupled_semi_discrete_additional_stabilization} provides additional stabilization to control oscillations \cite{Hughes1986b,Hughes1986c,Shakib1991}. Due to the choice of $\bar{\tau}$, the term \eqref{eq:fsi_coupled_semi_discrete_additional_stabilization} is proportional to the residual, and therefore the consistency of the semi-discrete formulation \eqref{eq:fsi_coupled_semi_discrete} is maintained. In this work, the whole fluid-solid domain is treated as a single continuum body and is discretized by a single set of mesh with continuous basis functions across the fluid-solid interface. This mesh choice and equal-order interpolations directly imply the satisfaction of the following kinematic constraint relations on the interface,
\begin{align}
\label{eq:fsi_coupled_semi_discrete_kinematic_constraint_variables}
\bm u^s_h = \hat{\bm u}^m_h,  \quad \bm v^s_h = \bm v^f_h.
\end{align}
It also implies the following relations for the test functions,
\begin{align}
\label{eq:fsi_coupled_semi_discrete_kinematic_constraint_test_functions}
\bm w^s_{\bm v_h} = \bm w^f_{\bm v_h}.
\end{align}
Invoking standard variational arguments \cite{Bazilevs2008}, the relation \eqref{eq:fsi_coupled_semi_discrete_kinematic_constraint_test_functions} leads to the weak continuity of the traction across the fluid-solid interface, viz.
\begin{align*}
\int_{\Gamma_{\bm x}^{I}} \bm w^f_{\bm v_h} \cdot \left( \left( \bm \sigma_{dev}^s - p^s \bm I \right) \bm n^{s} + \left( \bm \sigma_{dev}^f - p^f \bm I \right) \bm n^{f} \right)  d\Gamma_{\bm x} = 0.
\end{align*}
The above relation, together with \eqref{eq:fsi_coupled_semi_discrete_kinematic_constraint_variables}, ensures the correct coupling between the fluid and the solid sub-systems.

\subsection{Temporal discretization}
\label{sec:fsi_temporal}
In this section, we apply the generalized-$\alpha$ method to the semi-discrete formulation \eqref{eq:fsi_coupled_semi_discrete} to construct the fully discrete scheme. The time interval $[0,T)$ is divided into a set of $n_{ts}$ subintervals of size $\Delta t_n := t_{n+1} - t_n$ delimited by a discrete time vector $\left\lbrace t_n \right\rbrace_{n=0}^{n_{ts}}$. Let $\bm y^s_n := \left\lbrace \bm u^s_n, p^s_n, \bm v^s_n \right\rbrace$ and $\dot{\bm y}^s_n := \left\lbrace \dot{\bm u}^s_n, \dot{p}^s_n, \dot{\bm v}^s_n \right\rbrace$ denote the solution vector and its first-order time derivative of the solid displacement, pressure, and velocity evaluated at time $t_n$; let $\bm y^f_n := \left\lbrace \hat{\bm u}^m_n, p^f_n, \bm v^f_n \right\rbrace$ and $\dot{\bm y}^f_n := \left\lbrace \hat{\bm v}^m_n, \dot{p}^f_n, \dot{\bm v}^f_n \right\rbrace$ denote the solution vector and its first-order time derivative of the mesh displacement in the fluid subdomain, fluid pressure, and fluid velocity evaluated at time $t_n$. Let $N_A$ denote the basis function on the current configuration and let $\hat{N}_A$ denote the corresponding basis function defined on the reference configuration. They are related by the mapping $\hat{\varphi}$ as $N_A = \hat{N}_A \circ \hat{\varphi}^{-1}$. We denote the residual vectors as
\begin{align*}
& \boldsymbol{\mathrm R}^m\left( \dot{\bm y}_n, \bm y_n \right) := \left\lbrace  \mathbf B^m\left( \hat{N}_A \bm e_i ; \hat{\bm u}_n, \bm u^s_n \right) \right\rbrace , \displaybreak[2] \\
& \boldsymbol{\mathrm R}^f_p\left( \dot{\bm y}_n, \bm y_n \right) := \left\lbrace  \mathbf B^f_p\left( N_A \bm e_i ; \dot{\bm y}^s_n, \bm y^s_n \right) \right\rbrace , \displaybreak[2] \\
& \boldsymbol{\mathrm R}^f_m\left( \dot{\bm y}_n, \bm y_n \right) := \left\lbrace  \mathbf B^f_m\left( N_A \bm e_i ; \dot{\bm y}^s_n, \bm y^s_n \right) \right\rbrace, \displaybreak[2] \\
& \boldsymbol{\mathrm R}^s_k\left( \dot{\bm y}_n, \bm y_n \right) := \left\lbrace  \mathbf B^s_k\left( N_A \bm e_i ; \dot{\bm y}^s_n, \bm y^s_n \right) \right\rbrace ,\displaybreak[2] \\
& \boldsymbol{\mathrm R}^s_p\left( \dot{\bm y}_n, \bm y_n \right) := \left\lbrace  \mathbf B^s_p\left( N_A \bm e_i ; \dot{\bm y}^s_n, \bm y^s_n \right) \right\rbrace , \displaybreak[2]\\
& \boldsymbol{\mathrm R}^s_m\left( \dot{\bm y}_n, \bm y_n \right) := \left\lbrace  \mathbf B^s_m\left( N_A \bm e_i ; \dot{\bm y}^s_n, \bm y^s_n \right) \right\rbrace .
\end{align*}
The fully discrete scheme can be stated as follows. At time step $t_n$, given $\dot{\bm y}_n$ and $\bm y_n$, find $\dot{\bm y}_{n+1}$ and $\bm y_{n+1}$ such that
\begin{align}
\label{eq:gen_alpha_mesh_motion}
& \boldsymbol{\mathrm R}^m\left( \dot{\bm y}_{n+\alpha_m}, \bm y_{n+\alpha_f} \right) = \bm 0, \displaybreak[2] \\
& \boldsymbol{\mathrm R}^f_p\left( \dot{\bm y}_{n+\alpha_m}, \bm y_{n+\alpha_f} \right) = \bm 0, \displaybreak[2] \\
& \boldsymbol{\mathrm R}^f_m\left( \dot{\bm y}_{n+\alpha_m}, \bm y_{n+\alpha_f} \right) = \bm 0, \displaybreak[2] \\
& \boldsymbol{\mathrm R}^s_k\left( \dot{\bm y}_{n+\alpha_m}, \bm y_{n+\alpha_f} \right) = \bm 0, \displaybreak[2] \\
& \boldsymbol{\mathrm R}^s_p\left( \dot{\bm y}_{n+\alpha_m}, \bm y_{n+\alpha_f} \right) = \bm 0, \displaybreak[2] \\
\label{eq:gen_alpha_solid_momentum}
& \boldsymbol{\mathrm R}^s_m\left( \dot{\bm y}_{n+\alpha_m}, \bm y_{n+\alpha_f} \right) = \bm 0, \displaybreak[2] \\
\label{eq:gen_alpha_y_n1_def}
& \bm y_{n+1} = \bm y_n + \Delta t_n \dot{\bm y}_n + \gamma \Delta t_n \left( \dot{\bm y}_{n+1} - \dot{\bm y}_n \right), \displaybreak[2] \\
\label{eq:gen_alpha_y_alpha_m}
& \dot{\bm y}_{n+\alpha_m} = \dot{\bm y}_n + \alpha_m \left( \dot{\bm y}_{n+1} - \dot{\bm y}_n \right), \displaybreak[2] \\
\label{eq:gen_alpha_y_alpha_f}
& \bm y_{n+\alpha_f} = \bm y_n + \alpha_f \left( \bm y_{n+1} - \bm y_n \right).
\end{align}
The choice of the parameters $\alpha_m$, $\alpha_f$ and $\gamma$ is given by \eqref{eq:gen_alpha_1st_order_ODE_parametrization_in_varrho}.

\begin{remark}
\label{remark:fsi_gen_alpha_dilema}
It is shown in \cite{Chung1993} that for second-order ordinary differential equations, the parametrization for $\alpha_m$ in terms of $\varrho_{\infty}$ takes the form
\begin{align*}
\alpha_m = \frac{2-\varrho_{\infty}}{1+\varrho_{\infty}}.
\end{align*}
Only when $\varrho_{\infty} = 0$, the parametrizations of $\alpha_m$ for the first-order and the second-order system coincide. In FSI simulations, oftentimes one solves the solid problem using the pure displacement formulation and solves the fluid problem using the Navier-Stokes equations \cite{Bazilevs2013}. In staggered FSI solvers, one can choose separate, optimal parameters for the generalized-$\alpha$ methods \cite{Dettmer2006}; in monolithic FSI solves, the mismatch of the $\alpha_m$ parameter engenders a dilemma \cite[pp.~119-120]{Bazilevs2013}. In \cite{Bazilevs2008}, the parametrization \eqref{eq:gen_alpha_1st_order_ODE_parametrization_in_varrho} is adopted for the whole FSI system since the Navier-Stokes equations are more challenging than the structural problems. However, a drawback with that approach is the non-optimal high-frequency dissipation in the structural equations. In \cite{Kuhl2003}, the optimal choices of $\alpha_m$ are applied to the fluid and the solid equations separately. That leads to incompatibility of the acceleration at the fluid-solid interface and has been observed to engender failure in FSI simulations \cite{Joosten2010}. In our proposed approach, this issue is properly addressed since the solid mechanics problem is written as a first-order system by introducing the displacement-velocity kinematic equations. This allows one to achieve optimal numerical dissipation in the monolithic FSI solver without violating the kinematic compatibility. 
\end{remark}

\subsection{A predictor multi-corrector algorithm based on the segregated algorithm}
The coupled nonlinear system of equations \eqref{eq:gen_alpha_mesh_motion}-\eqref{eq:gen_alpha_solid_momentum} can be solved monolithically using the Newton-Raphson method \cite{Quarteroni2007}. For FSI problems, a monolithic solver can be implemented using the block-iterative \cite{Tezduyar2007}, quasi-direct \cite{Tezduyar2007a}, or the direct coupling techniques \cite{Bazilevs2008}. In the quasi-direct coupling approach, the fluid and solid equations are treated as a block and the mesh is treated as a block. One solves a block of equations using the most recent unknowns from the other block. The quasi-direct coupling approach enjoys a good balance of robustness and computational cost and has been recommended for general FSI problems \cite{Bazilevs2013a}. In this work, we invoke the quasi-direct coupling methodology and the segregated algorithm for the solid equations developed in Section \ref{sec:solid_segregated_algorithm}. This leads to a novel algorithm for the nonlinear solver in FSI problems. This algorithm is stated as follows.

\noindent \textbf{Predictor stage}: 
\begin{enumerate}
\item Set
\begin{align*}
\bm y_{n+1, (0)} &= \bm y_{n}, \\
\dot{\bm y}_{n+1, (0)} &= \frac{\gamma - 1}{\gamma} \dot{\bm y}_{n}.
\end{align*}
\item Evaluate the solutions at the intermediate stage as
\begin{align*}
\dot{\bm y}_{n+\alpha_m, (1)} &= \dot{\bm y}_n + \alpha_m \left( \dot{\bm y}_{n+1,(0)} - \dot{\bm y}_n \right), \\
\bm y_{n+\alpha_f, (1)} &= \bm y_n + \alpha_f \left( \bm y_{n+1,(0)} - \bm y_n \right).
\end{align*}
\end{enumerate}

\noindent \textbf{Multi-corrector stage}:
Repeat the following steps \(i =1,2, \dots, i_{max}\):
\begin{enumerate}
\item Assemble the residual vectors of the nonlinear system using the above intermediate stage solutions:
\begin{align*}
\boldsymbol{\mathrm R}^f_{p,(i)} &:= \boldsymbol{\mathrm R}^f_p\left( \dot{\bm y}_{n+\alpha_m,(i)}, \bm y_{n+\alpha_f, (i)} \right), \displaybreak[2] \\
\boldsymbol{\mathrm R}^f_{m,(i)} &:= \boldsymbol{\mathrm R}^f_m\left( \dot{\bm y}_{n+\alpha_m, (i)}, \bm y_{n+\alpha_f, (i)} \right), \displaybreak[2] \\
\boldsymbol{\mathrm R}^s_{p,(i)} &:= \boldsymbol{\mathrm R}^s_p\left( \dot{\bm y}_{n+\alpha_m, (i)}, \bm y_{n+\alpha_f, (i)} \right), \displaybreak[2] \\
\boldsymbol{\mathrm R}^s_{m,(i)} &:= \boldsymbol{\mathrm R}^s_m\left( \dot{\bm y}_{n+\alpha_m, (i)}, \bm y_{n+\alpha_f, (i)} \right).
\end{align*}

\item Let $\| \left(\boldsymbol{\mathrm R}^f_{p,(i)}; \boldsymbol{\mathrm R}^f_{m,(i)}, \boldsymbol{\mathrm R}^s_{p,(i)}; \boldsymbol{\mathrm R}^s_{m,(i)} \right) \|_{l^2}$ denote the $l^2$-norm of the residual vector. If either one of the stopping criteria
\begin{align*}
& \frac{\| \left(\boldsymbol{\mathrm R}^f_{p,(i)}; \boldsymbol{\mathrm R}^f_{m,(i)}, \boldsymbol{\mathrm R}^s_{p,(i)}; \boldsymbol{\mathrm R}^s_{m,(i)} \right) \|_{l^2}}{\| \left(\boldsymbol{\mathrm R}^f_{p,(0)}; \boldsymbol{\mathrm R}^f_{m,(0)}, \boldsymbol{\mathrm R}^s_{p,(0)}; \boldsymbol{\mathrm R}^s_{m,(0)} \right) \|_{l^2}} \leq tol_R, \\
& \| \left(\boldsymbol{\mathrm R}^f_{p,(i)}; \boldsymbol{\mathrm R}^f_{m,(i)}, \boldsymbol{\mathrm R}^s_{p,(i)}; \boldsymbol{\mathrm R}^s_{m,(i)} \right) \|_{l^2} \leq tol_A
\end{align*}
is satisfied for prescribed tolerances  $tol_R$, $tol_A$, set the solution vector at time step $t_{n+1}$ as $\bm y_{n+1} = \bm y_{n+1, (i-1)}$ and exit the multi-corrector stage; otherwise, continue to step 3. 

\item Assemble the tangent matrices and solve the linear system of equations:
\label{step:fsi_implict_solve}
\begin{align}
\label{eq:fsi_coupled_linear_equation_fluid_p}
\frac{\partial \boldsymbol{\mathrm R}^f_p}{\partial \dot{\bm v}^f_{n+1}} \Delta \dot{\bm v}^f_{n+1,(i)} + \frac{\partial \boldsymbol{\mathrm R}^f_p}{\partial \dot{p}^f_{n+1}} \Delta \dot{p}^f_{n+1, (i)} &= -\boldsymbol{\mathrm R}^f_{p,(i)} ,  \displaybreak[2] \\
\frac{\partial \boldsymbol{\mathrm R}^f_m}{\partial \dot{\bm v}^f_{n+1}}\Delta \dot{\bm v}^f_{n+1,(i)} + \frac{\partial \boldsymbol{\mathrm R}^f_m}{\partial \dot{p}^f_{n+1}} \Delta \dot{p}^f_{n+1, (i)} &= -\boldsymbol{\mathrm R}^f_{m,(i)} , \displaybreak[2] \\
\frac{\partial \boldsymbol{\mathrm R}^s_p}{\partial \dot{\bm v}^s_{n+1}} \Delta \dot{\bm v}^s_{n+1,(i)} + \frac{\partial \boldsymbol{\mathrm R}^s_p}{\partial \dot{p}^s_{n+1}} \Delta \dot{p}^s_{n+1, (i)} &= -\boldsymbol{\mathrm R}^s_{p,(i)} ,  \displaybreak[2] \\
\label{eq:fsi_coupled_linear_equation_solid_m}
\frac{\partial \boldsymbol{\mathrm R}^s_m}{\partial \dot{\bm v}^s_{n+1}}\Delta \dot{\bm v}^s_{n+1,(i)} +  \frac{\partial \boldsymbol{\mathrm R}^s_m}{\partial \dot{p}^s_{n+1}} \Delta \dot{p}^s_{n+1, (i)} &= -\boldsymbol{\mathrm R}^s_{m,(i)}.
\end{align}

\item Update the iterates as
\begin{align*}
\dot{\bm v}^f_{n+1,(i)} &= \dot{\bm v}^f_{n+1,(i-1)} + \Delta \dot{\bm v}^f_{n+1,(i)}, \displaybreak[2] \\
\bm v^f_{n+1,(i)} &= \bm v^f_{n+1,(i-1)} + \gamma \Delta t_n \Delta \dot{\bm v}^f_{n+1,(i)}, \displaybreak[2] \\
\dot{p}^f_{n+1, (i)} &= \dot{p}^f_{n+1,(i-1)} + \Delta \dot{p}^f_{n+1,(i)}, \displaybreak[2] \\
p^f_{n+1,(i)} &= p^f_{n+1,(i-1)} + \gamma \Delta t_n \Delta \dot{p}^f_{n+1,(i)}, \displaybreak[2] \\
\dot{\bm v}^s_{n+1,(i)} &= \dot{\bm v}^s_{n+1,(i-1)} + \Delta \dot{\bm v}^s_{n+1,(i)}, \displaybreak[2] \\
\bm v^s_{n+1,(i)} &= \bm v^s_{n+1,(i-1)} + \gamma \Delta t_n \Delta \dot{\bm v}^s_{n+1,(i)}, \displaybreak[2] \\
\dot{p}^s_{n+1, (i)} &= \dot{p}^s_{n+1,(i-1)} + \Delta \dot{p}^s_{n+1,(i)}, \displaybreak[2] \\
p^s_{n+1,(i)} &= p^s_{n+1,(i-1)} + \gamma \Delta t_n \Delta \dot{p}^s_{n+1,(i)}.
\end{align*}

\item Obtain $\Delta \dot{\bm u}^s_{n+1,(i)}$ from $\Delta \dot{\bm v}^s_{n+1,(i)}$ using the relation \eqref{eq:newton_linear_system_smaller_eqn_2}.

\item Update the solid displacement as
\begin{align*}
\dot{\bm u}^s_{n+1,(i)} &= \dot{\bm u}^s_{n+1,(i-1)} + \Delta \dot{\bm u}^s_{n+1,(i)}, \displaybreak[2] \\
\bm u^s_{n+1,(i)} &= \bm u^s_{n+1,(i-1)} + \gamma \Delta t_n \Delta \dot{\bm u}^s_{n+1,(i)}.
\end{align*}

\item Solve the mesh motion equation 
\begin{align}
\label{eq:fsi_coupled_linear_equation_mesh_motion_in_algorithm}
\boldsymbol{\mathrm R}^m\left( \dot{\bm y}_{n+\alpha_m, (i)}, \bm y_{n+\alpha_f, (i)} \right) = 0.
\end{align}

\item Obtain the mesh velocity by 
\begin{align*}
\hat{\bm v}^m_{n+1,(i)} = \frac{\hat{\bm u}^m_{n+1,(i)} - \hat{\bm u}^m_{n}}{\gamma \Delta t_n} + \frac{\gamma - 1}{\gamma } \hat{\bm v}^m_{n}.
\end{align*}

\item Evaluate the solution vectors at the intermediate stage as
\begin{align*}
\dot{\bm y}_{n+\alpha_m, (i+1)} &= \dot{\bm y}_n + \alpha_m \left( \dot{\bm y}_{n+1,(i)} - \dot{\bm y}_n \right), \\
\bm y_{n+\alpha_f, (i+1)} &= \bm y_n + \alpha_f \left( \bm y_{n+1,(i)} - \bm y_n \right).
\end{align*}

\end{enumerate}
There are $2n_{d}+1$ degrees of freedom associated with each computational node in both subdomains; in step \ref{step:fsi_implict_solve}, $n_d + 1$ degrees of freedom need to be solved by the implicit solver. In the implicit solver, the mass and linear momentum equations are discretized by VMS in space and the generalized-$\alpha$ method in time. In the assembly routine, one only needs to call the proper constitutive relations for the material. In doing so, one can implement the solution vector and the tangent matrix problem over the whole continuum body in a unified approach, which may simplify the data structure management in the numerical implementation. The linear system of equations \eqref{eq:fsi_coupled_linear_equation_fluid_p}-\eqref{eq:fsi_coupled_linear_equation_solid_m} can be solved by means of the Generalized Minimal Residual (GMRES) algorithm \cite{Saad1986} with an additive Schwarz preconditioner; the linear system of equations for the mesh motion \eqref{eq:fsi_coupled_linear_equation_mesh_motion_in_algorithm} is symmetric positive definite and hence can be solved by means of the conjugate gradient method. The PETSc package \cite{Balay1997} is adopted to provide an interface for a wide range of solver options.


\section{Benchmark computations}
\label{sec:benchmark}
In this section, we first use manufactured solutions to assess the convergence behavior of the algorithm developed in Section \ref{sec:numerical_solid} for both compressible and incompressible materials. Following that, a classical benchmark problem for incompressible solids is studied, and the results for the VMS formulation with $P_1/P_1$ and $Q_1/Q_1$ elements are compared with the results of the $\bar{F}$-projection method. In the last, two FSI benchmark problems are studied to demonstrate the effectiveness of the new FSI formulation proposed in Section \ref{sec:numerical_FSI}.
 
\subsection{Manufactured solution for compressible hyperelasticity}
In the first example, the displacement field is given in a closed-form formulation
\begin{align}
\label{eq:3d_manu_solution_disp}
\bm U = \frac{1}{T_0^2} t^2 \begin{bmatrix}
X \cos(\beta_1 Z) - Y \sin(\beta_1 Z) - X \\
X \sin(\beta_1 Z) + Y \cos(\beta_1 Z) - Y \\
0
\end{bmatrix}
\end{align}
for a unit cube (1 cm $\times$ 1 cm $\times$ 1 cm). In the prescribed displacement field, $T_0 = 1.0\times 10^{-3}$ s, $\beta_1 = 10^{-3}\pi $ rad/cm. This manufactured solution describes rotation of the cube with the bottom surface $Z=0$ clamped. A static version of this manufactured solution has been used in \cite{Leger2016} for code verification purposes. The material model is chosen as the Neo-Hookean model with the volumetric energy given by \eqref{eq:psi_vol_ST91}. Its Gibbs free energy is
\begin{align*}
G\left( \tilde{\bm C}, p \right) = \frac{\mu^s}{2\rho_0}  \left( \textup{tr}\tilde{\bm C} - 3 \right) + \frac{p \sqrt{p^2+\kappa^2}-p^2}{2\kappa \rho_0} - \frac{\kappa}{2\rho_0}\ln \left( \frac{\sqrt{p^2+\kappa^2}-p}{\kappa} \right).
\end{align*}
The material parameters are chosen as $\mu^s = 3.70$ MPa, $\kappa = 1.11 \times 10^1$ MPa, and $\rho_0 = 1.0 \times 10^3$ $\textup{kg/m}^3$. The corresponding Poisson's ratio is 0.35. The displacement and velocity on the bottom surface are fixed to be zero, traction boundary conditions are applied on the rest surfaces. The analytic form of the traction as well as the forcing terms can be obtained from the analytic formulation of the displacement field \eqref{eq:3d_manu_solution_disp}. In particular, to corroborate the proposed formulation \eqref{eq:isothermal_governing_disp_velo}-\eqref{eq:isothermal_governing_eqn_momentum}, the forcing term $\bm B$ is obtained by inserting \eqref{eq:3d_manu_solution_disp} into the pure displacement formulation. To obtain a mesh with uniform element size, we generate a hexahedral mesh first and decompose each hexahedron into six tetrahedrons with equal size, as is shown in Figure \ref{fig:decomp_cube_into_six_tets}.  The problem is simulated up to $T=5.0\times 10^{-4}$ s with a fixed time step size $\Delta t = 5.0\times 10^{-6}$ s. The stabilization parameters are chosen with $c_m = 0.1$ and $c_c = 0.1$. The relative errors of the displacement, velocity, pressure, deformation gradient, and the deviatoric part of the Cauchy stress in the $\mathcal L^2$-norm are plotted in Figure \ref{fig:manu_sol_3d_conv_plot}. Note immediately that the $\mathcal L^2$-norms of the errors in displacement and velocity converge quadratically; the asymptotic convergence rate for the pressure is 1.8; the deformation gradient and the deviatoric part of the Cauchy stress converge linearly. In \cite{Masud2005}, a similar sub-optimal convergence rate for the pressure field has been observed, using a different stabilized formulation.
\begin{figure}
	\begin{center}
	\begin{tabular}{c}
\includegraphics[angle=0, trim=0 0 700 0, clip=true, scale = 0.20]{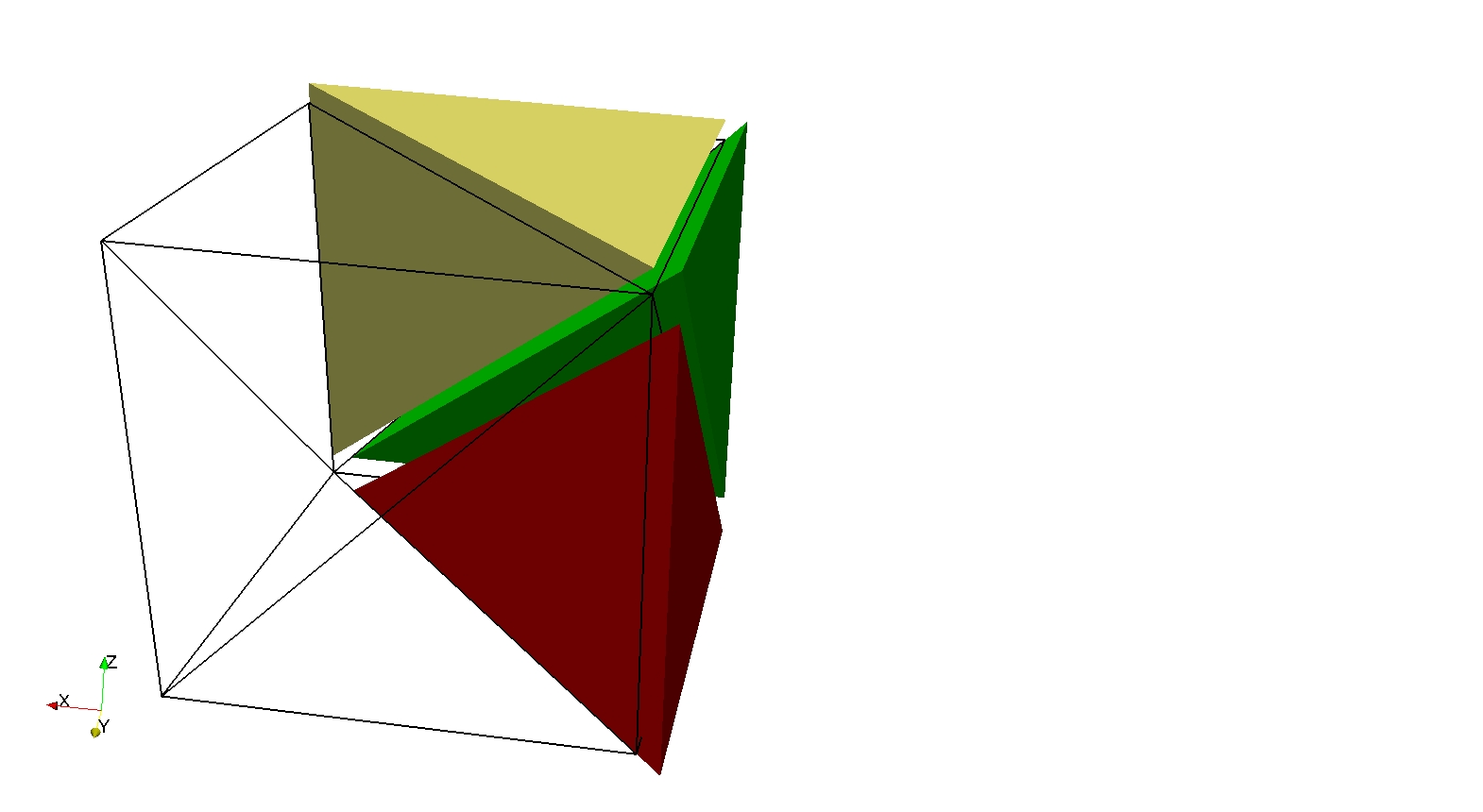}
\end{tabular}
\caption{Decomposition of a unit cube into six tetrahedrons with equal size. The cube is first split into two triangular prisms. Each prism is further divided into three tetrahedrons. The diameter of the tetrahedrons' circumscribing sphere is $\sqrt{3}$ times the edge length of the cube. This decomposition allows one to create a structured tetrahedral mesh.} 
\label{fig:decomp_cube_into_six_tets}
\end{center}
\end{figure}

\begin{figure}
	\begin{center}
	\begin{tabular}{c}
\includegraphics[angle=0, trim=50 80 50 50, clip=true, scale = 0.55]{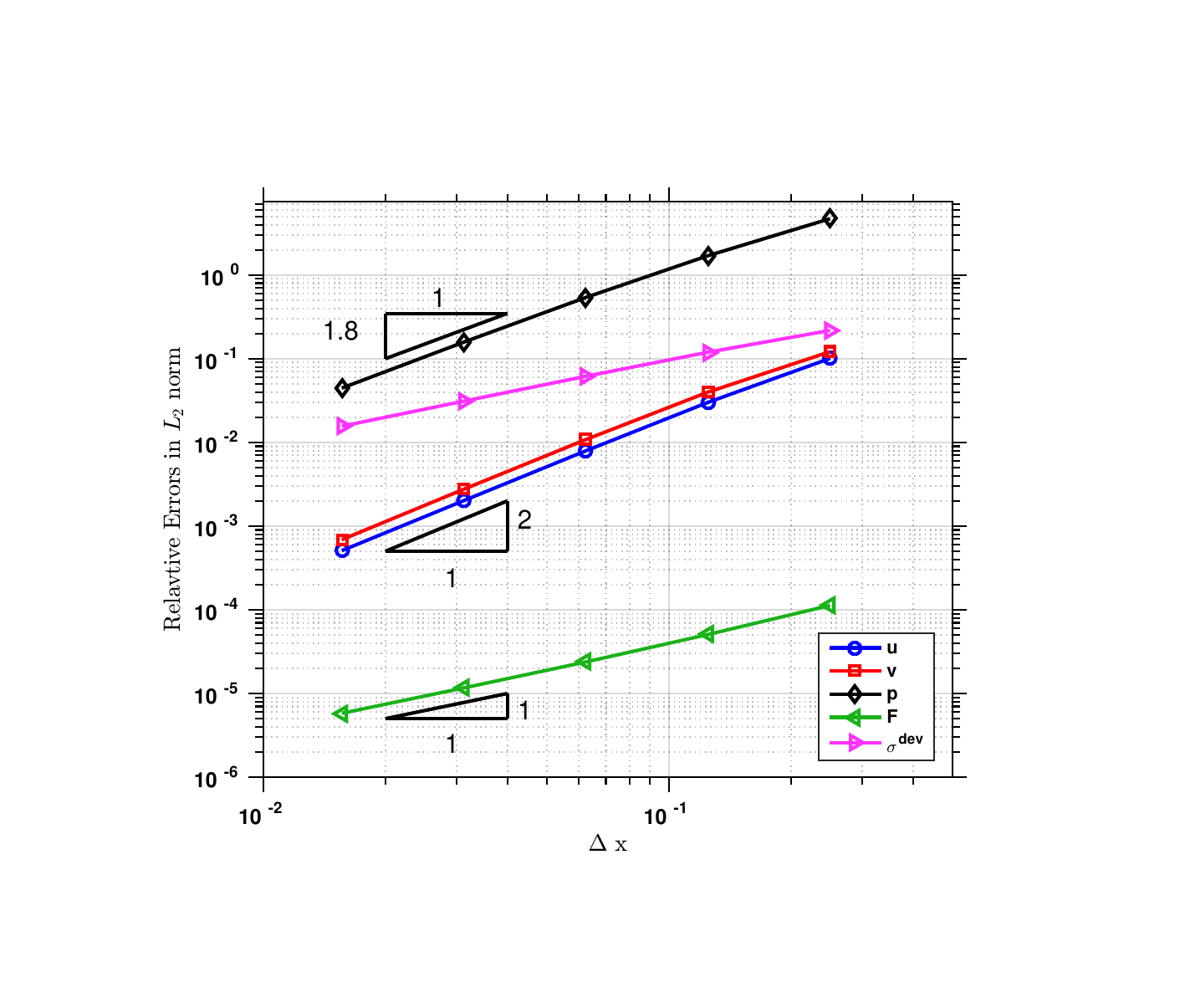}
\end{tabular}
\end{center}
\caption{Three-dimensional manufactured solution for compressible hyperelasticity: Spatial convergence rates of the displacement, velocity, pressure, deformation gradient, and the deviatoric part of the Cauchy stress.} 
\label{fig:manu_sol_3d_conv_plot}
\end{figure}

\begin{figure}
	\begin{center}
	\begin{tabular}{c}
\includegraphics[angle=0, trim=50 80 50 50, clip=true, scale = 0.55]{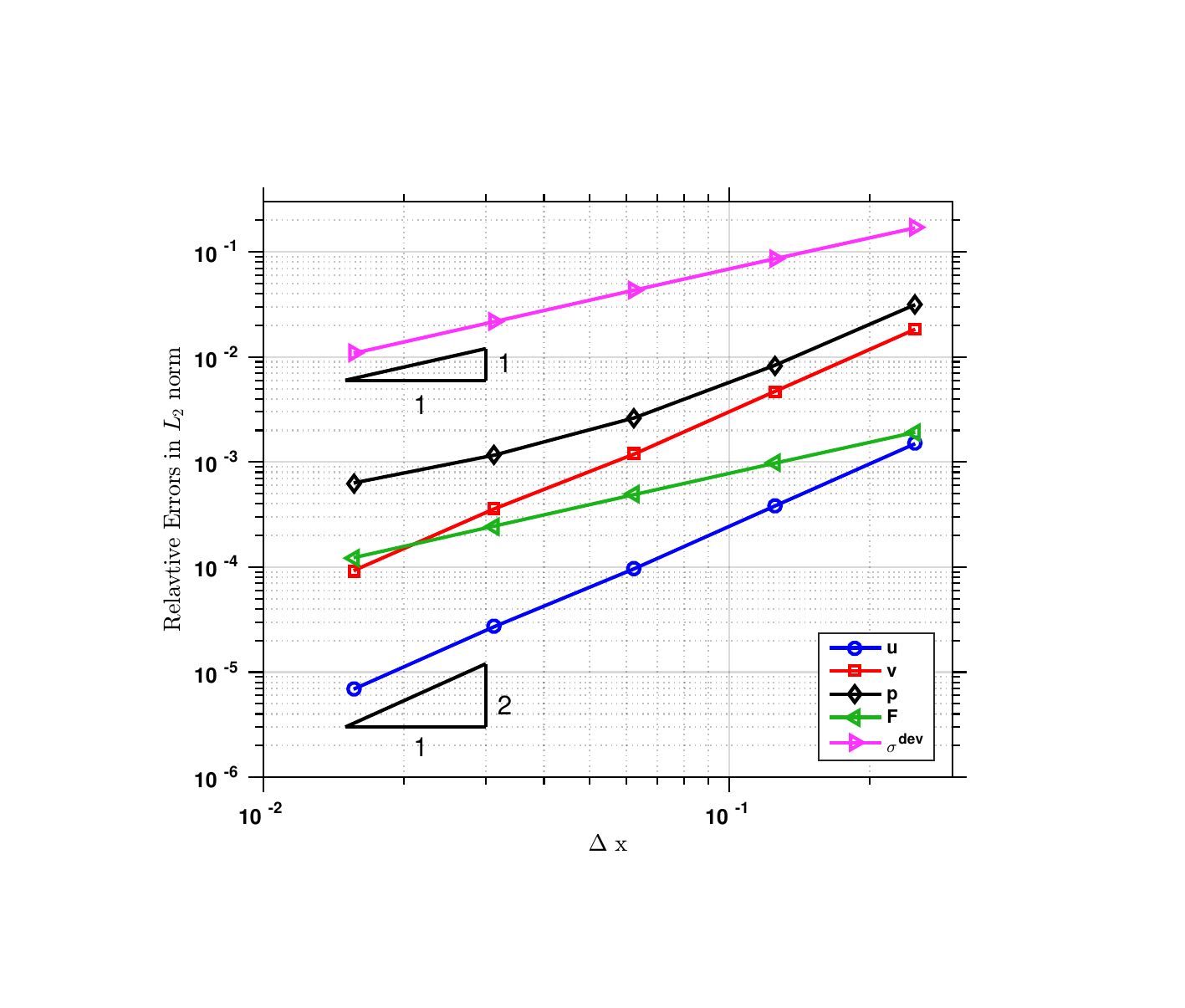}
\end{tabular}
\end{center}
\caption{Three-dimensional manufactured solution for incompressible hyperelasticity: Spatial convergence rates of the displacement, velocity, pressure, deformation gradient, and the deviatoric part of the Cauchy stress.} 
\label{fig:manu_sol_3d_incompressible_conv_plot}
\end{figure}

\subsection{Manufactured solution for fully incompressible hyperelasticity}
In the second example, we consider a fully incompressible Neo-Hookean material model. Its Gibbs free energy is 
\begin{align}
G\left( \tilde{\bm C}, p \right) = \frac{\mu^s}{2\rho_0} \left( \textup{tr}\tilde{\bm C} -3 \right) + \frac{p}{\rho_0} = \frac{\mu^s}{2\rho_0} \left( \textup{tr}\bm C -3 \right) + \frac{p}{\rho_0}.
\end{align}
The geometrical domain is again a unit cube (1 cm $\times$ 1 cm $\times$ 1 cm). The shear modulus is chosen as $100$ KPa; the density $\rho_0$ is chosen to be $1.0\times10^3$ $\textup{kg/m}^3$. 
The analytic forms of the displacement and the pressure fields are given as follows.
\begin{align}
\label{eq:manu_2_exa_disp}
\bm U &= \frac{L_0}{T_0^2} t^2 \begin{bmatrix}
\sin(\gamma_2 Y) \sin(\gamma_2 Z) \\
0 \\
0
\end{bmatrix}, \\
\label{eq:manu_2_exa_pres}
P &= \frac{M_0}{L_0T_0^4}t^2 \sin(\beta_2 X) \sin(\beta_2 Y) \sin(\beta_2 Z).
\end{align}
In the above, $M_0=1.0\times 10^{-3}$ kg, $L_0 = 1.0\times 10^{-2}$m, $T_0 = 1.0\times 10^{-3}$ s, $\beta_2 = 0.2\pi$ rad/cm, and $\gamma_2 = 0.1\pi$ rad/cm. The displacement and velocity on the bottom surface are fixed to be zero, and traction boundary conditions are applied on the rest surfaces. The analytic form of the tractions over the boundaries and the forcing term can be obtained from the given displacement and pressure fields \eqref{eq:manu_2_exa_disp}-\eqref{eq:manu_2_exa_pres}. Similar to the previous example, the forcing term is obtained by inserting \eqref{eq:manu_2_exa_disp} into the pure displacement formulation. The problems are computed up to $T=5.0 \times 10^{-4}$ s with a fixed time step size $\Delta t = 2.5\times 10^{-6}$ s. The stabilization parameters are chosen as $c_m = 0.1$ and $c_c = 0.1$. The relative errors of the displacement, velocity, pressure, deformation gradient, and the deviatoric part of the Cauchy stress in the $\mathcal L^2$-norm are plotted in Figure \ref{fig:manu_sol_3d_incompressible_conv_plot}. For the fully incompressible material, the convergence rates for the displacement and velocity are of order two; the convergence rates for the deformation gradient and the deviatoric part of the Cauchy stress are first-order; the convergence rate for the pressure becomes first-order for the fully incompressible material. Although there is no convergence proof for this nonlinear problem, the convergence rate of the pressure in the $\mathcal L^2$-norm for linear elasticity is proved to be first-order \cite{Franca1991,Hughes1986a}. Therefore, the first-order convergence rate for the pressure field is indeed expected.

\begin{figure}
	\begin{center}
	\begin{tabular}{c}
\includegraphics[angle=0, trim=60 10 10 30, clip=true, scale = 0.4]{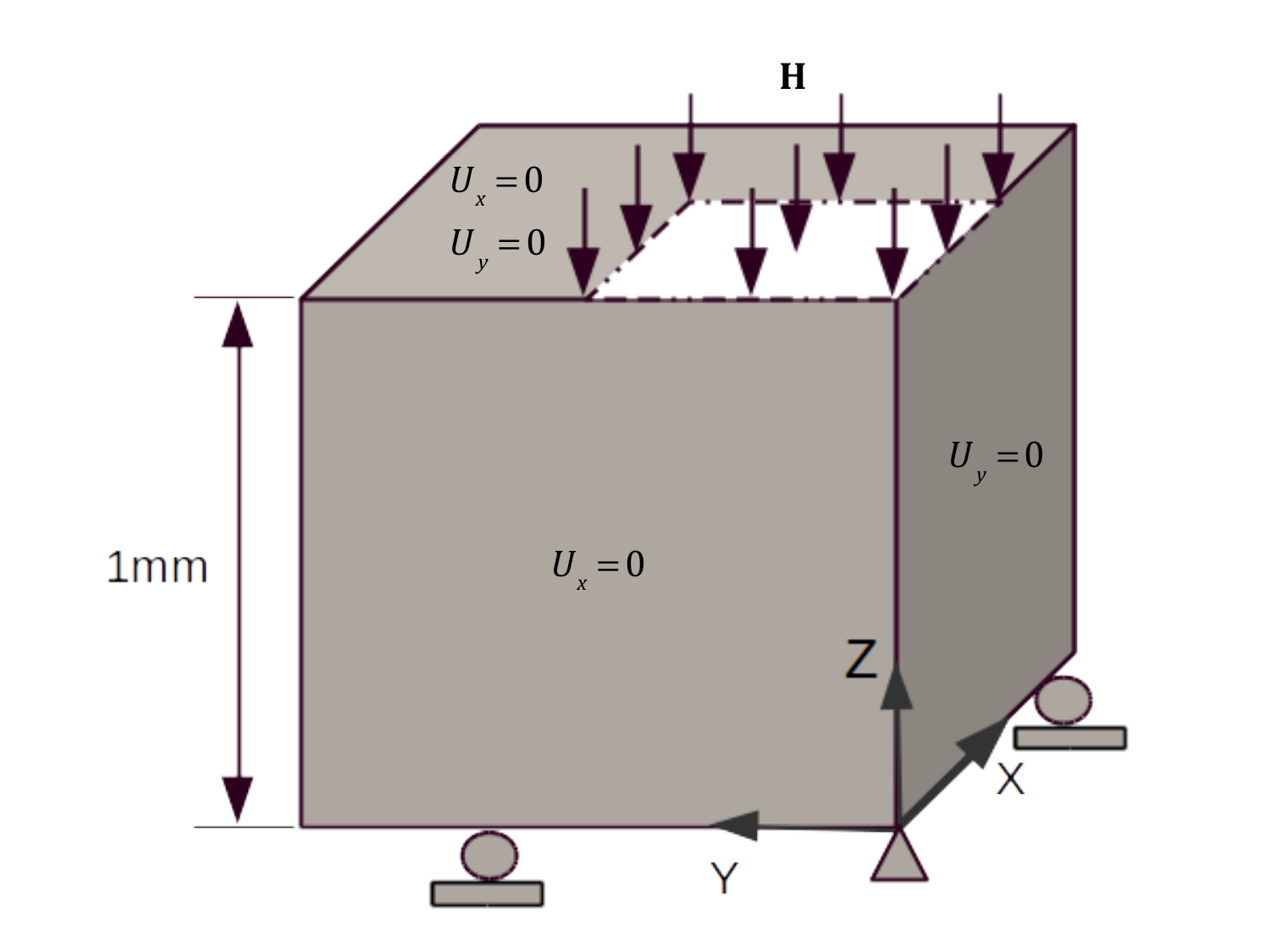}
\end{tabular}
\caption{Three-dimensional compression of a block: geometry of the referential configuration and the boundary conditions.} 
\label{fig:block_compression_setting}
\end{center}
\end{figure}

\begin{figure}
	\begin{center}
	\begin{tabular}{c}
\includegraphics[angle=0, trim=100 80 100 80, clip=true, scale = 0.5]{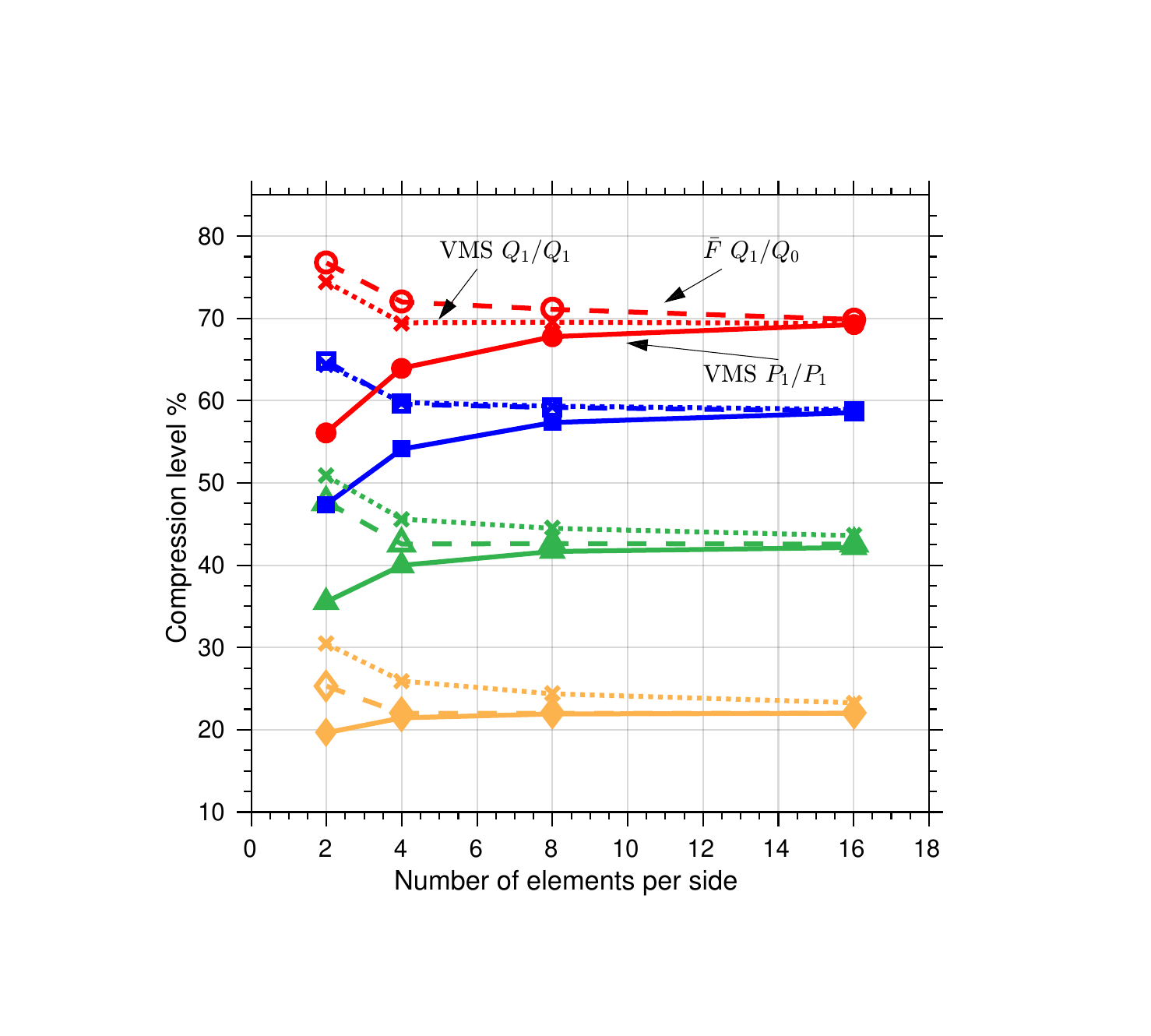}
\end{tabular}
\caption{Three-dimensional compression level in \% versus the number of elements per side: Comparison of results for $\Lambda = 20$ (in yellow color), $\Lambda = 40$ (in green color), $\Lambda = 60$ (in blue color), and $\Lambda = 80$ (in red color). The solid lines represent solutions from the VMS formulation with P1/P1 element; the dashed line represent solutions from the $\bar{F}$-projection method with Q1/Q0 element; the dotted line represent solutions from the VMS formulation with Q1/Q1 element.} 
\label{fig:block_compression_compare}
\end{center}
\end{figure}

\subsection{Nearly incompressible block under compression }
This example was initially proposed in \cite{Reese2000} as a benchmark problem for nearly incompressible solids. The original problem was proposed as a quasi-static problem. In this work, we pose the problem in the dynamic setting and adopt the Neo-Hookean model described by the following Gibbs free energy function.
\begin{align}
\label{eq:3d_block_compression_material_model}
G\left( \tilde{\bm C}, p \right) = \frac{\mu^s}{2\rho_0}  \left( \textup{tr}\tilde{\bm C} - 3 \right)  + \frac{p \sqrt{p^2+\kappa^2}-p^2}{2\kappa \rho_0} - \frac{\kappa}{2\rho_0}\ln \left( \frac{\sqrt{p^2+\kappa^2}-p}{\kappa} \right).
\end{align}
Notice that the material model we used here is slightly different from the one used in the original paper \cite{Reese2000} because we demand that the Gibbs free energy takes the decoupled form \eqref{eq:gibbs_free_energy_additive_split} and the isochoric part of the energy is a function of $\tilde{\bm C}$. Following \cite{Reese2000}, the material parameters are chosen as $\mu^s = 80.194$ MPa, $\kappa = 400889.806$ MPa, and $\rho_0 = 1.0 \times 10^3$ $\textup{kg/m}^3$. The corresponding Poisson's ratio is 0.4999. The problem setting is illustrated in Figure \ref{fig:block_compression_setting}. Symmetry boundary conditions are applied on the $X=Y=Z=0$ planes. A `dead' surface load $\bm H$ is applied on a quarter portion of the top of the block, and the load assumes the negative z-direction in the reference configuration. The magnitude of $\bm H$ is measured by $|\bm H| = \Lambda H_0$, and the reference value $H_0 = 4$ MPa. We calculate the compression level of the upper center point with $\Lambda = 20$, $40$, $60$, and $80$. The load force is gradually increased as a linear function of time. The problem is integrated in time with $200$ time steps. We performed simulations using the $\bar{F}$-projection method with $Q_1/Q_0$ element and the VMS formulation with $Q_1/Q_1$ element and $P_1/P_1$ element. The stabilization parameters are chosen as $c_m = 0.1$ and $c_c = 0.1$. In Figure \ref{fig:block_compression_compare}, the compression levels of the upper center point with mesh refinement for different loading ratios and different numerical methods are illustrated. The VMS formulation with $P_1/P_1$ element tends to give very stiff response with coarse meshes. In contrast, the $\bar{F}$-projection method with $Q_1/Q_0$ element and the VMS formulation with $Q_1/Q_1$ element tend to give very soft response with coarse meshes. With mesh refinement, convergent results are obtained for all three different methods.

\subsection{Flow over an elastic cantilever}
The two-dimensional flow-induced oscillation of an elastic cantilever attached to a fixed square block was initially designed in \cite{Wall2000} as a benchmark problem for FSI algorithms. It has been used extensively in the literature to assess the quality of FSI algorithms \cite{Bazilevs2008,Dettmer2006,Wood2010}. In this work, the original two-dimensional problem is extended to a three-dimensional problem by extruding the original problem in the third direction \cite{Wood2010}. On the inflow surface, a uniform flow in the x-direction is imposed with magnitude $51.3$ $\textup{cm/s}$; on the outflow surface, a zero traction boundary condition is applied; on the four lateral boundary surfaces, slip boundary conditions (zero normal velocity and zero tangential traction) are applied. The geometry and boundary conditions are illustrated in Figure \ref{fig:fsi_beam_benchmark_geometry}. The fluid density and dynamic shear viscosity are set to be $\rho_0^f = 1.18\times 10^{-3}$ $\textup{g/cm}^3$ and  $\bar{\mu} = 1.82 \times 10^{-4}$ poise, respectively. The solid is modeled as the Neo-Hookean material with the volumetric free energy given by \eqref{eq:psi_vol_M94},
\begin{align*}
G\left( \tilde{\bm C}, p \right) = \frac{\mu^s }{2\rho_0} \left( \textup{tr}\tilde{\bm C} - 3 \right) -\frac{\kappa}{\rho_0} \ln \left(\frac{\kappa}{p+\kappa} \right).
\end{align*}
The solid referential density $\rho^s_0$ is $0.1$ $\textup{g/cm}^3$, the shear modulus $\mu^s$ and the Poisson's ratio are $9.2593 \times 10^5$ $\textup{dyn/cm}^2$ and $0.35$, respectively.

\begin{figure}
	\begin{center}
	\begin{tabular}{c}
\includegraphics[angle=0, trim=0 20 0 60, clip=true, scale = 0.5]{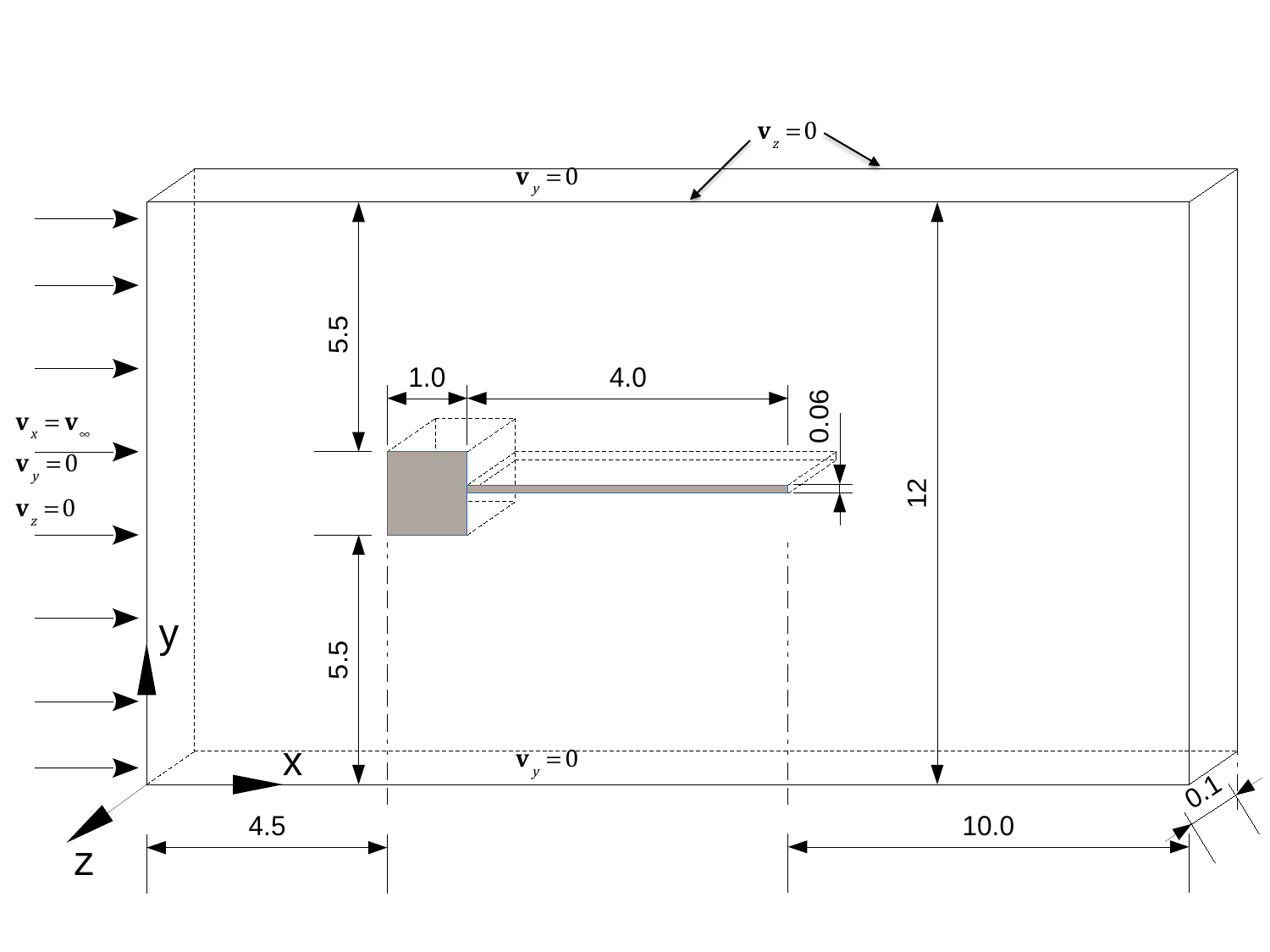} 
\end{tabular}
\caption{Flow over an elastic cantilever: geometry setting and boundary conditions.} 
\label{fig:fsi_beam_benchmark_geometry}
\end{center}
\end{figure}

\begin{figure}
	\begin{center}
	\begin{tabular}{cc}
\includegraphics[angle=0, trim=120 60 60 60, clip=true, scale = 0.11]{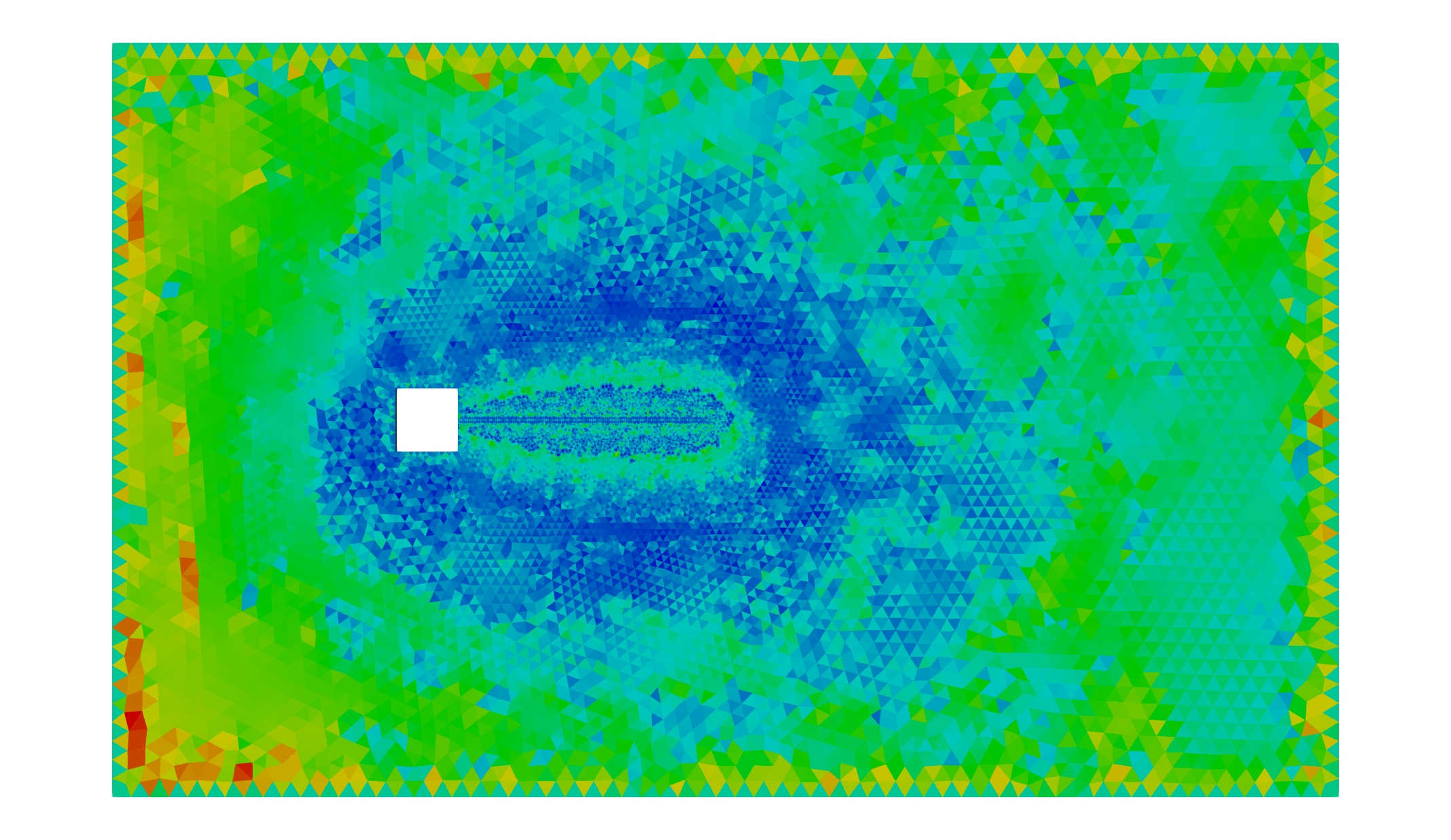} &
\includegraphics[angle=0, trim=120 60 60 60, clip=true, scale = 0.11]{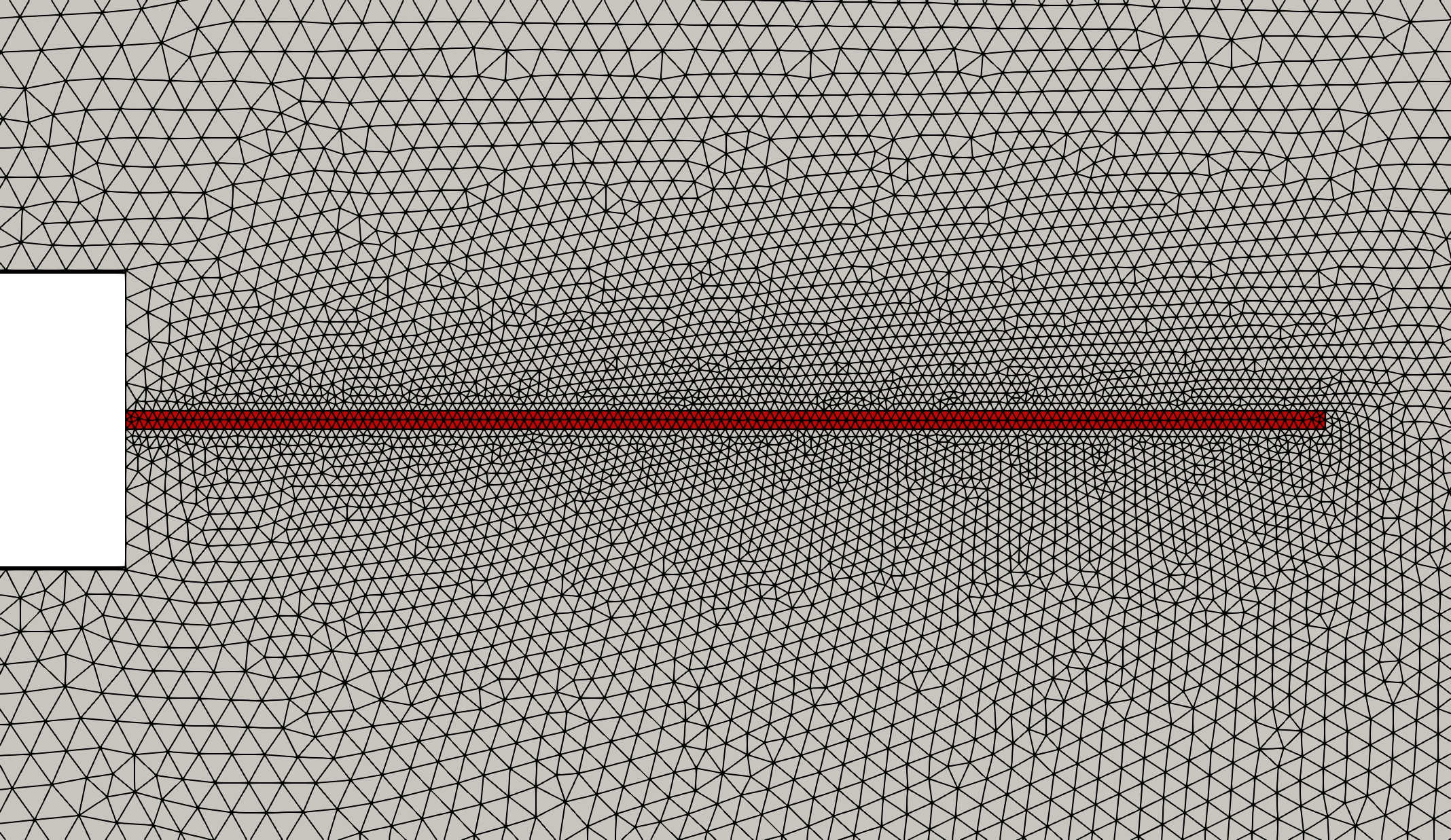} \\[1mm]
\includegraphics[angle=0, trim=500 1060 500 0, clip=true, scale = 0.12]{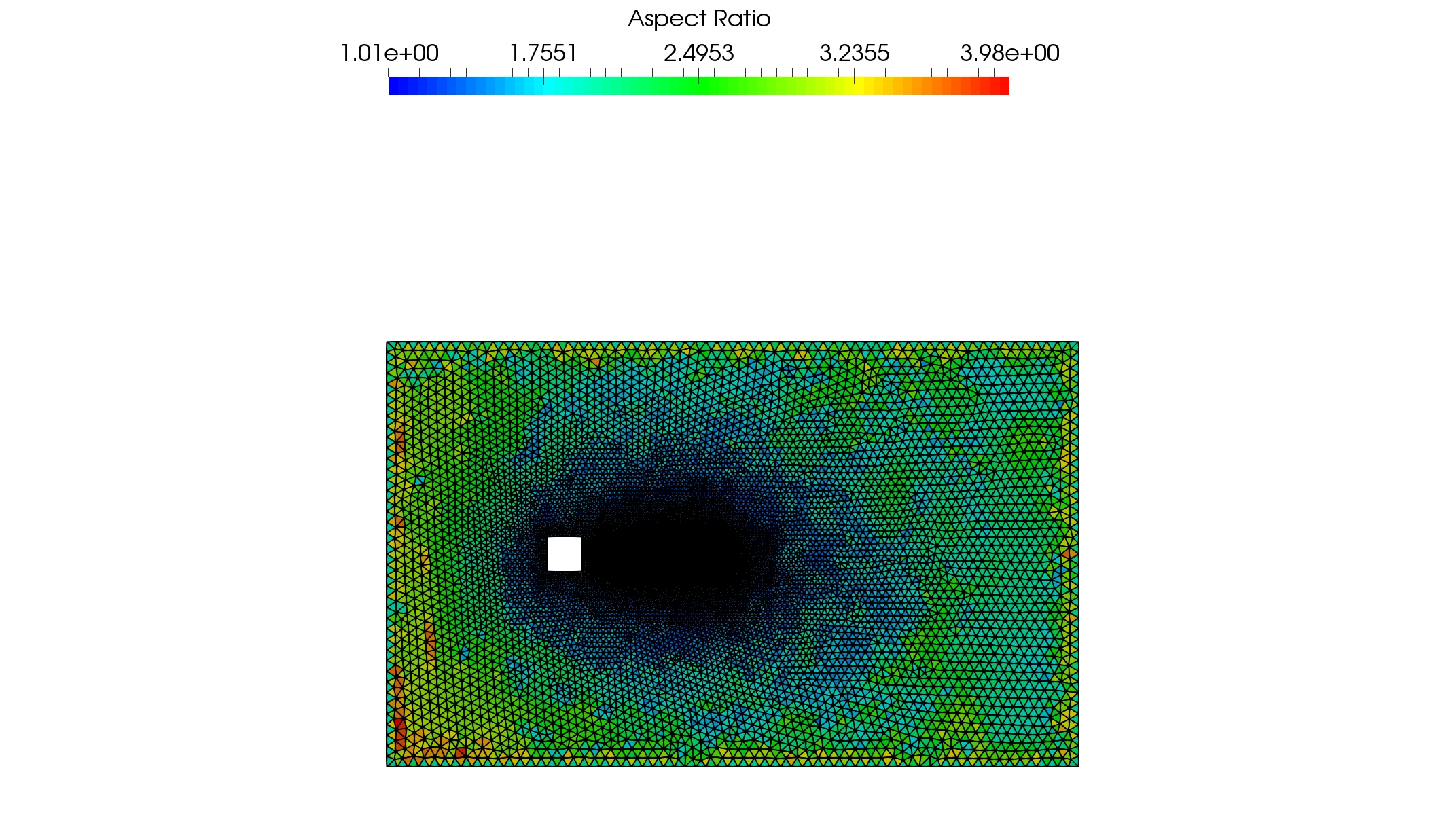}  & \\
(a) & (b)
\end{tabular}
\caption{The FSI mesh of the flow over an elastic cantilever problem. (a) The mesh employed in the computations. The color represents the aspect ratio of the tetrahedral elements. (b) Detailed view of the mesh near the cantilever. The solid subdomain is depicted with red color. There are two elements through the thickness of the cantilever.} 
\label{fig:fsi_beam_benchmark_mesh}
\end{center}
\end{figure}


\begin{figure}
\centering
\includegraphics[angle=0, trim=120 80 120 80, clip=true, scale = 0.5]{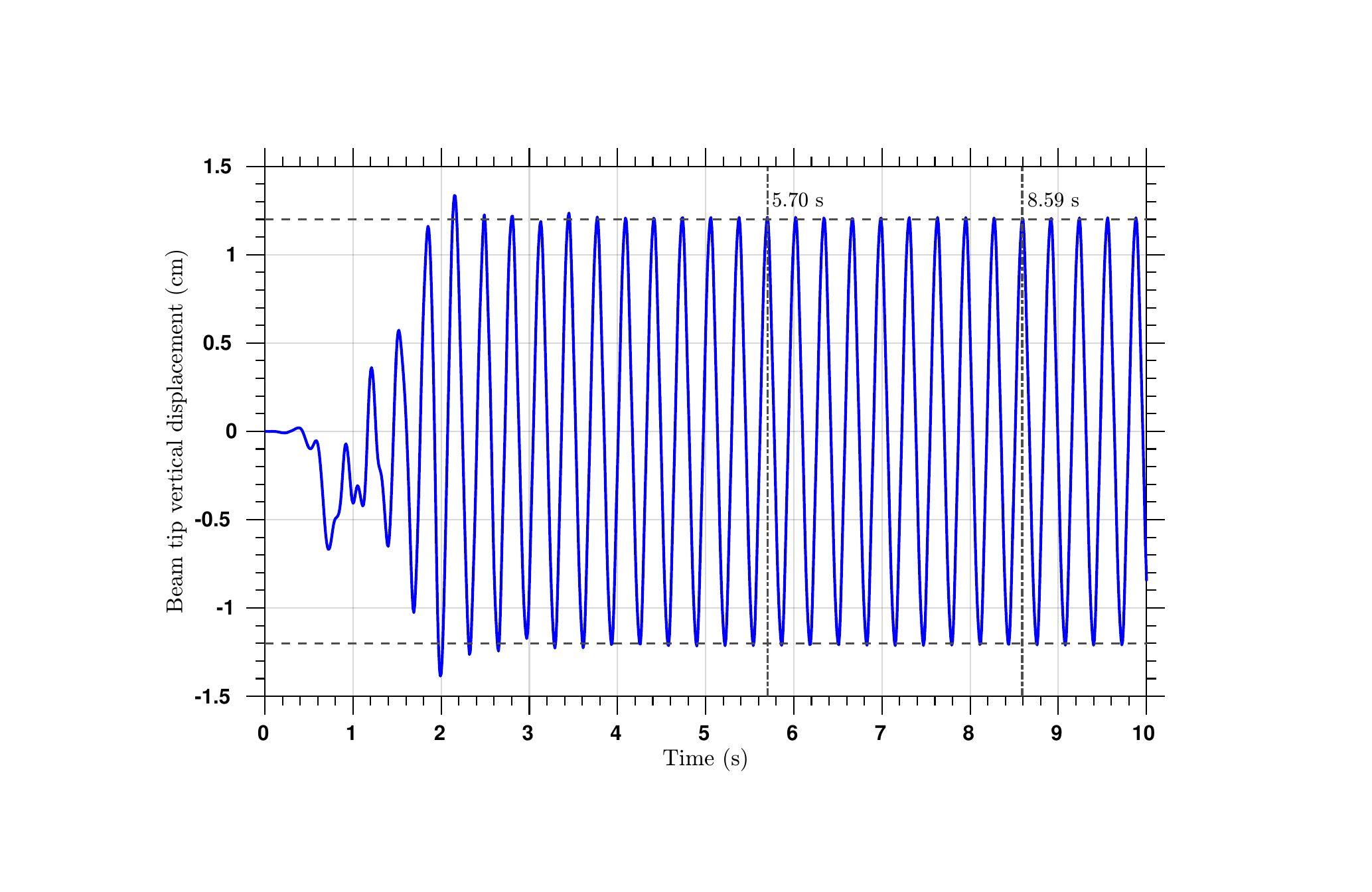}
\caption{Flow over an elastic cantilever: Vertical displacement of the tip of the cantilever. The tip vertical displacement is 1.2 cm. There are 9 periods between 5.70 s and 8.59 s. The average period of the oscillation is 0.32 s.} 
\label{fig:fsi_beam_benchmark_tip_disp}
\end{figure}

\begin{table}[htbp]
\begin{center}
\tabcolsep=0.19cm
\renewcommand{\arraystretch}{1.5}
\begin{tabular}{c c c }
\hline
Author & Oscillation period (s) & Tip displacement (cm) \\
\hline
W.A. Wall \cite{Wall2000} & 0.31 - 0.36 & 1.12 - 1.32 \\
W.G. Dettmer and D. Peri\'c \cite{Dettmer2006} & 0.32 - 0.34 & 1.1 - 1.4 \\
Y. Bazilevs, et al. \cite{Bazilevs2008} & 0.33 & 1.0 - 1.5 \\
C. Wood, et al. \cite{Wood2010} & 0.32 - 0.36 & 1.10 - 1.20 \\
Current work & 0.32 & 1.20 \\
\hline 
\end{tabular}
\end{center}
\caption{Comparison of the obtained results with reported results in the literature.}
\label{table:beam_disp_period_compare}
\end{table}

\begin{figure}
\begin{center}
\begin{tabular}{cc}
\scalebox{0.22}{\includegraphics[angle=0, trim=360 360 860 360,
clip=true]{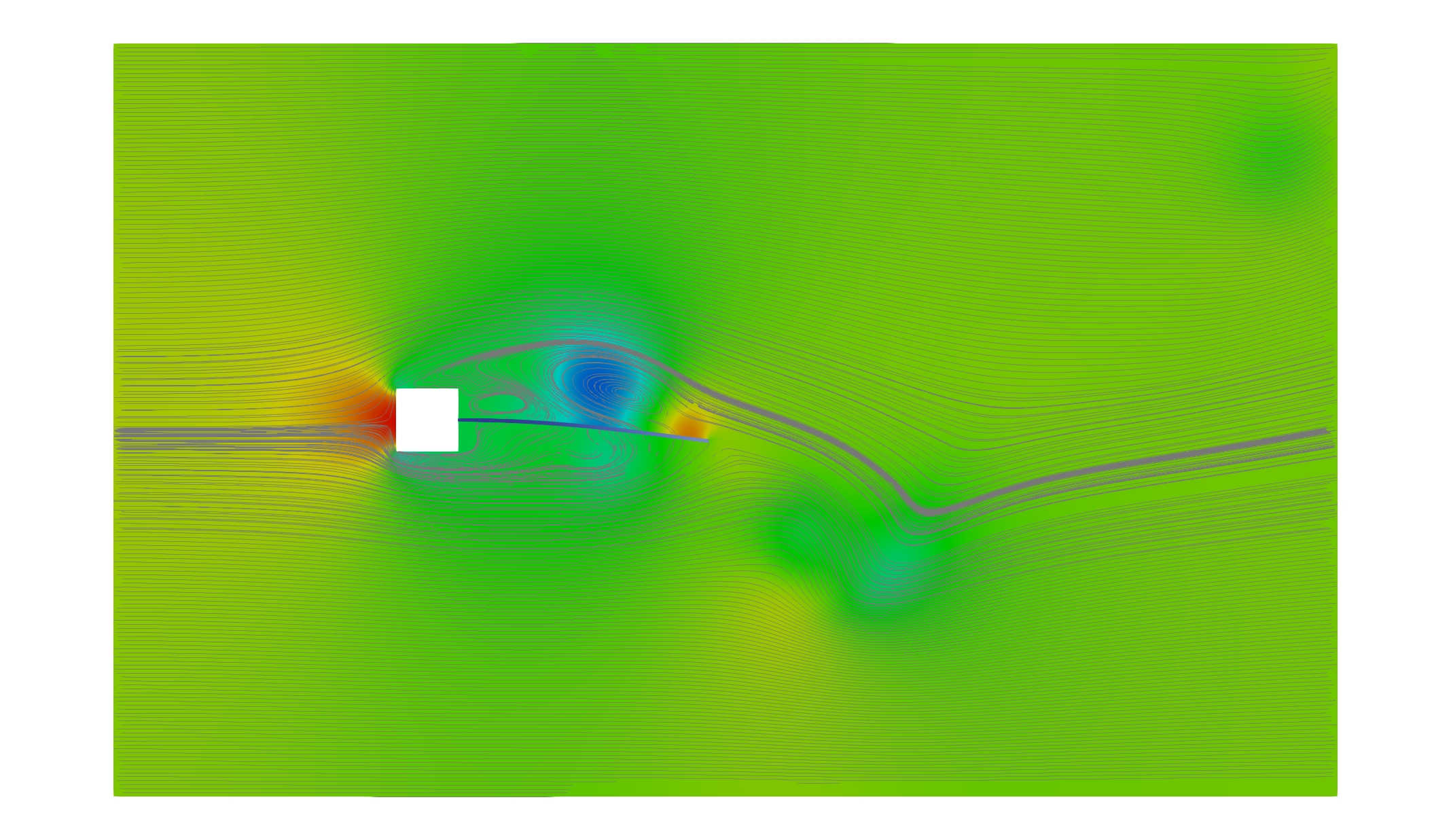} }  &  
\scalebox{0.22}{\includegraphics[angle=0, trim=360 360 860 360,
clip=true]{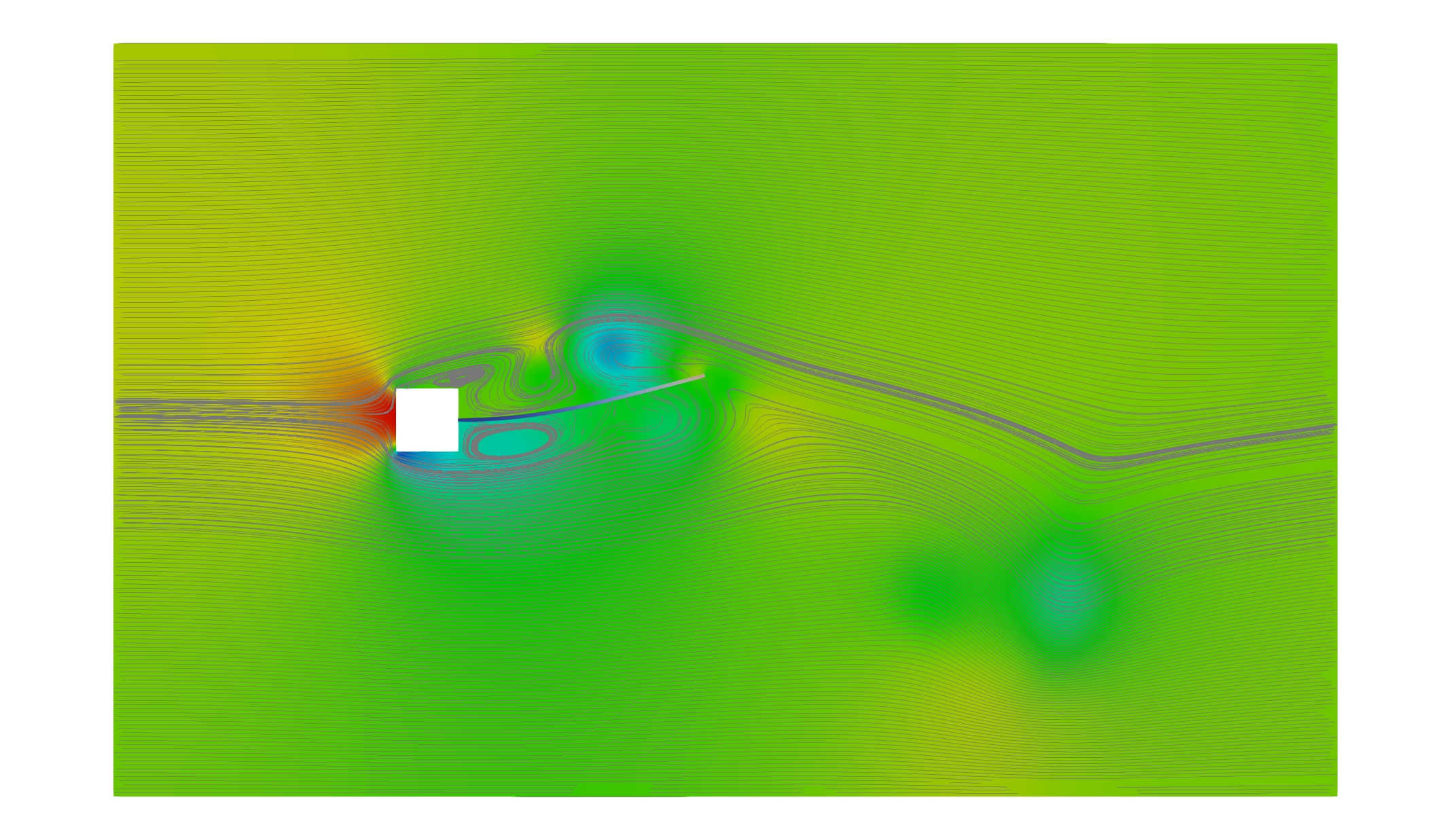} } \\
5.60 s & 5.65 s \\
\scalebox{0.22}{\includegraphics[angle=0, trim=360 360 860 360,
clip=true]{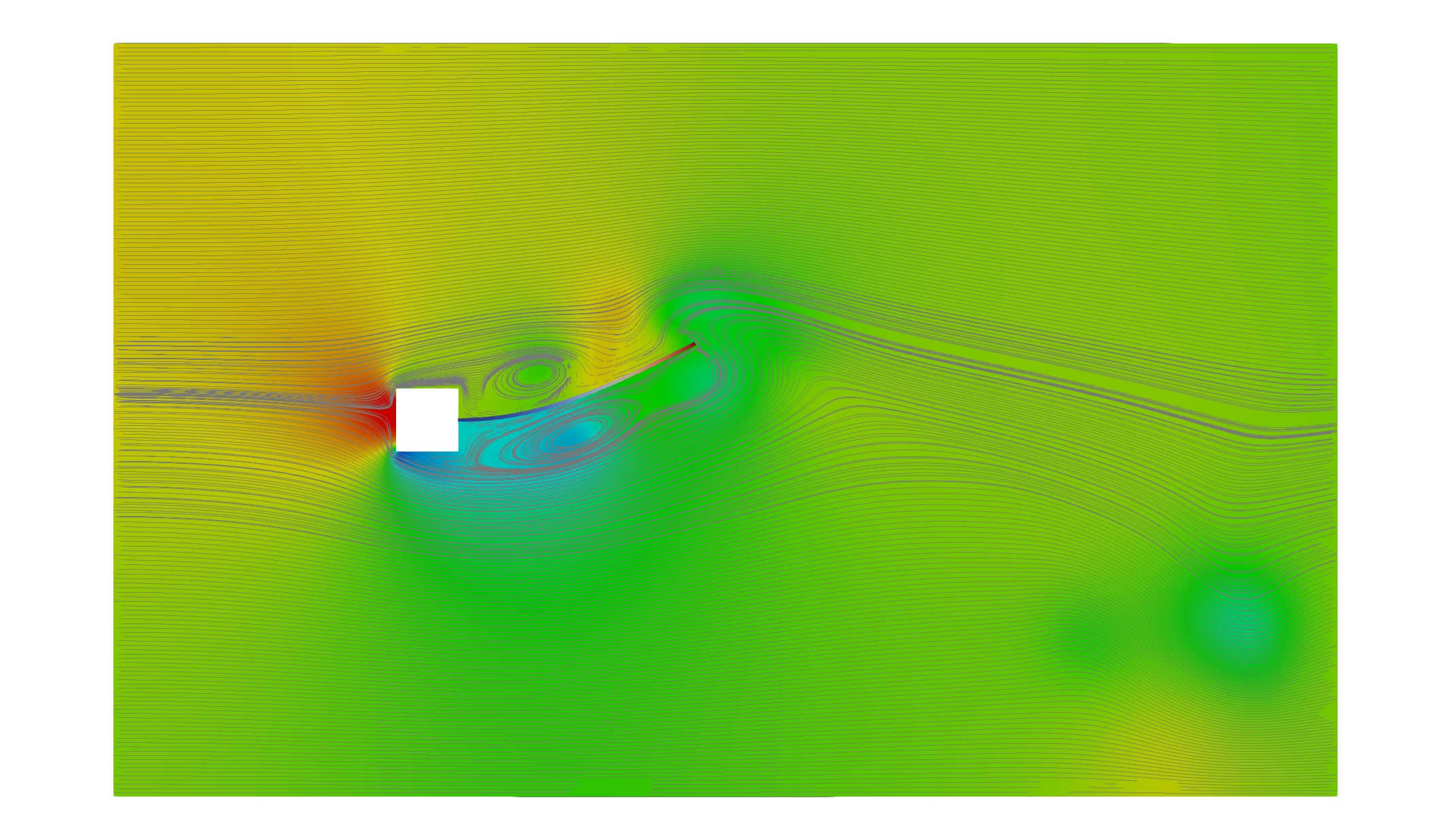} }  &  
\scalebox{0.22}{\includegraphics[angle=0, trim=360 360 860 360,
clip=true]{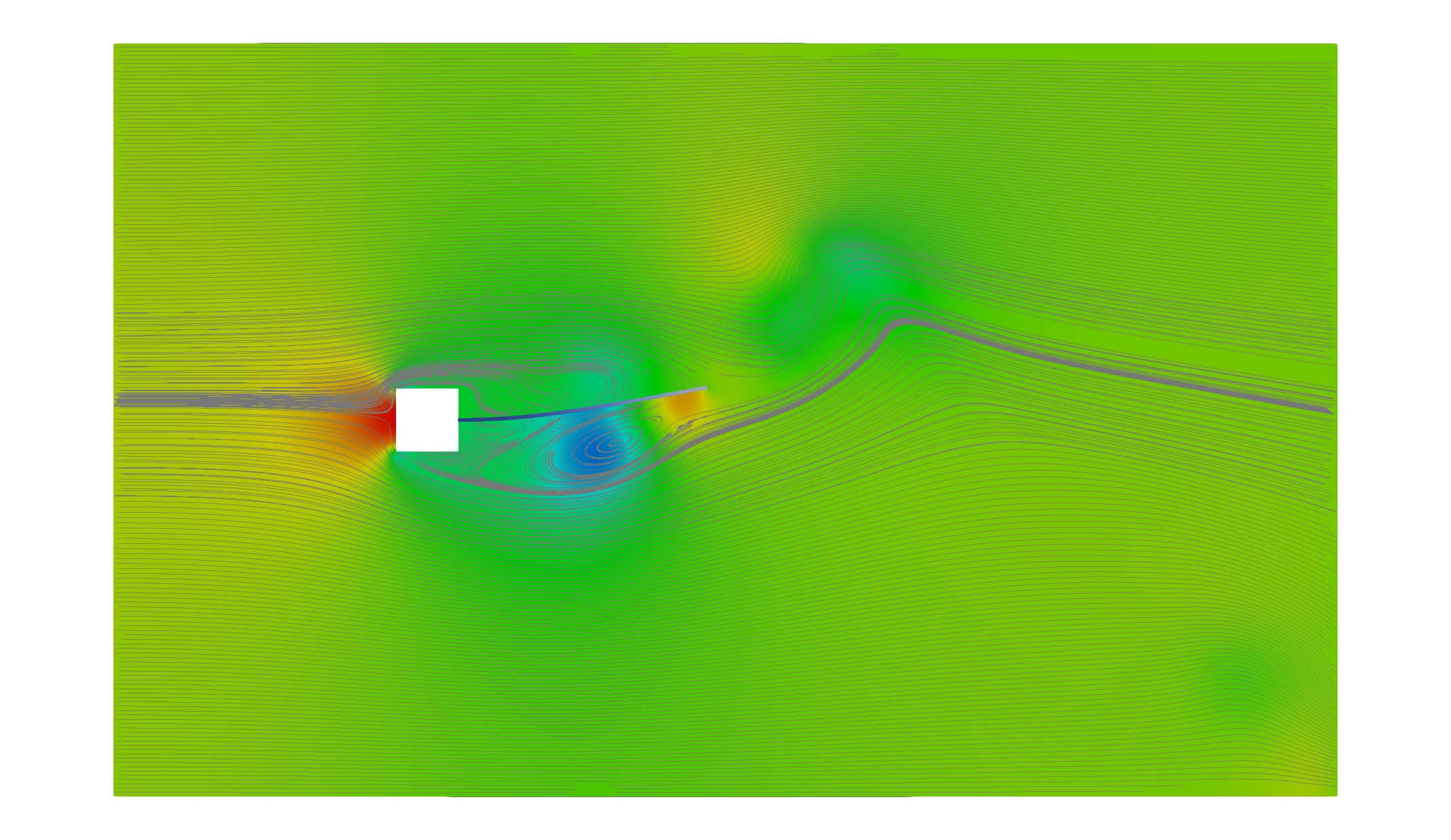} } \\
5.70 s & 5.75 s \\
\scalebox{0.22}{\includegraphics[angle=0, trim=360 360 860 360,
clip=true]{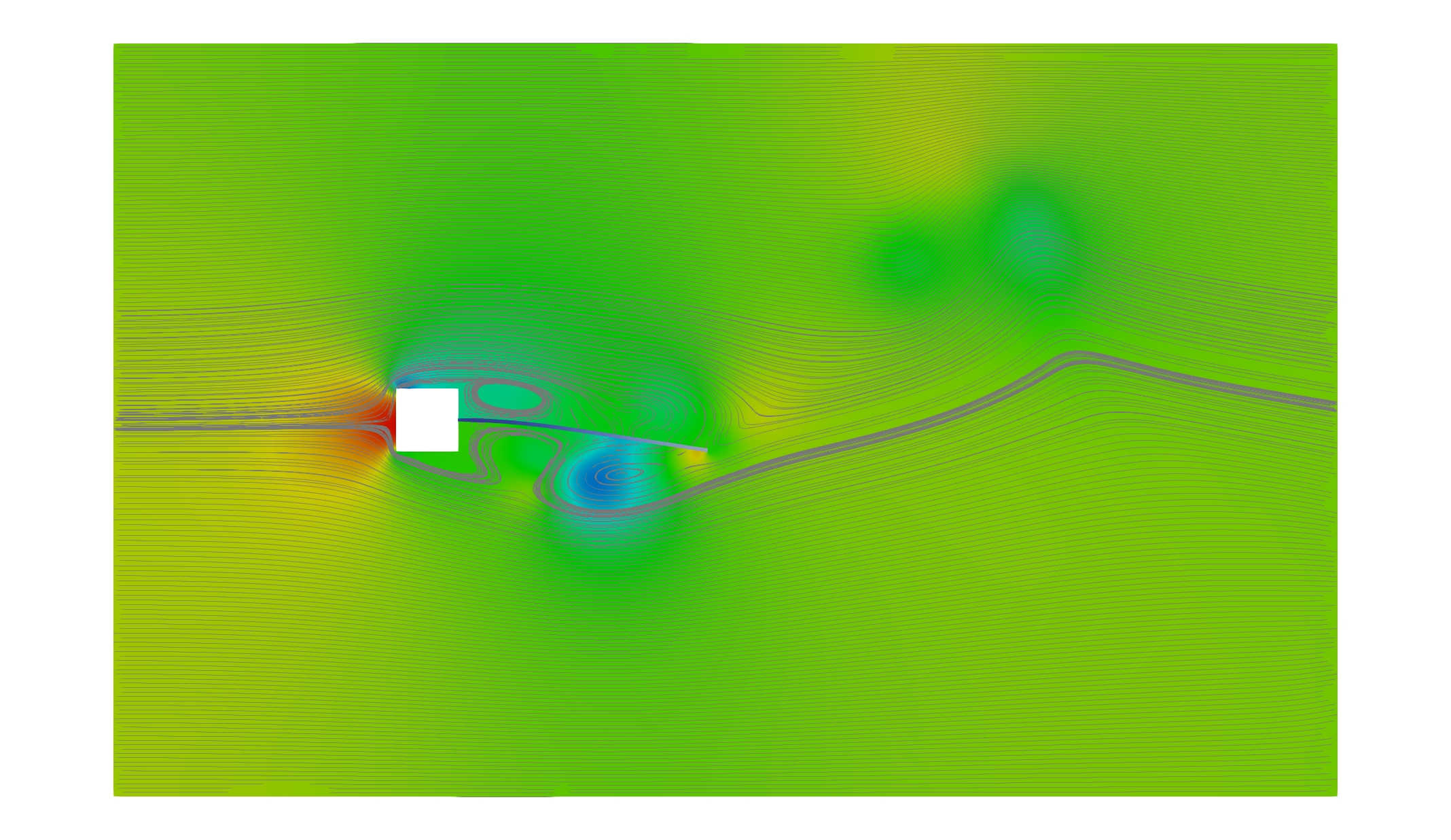} }  &  
\scalebox{0.22}{\includegraphics[angle=0, trim=360 360 860 360,
clip=true]{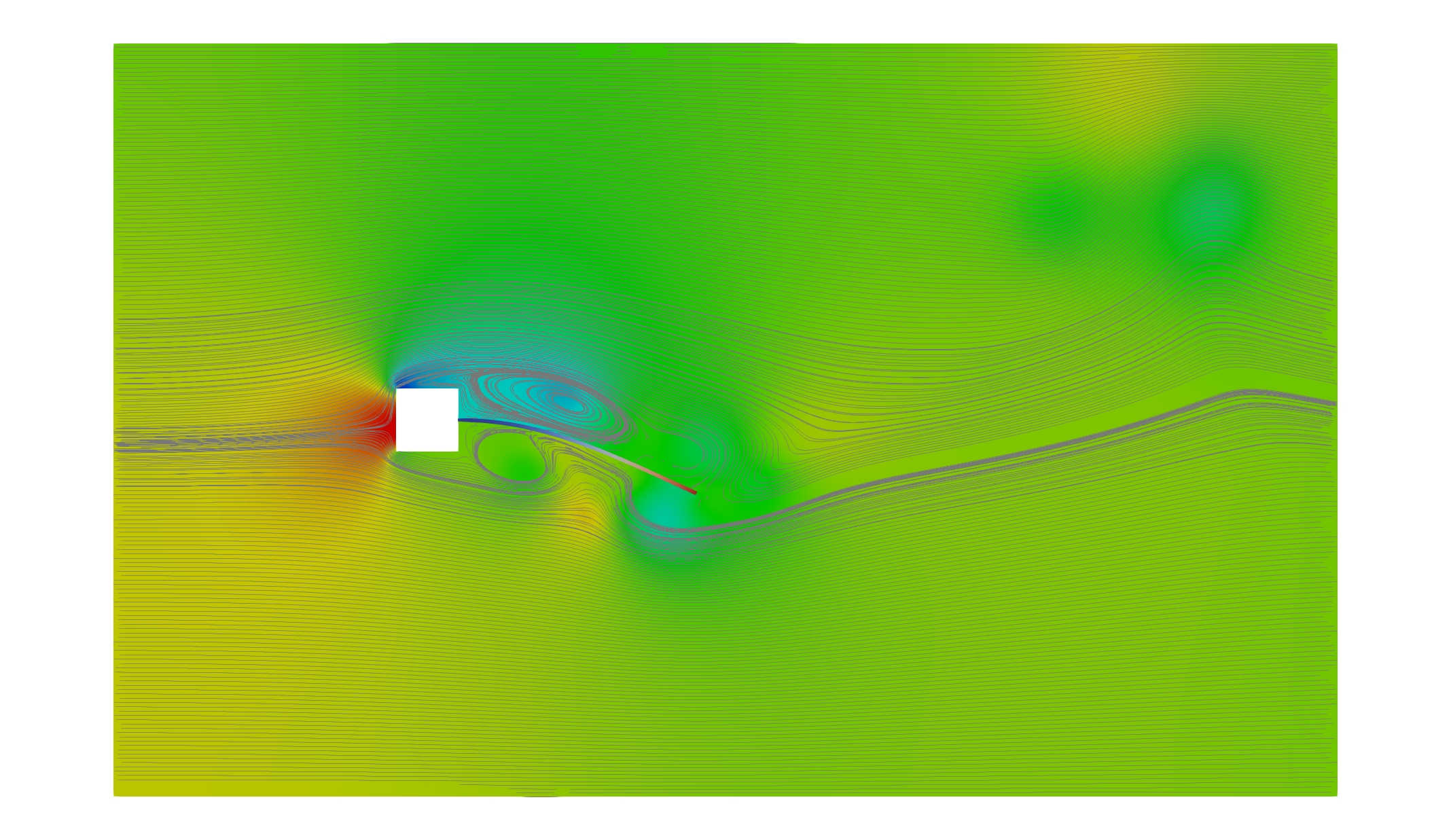} } \\
5.80 s & 5.85 s
\end{tabular}
\includegraphics[angle=0, trim=0 1000 0 0, clip=true, scale = 0.18]{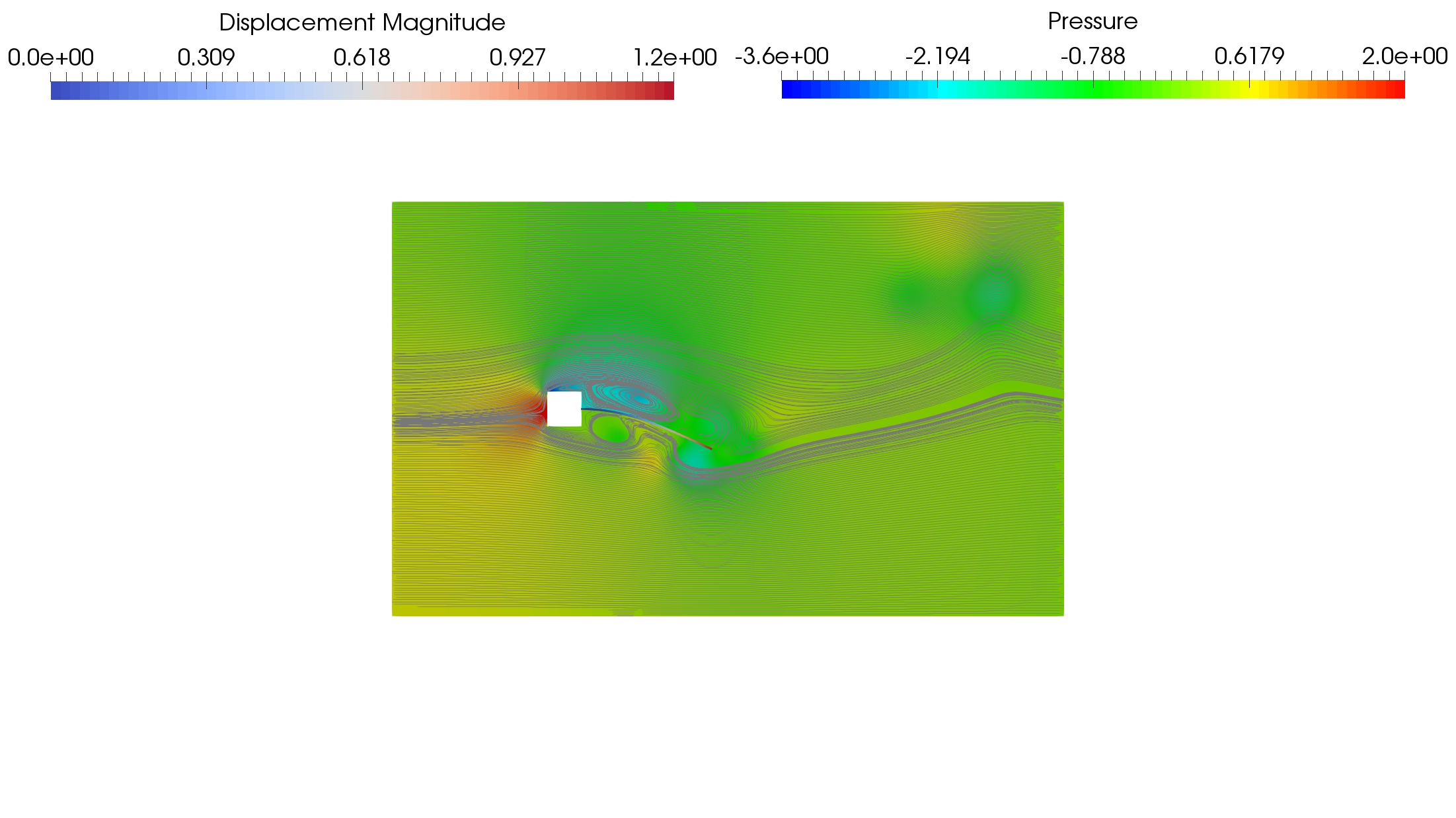}
\caption{Flow over an elastic cantilever: The fluid pressure and the cantilever displacement magnitude are depicted with the centimetre-gram-second units. The fluid streamlines are illustrated as well. } 
\label{fig:beam_illustration}
\end{center}
\end{figure}

The fluid subdomain is discretized with 74696 tetrahedral elements, and the solid subdomain is discretized with 5369 tetrahedral elements (Figure \ref{fig:fsi_beam_benchmark_mesh}). The mesh motion is governed by the pseudo-linear-elasticity algorithm, wherein the mesh Young's modulus is set to be unity and the mesh Poisson's ratio is set to be 0.3. The mesh is fixed at the inflow and outflow boundary surfaces as well as the square block surface. On the lateral surfaces, the mesh is only allowed to move in the tangential direction. In the algorithm, the mesh Young's modulus is multiplied with the inverse of the Jacobian determinant of the element mapping to provide the Jacobian stiffening. This procedure is found to be necessary for this simulation. The problem is simulated with $\Delta t = 1.0 \times 10^{-3}$ s up to $T=10.0$ s. We set  $c_m = c_c = 0$ in the study of this example. In our numerical experiences, for compressible materials, setting $c_m = c_c = 0$ gives more accurate results. The vertical displacement of the cantilever over time is plotted in Figure \ref{fig:fsi_beam_benchmark_tip_disp}, and the comparison of the obtained results with those in the literature is listed in Table \ref{table:beam_disp_period_compare}. The results of the new FSI computational framework are in good agreement with the reported results. Snapshots of the simulation results are depicted in Figure \ref{fig:beam_illustration}.

\subsection{The Greenshields-Weller numerical benchmark}
The FSI benchmark example we consider next describes wave propagation in an elastic tube \cite{Greenshields2005}. This benchmark example has been adopted to verify several existing cardiovascular FSI solvers \cite{Bazilevs2006,Passerini2013}. The computational domain consists of a right circular hollow cylinder representing the elastic tube and an inner right circular cylinder representing the fluid domain. The length of the tube is $10$ cm, the inner radius for the fluid domain is $1$ cm, and the outer radius is $1.2$ cm. For the elastic material, the reference density $\rho^s_0 = 1$ $\textup{g/cm}^3$, the Young's modulus $E^s = 1.0 \times 10^7$ $\textup{dyn/cm}^2$, and the Poisson's ratio is $0.3$. In the original benchmark problem \cite{Greenshields2005}, a small-strain linear elastic material is used. Apparently, the small-strain model is not incorporated in our theory. We adopt the Neo-Hookean model with the volumetric free energy given by \eqref{eq:psi_vol_M94},
\begin{align}
\label{eq:GW_material_model}
G\left( \tilde{\bm C}, p \right) = \frac{\mu^s }{2\rho_0} \left( \textup{tr}\tilde{\bm C} - 3 \right) -\frac{\kappa}{\rho_0} \ln \left(\frac{\kappa}{p+\kappa} \right).
\end{align}
In doing so, the density is given by $\rho = \rho_0 (1+p/\kappa)$, which is identical to the density-pressure relation used in \cite{Greenshields2005}. The fluid is described by the incompressible Navier-Stokes equations with fluid density $\rho^f_0 = 1$ $\textup{g/cm}^3$ and dynamic shear viscosity $\bar{\mu} = 0.04$ poise. The initial conditions for all fields are set to be zero. A step change in pressure is imposed on the inlet surface and the pressure value $p_{in}$ is set to be $5$ kPa. At the fluid outlet surface and the lateral surface of the solid, stress-free boundary conditions are applied. The solid velocity in the axial direction on the inlet and outlet surfaces are set to be zero. The geometry and the boundary conditions are illustrated in Figure \ref{fig:fsi_benchmark_setting} (a). The solid domain is discretized by structured tetrahedral elements. In the fluid domain, a boundary layer mesh is created with thickness $0.2$ cm, and the rest fluid domain is unstructured. The mesh we used is depicted in Figure \ref{fig:fsi_benchmark_setting} (b). There are $2.19\times 10^6$ elements and $3.76 \times 10^5$ nodal points in the finite element mesh. The problem is simulated with a fixed time step $\Delta t = 2.0 \times 10^{-7}$ s up to $T=8.0\times 10^{-3}$ s. Since the material is compressible, we set $c_m = c_c = 0$ in the study of this example. The mesh motion is given by the harmonic extension algorithm. The analytic value of the wave speed is $8.77$ m/s. To obtain the numerically predicted wave speed, we define the wave front as the location where the pressure is $2.5$ kPa. Using a linear function to fit the locations of the wave front with respect to time, we obtain that the numerical wave speed is $8.49$ m/s, which is 3\% lower than the analytic value and 1\% lower than the numerical prediction given in \cite{Greenshields2005} (8.58 m/s). With the pressure wave at time $8.0$ ms, we may obtain the distance between two peaks is $2.84$ cm, and consequently the frequency of the wave is 298.9 Hz. The analytic values for the wave frequency using three different formulas are 269 Hz, 308 Hz, and 336 Hz, respectively \cite{Greenshields2005}, and the numerical prediction given in \cite{Greenshields2005} is 318 Hz. Our numerical predicted wave frequency is 10 \% higher, 3\% lower, and 12\% lower than the analytic values, and 6 \% lower than the numerical prediction given by \cite{Greenshields2005}. The same problem is also simulated with the Young's modulus $E^s = 1.0 \times 10^6$ $\textup{dyn/cm}^2$ and $p_{in} = 0.5$ kPa. In this case, the analytic value of the wave speed  $2.77$ m/s. Our numerical prediction gives $2.70$ m/s, which is 3\% lower than the theoretical value.  These results demonstrate satisfactory agreement with published results.

\begin{figure}
	\begin{center}
	\begin{tabular}{c}
\includegraphics[angle=0, trim=80 50 80 30, clip=true, scale = 0.66]{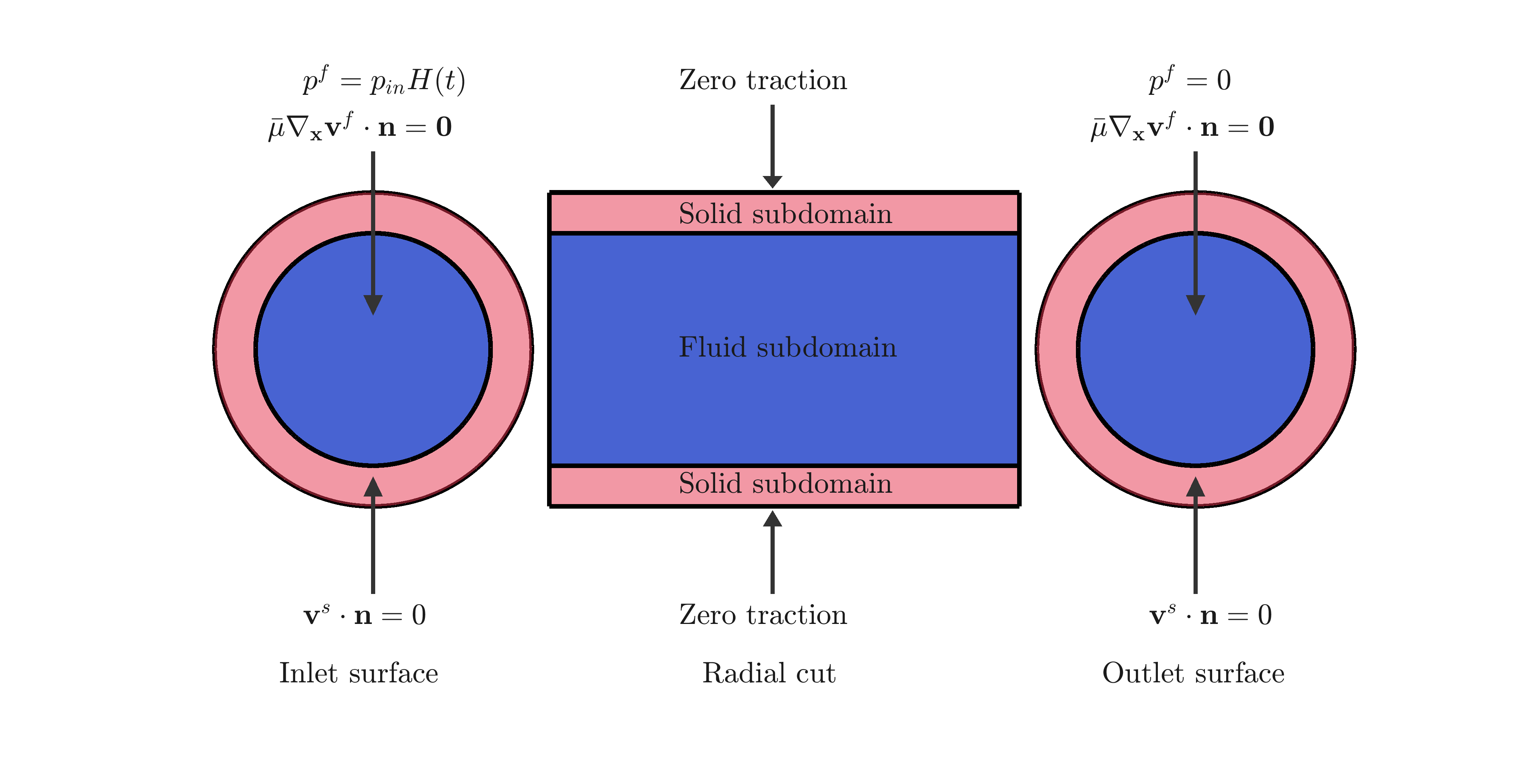} \\
(a) \\[1cm]
\includegraphics[angle=0, trim=80 0 80 0, clip=true, scale = 0.12]{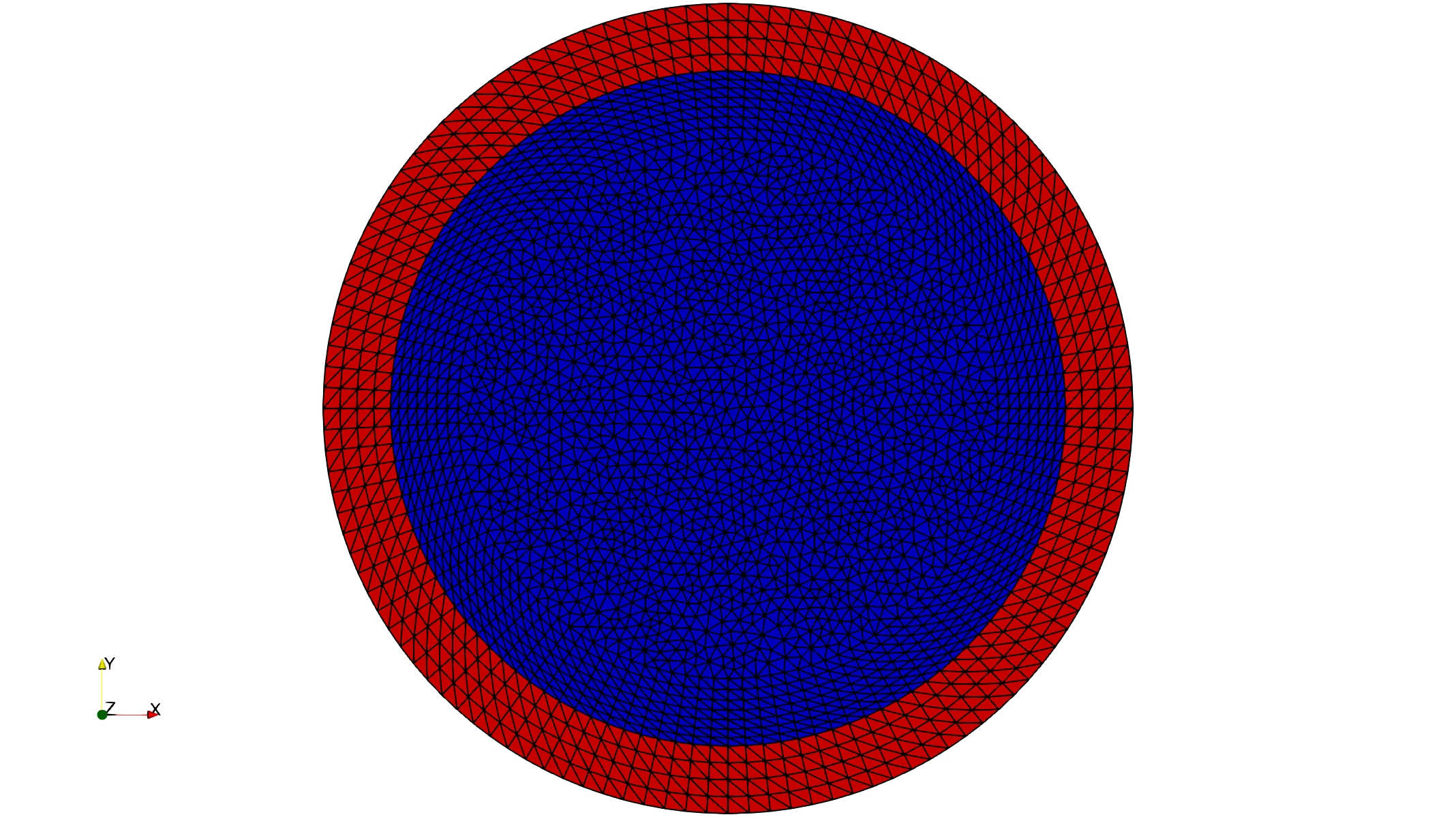} \\[5mm]
(b)
\end{tabular}
\caption{The Greenshields-Weller benchmark problem: (a) Geometry setting and boundary conditions. $H(t)$ represents the Heaviside step function of time $t$; (b) FSI mesh on the outflow surface. The solid subdomain is depicted with red color and the fluid subdomain is depicted with blue color.} 
\label{fig:fsi_benchmark_setting}
\end{center}
\end{figure}

\begin{figure}
\begin{center}
\begin{tabular}{c}
\scalebox{0.18}{\includegraphics[angle=0, trim=0 300 0 300,
clip=true]{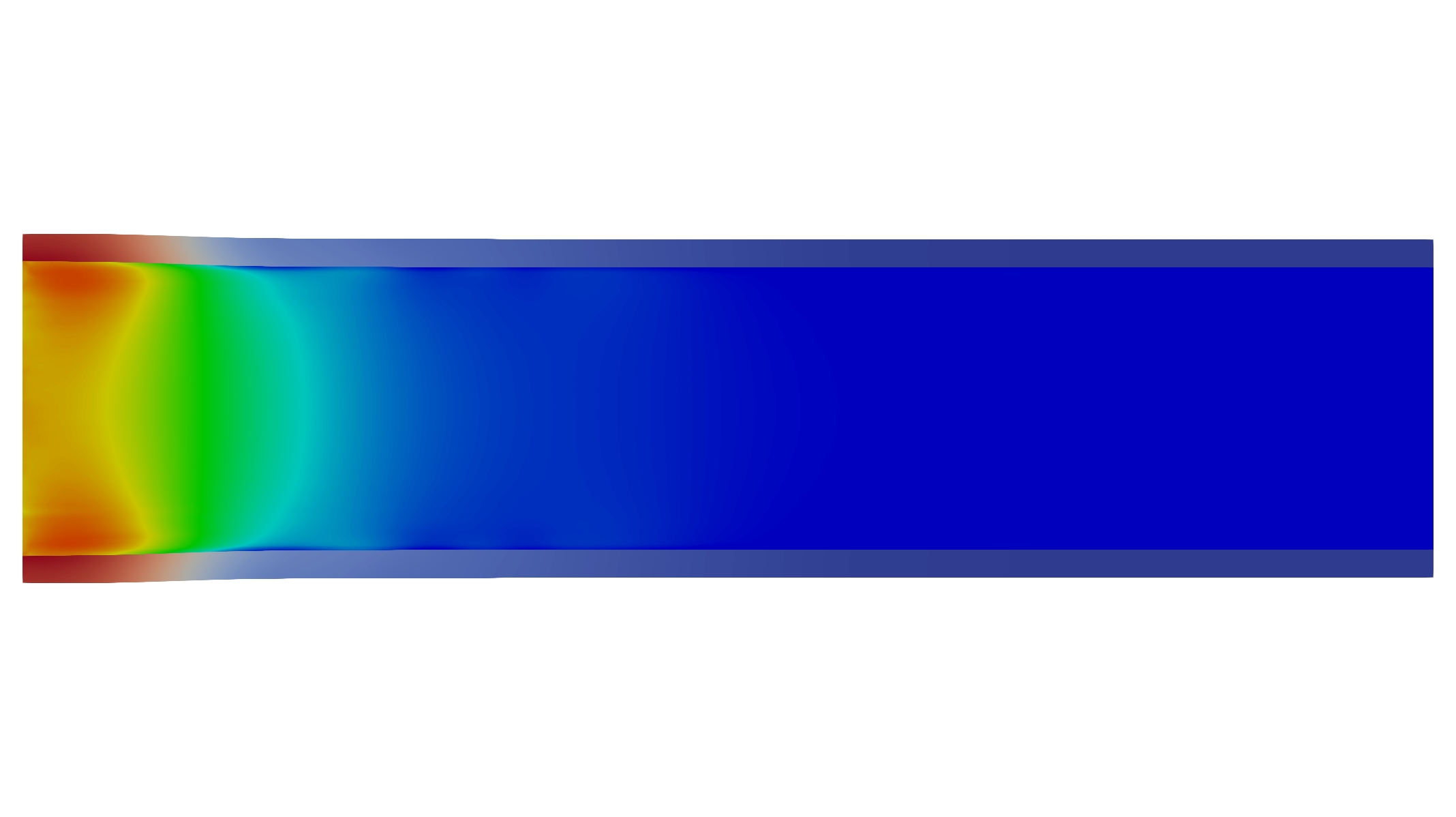} }  \\
$t=2.0$ ms \\
\scalebox{0.18}{\includegraphics[angle=0, trim=0 300 0 300,
clip=true]{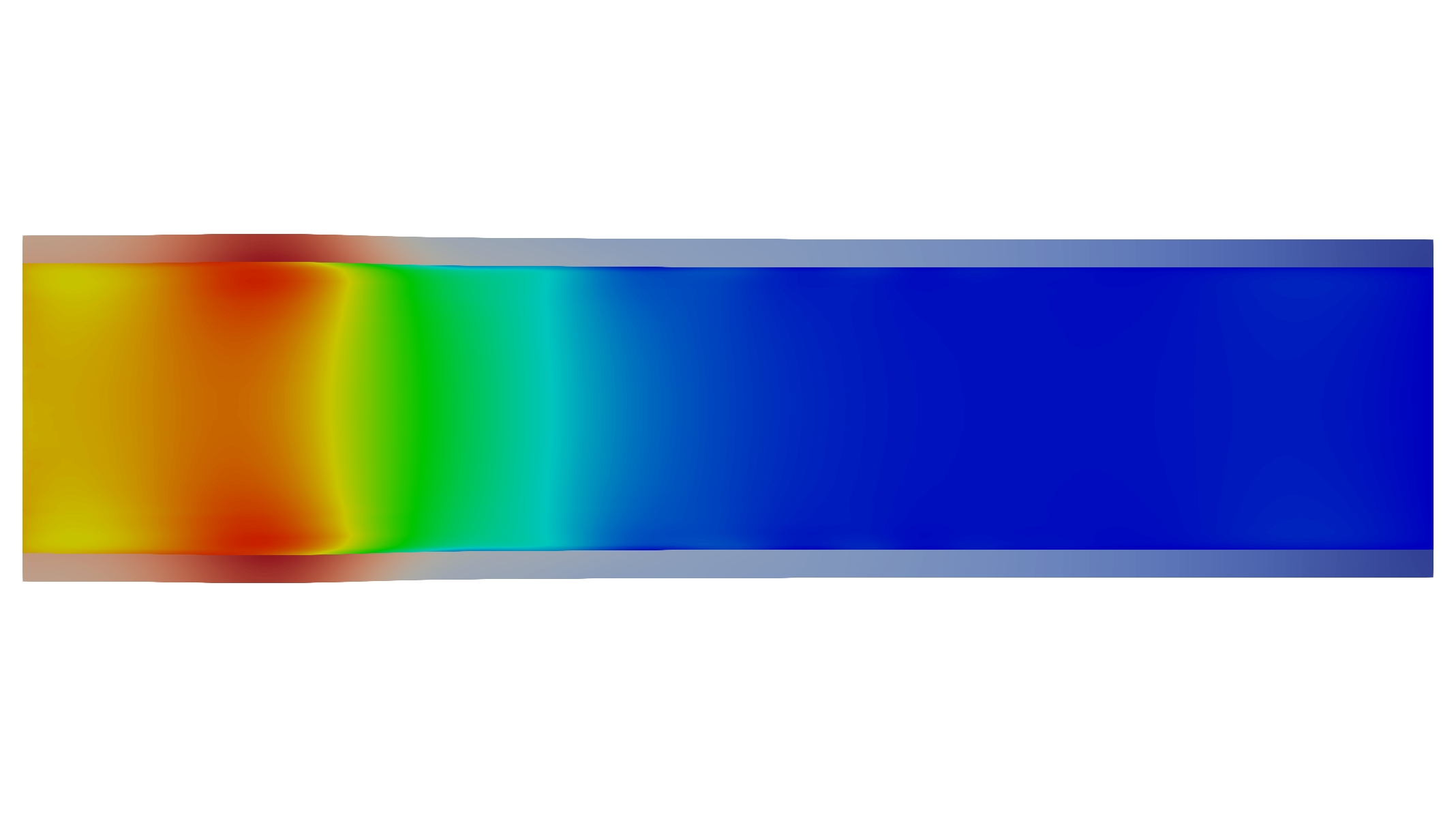} }  \\
$t=4.0$ ms\\
\scalebox{0.18}{\includegraphics[angle=0, trim=0 300 0 300,
clip=true]{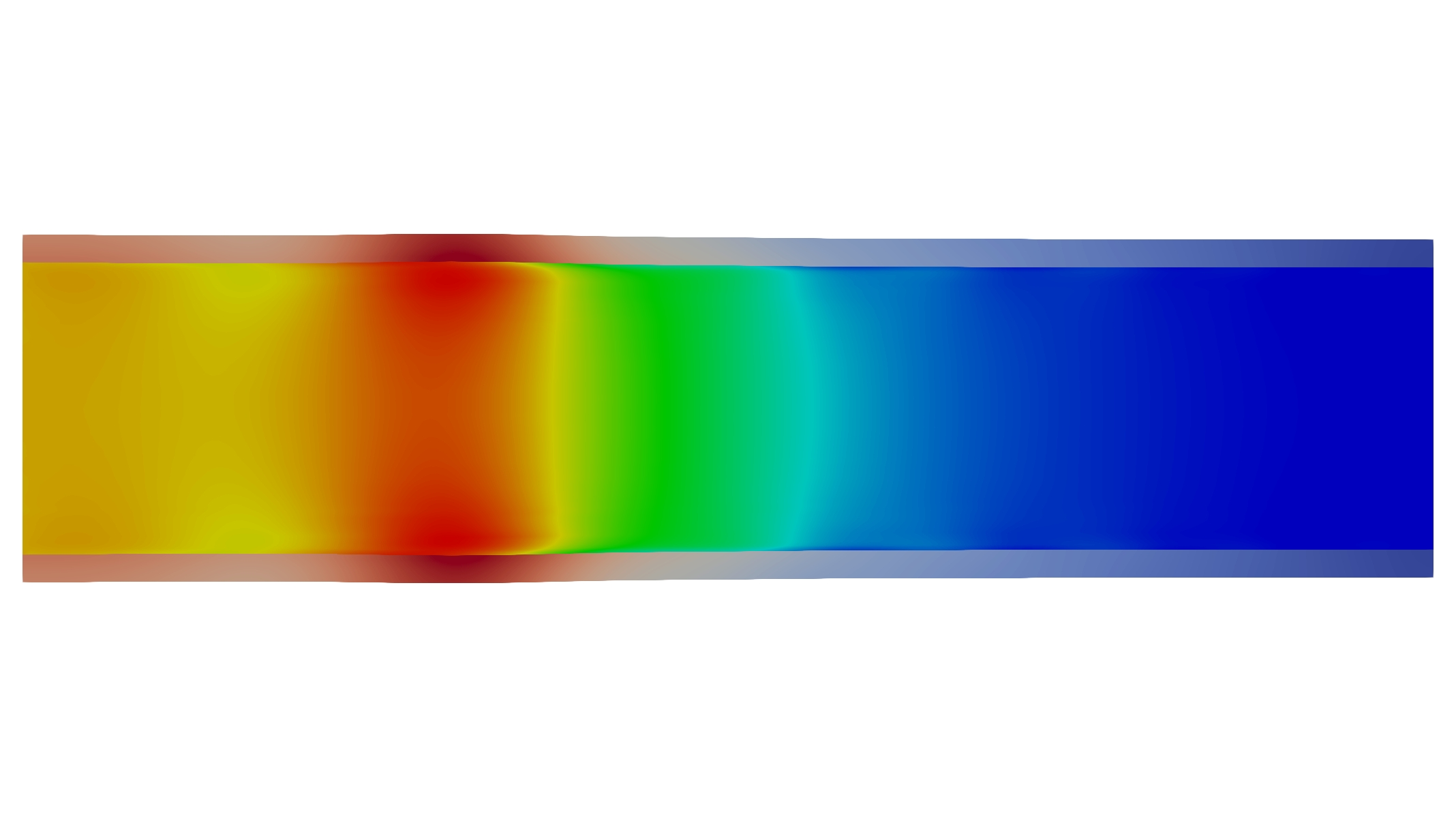} }  \\
$t=6.0$ ms \\
\scalebox{0.18}{\includegraphics[angle=0, trim=0 300 0 300,
clip=true]{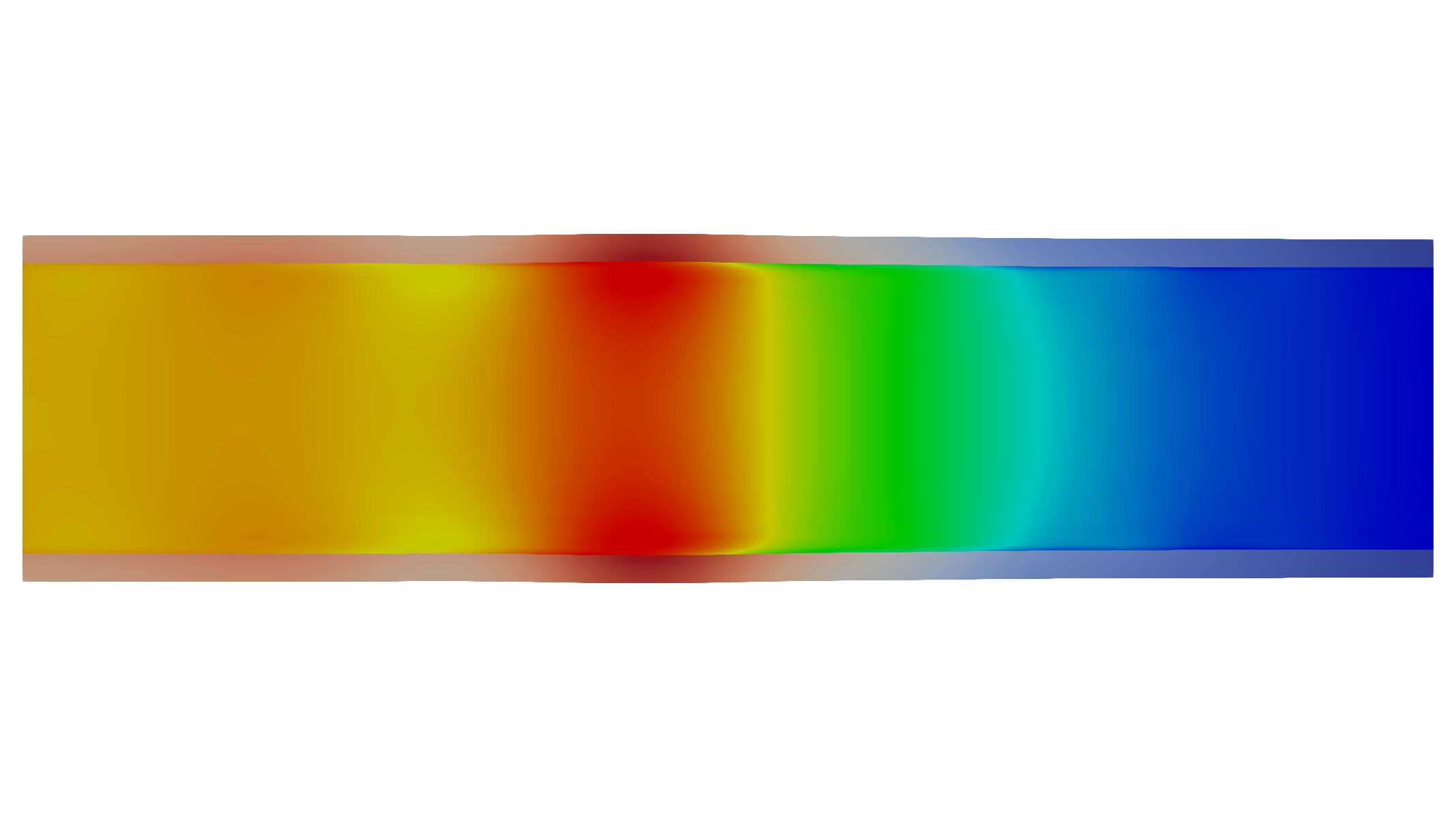} }  \\
$t=8.0$ ms
\end{tabular}
\includegraphics[angle=0, trim=0 900 0 0, clip=true, scale = 0.18]{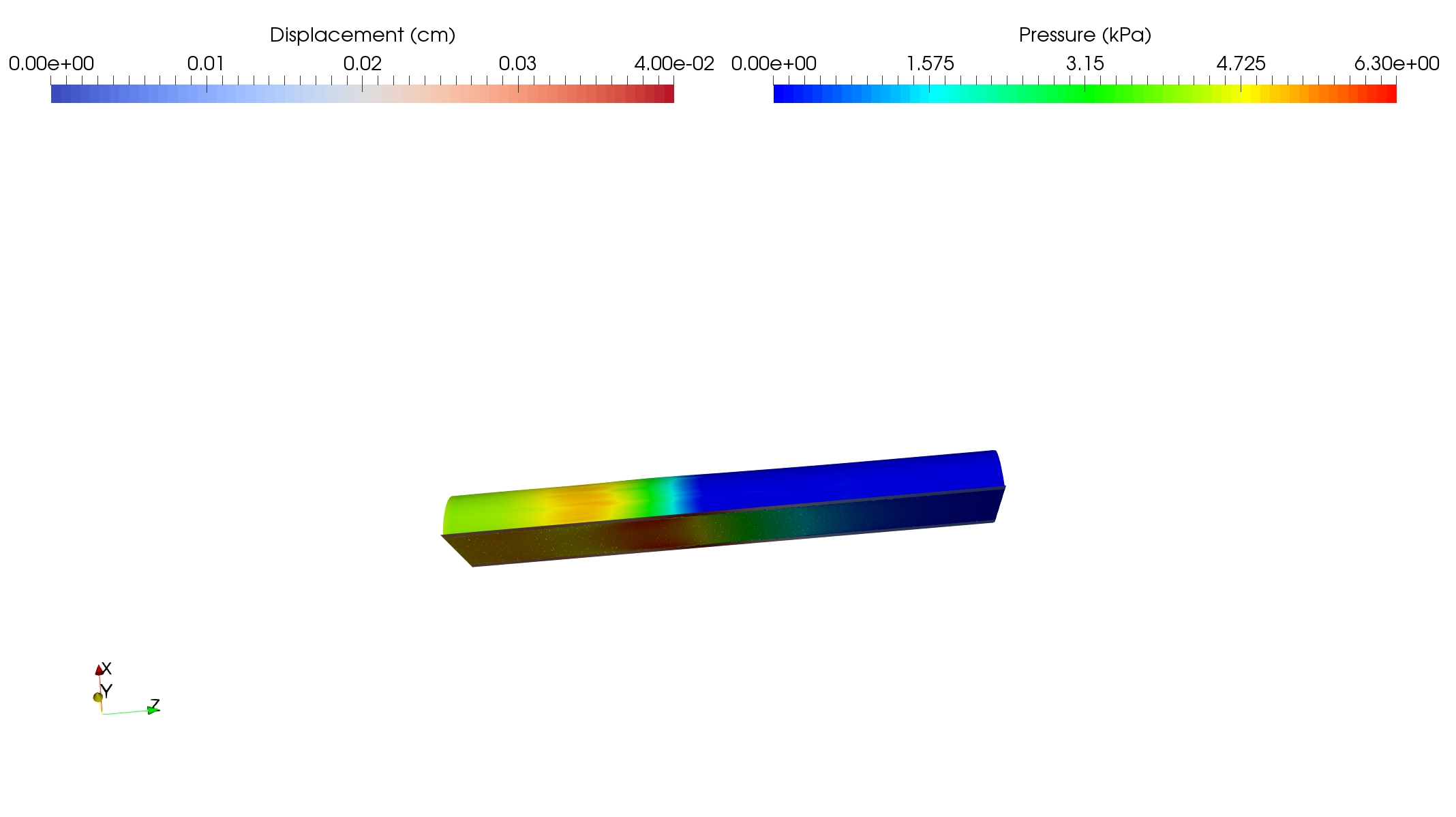}
\caption{The Greenshields-Weller benchmark problem: Contours of the fluid pressure and wall displacement magnitude on a radial slice at time $t=2.0$, $4.0$, $6.0$, and $8.0$ ms. } 
\end{center}
\end{figure}

\begin{figure}
	\begin{center}
	\begin{tabular}{c}
\includegraphics[angle=0, trim=0 80 100 0, clip=true, scale = 0.5]{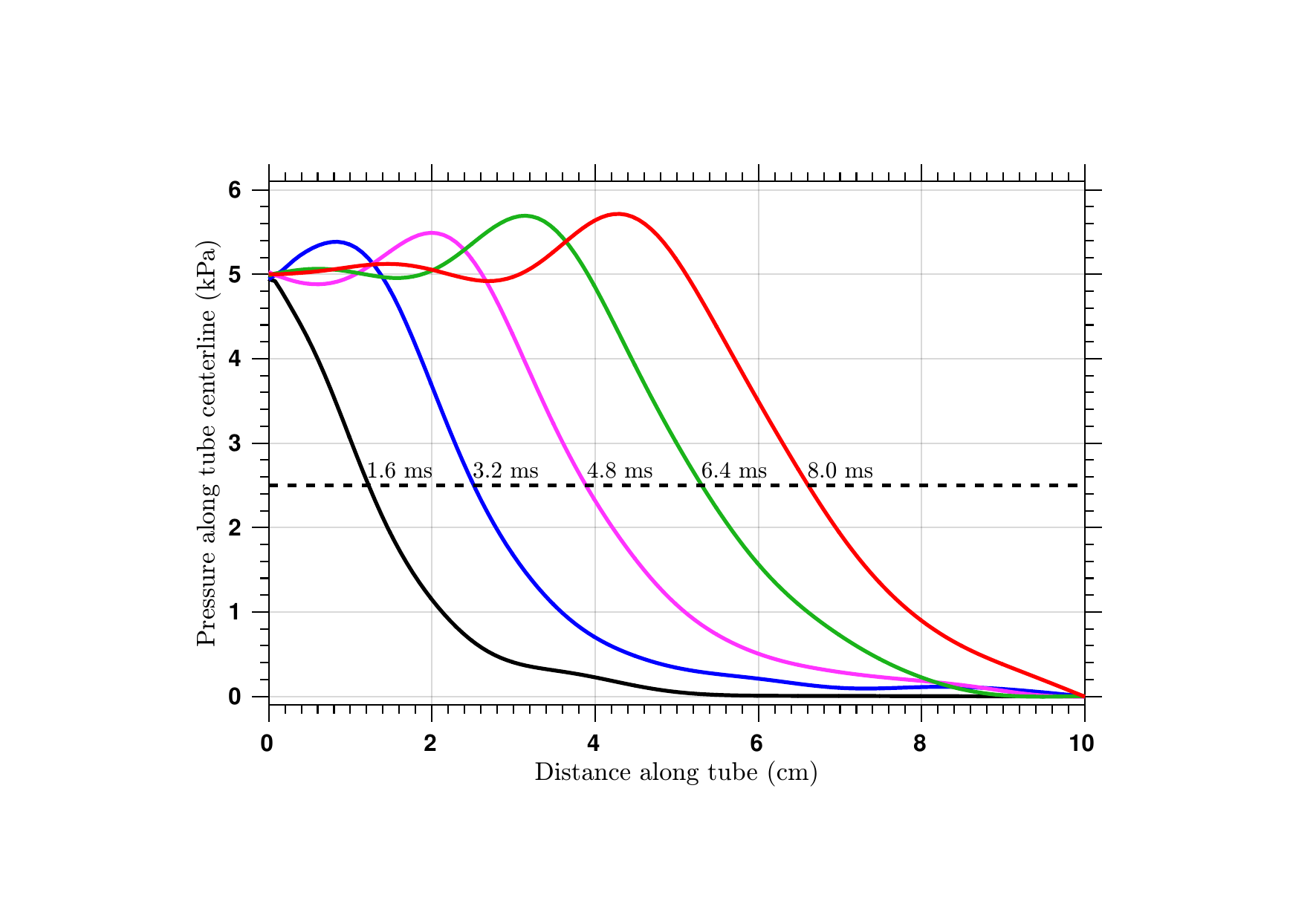} \\
(a) \\
\includegraphics[angle=0, trim=0 80 100 0, clip=true, scale = 0.5]
{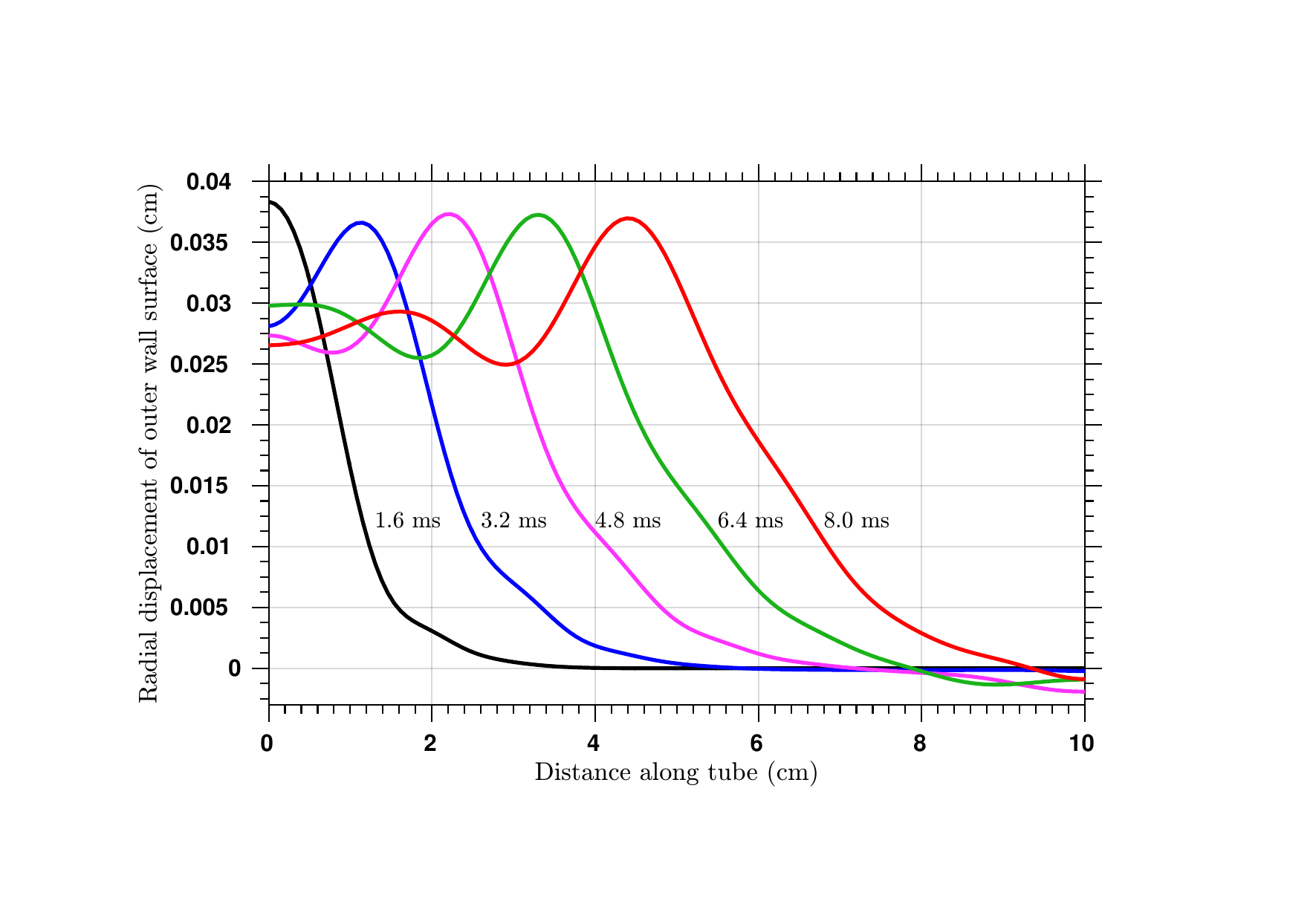} \\
(b)
\end{tabular}
\caption{The Greenshields-Weller benchmark problem: (a) Pressure along the tube centreline; (b) Radial displacement of the outer wall surface.} 
\label{fig:fsi_benchmark_pressure_tube_wall_disp}
\end{center}
\end{figure}

\section{Conclusions}
\label{sec:conclusion}
We present a comprehensive suite of theoretical and numerical methodologies for solid dynamics, fluid dynamics, and FSI problems. We summarize the contributions in detail as follows.

A novel continuum modeling framework is developed. Our derivation differs from traditional approaches in three aspects. (1) We start with a dynamic model for the continuum body. In the thermodynamic configuration space, a quasi-static process is constrained on an equilibrium manifold in a lower dimension, representing a succession of equilibrium states. This is really an idealized scenario, since it ignores several critical physical effects, and hence does not represent real physics \cite[Chapter 4]{Callen1985}. (2) We choose the thermodynamic potential as the Gibbs free energy instead of the Helmholtz free energy. The two potentials both give equivalent descriptions of material behavior for any unconstrained process. However, for constrained processes, one should judiciously choose the potential. In our discussion, we are particularly interested in the isochoric process, and the Helmholtz free energy degenerates in this case. A Legendre transformation is performed to transform the independent variable from the specific volume to the pressure. This leads to a theory based on the Gibbs free energy and consequently a pressure primitive variable formulation. This formulation is known to be well-behaved in both compressible and incompressible regimes. (3) Our pressure equation is derived from the mass balance equation instead of the equation of state. In the traditional two-field variational principle, the equation for the pressure field is often introduced as an algebraic equation of state \cite{Bogert1991,Holzapfel2000,Maniatty2001}. We feel that it is unnecessary and improper to put an algebraic equation in a weak form and discretize it by the finite element method (see Remark \ref{remark:scovazzi_2016_discussion}).

Our numerical formulation is designed based on VMS. We provide a formal derivation of the fine-scale model based on the general continuum model. The VMS framework is applied to construct a numerical formulation for finite hyper-elastodynamics; the generalized-$\alpha$ method is applied for time integration; the segregated algorithm introduced in \cite{Scovazzi2016} is utilized for the solution procedure. The properties of the new numerical formulation are examined using manufactured solutions and the block compression benchmark problem.

The new formulation for hyper-elastodynamics is naturally extended to FSI problems. Due to the general continuum framework we developed, the coupling between viscous fluids and hyper-elastodynamics can be viewed as a uniform continuum problem, where the body is subdivided into subdomains with different material behavior. Over the whole continuum domain, VMS is applied for spatial discretization, and the generalized-$\alpha$ method is applied for time integration. The uniform treatment of FSI problems enjoys several desirable attributes. (1) The formulation is well-behaved in both compressible and incompressible regimes for both fluids and solids. (2) The pressure instability arising from the equal-order interpolation is handled by the fine-scale modeling in both subdomains. (3) The generalized-$\alpha$ method is applied for the first-order FSI system and naturally achieves optimal high-frequency dissipation in both subdomains. (4) The resulting matrix problems in both subdomains take similar structure, which simplifies the data structure management in the code implementation. Two benchmark problems are simulated to examine the effectiveness of the new FSI formulation, and the results are in good agreement with published results.

There are several promising directions for future work. (1) This new computational framework will be extended to handle nonlinear anisotropic viscoelastic materials to account for more realistic tissue behavior. (2) The FSI computational framework is well-suited to patient-specific cardiovascular biomedical modeling with more physically realistic material models and complex geometries. (3) The concept of the immersed finite element method \cite{Casquero2015a,Zhang2004} can be utilized with the new FSI formulation developed in this work, which may open a door to novel approaches for complex FSI problems. In summary, a new computational methodology has been established, and it may provide an effective approach to handle problems that cannot be addressed with the previously existing methods.

\section*{Acknowledgements}
This work is supported by the National Institutes of Health under the award numbers 1R01HL121754 and 1R01HL123689, a Burroughs Welcome Fund Career award at the Scientific Interface, the National Science Foundation (NSF) CAREER award OCI-1150184, and the computational resources from the Extreme Science and Engineering Discovery Environment (XSEDE) supported by the NSF grant ACI-1053575. Additional computational resources were provided by the Stanford Research computing facility. 

We want to thank Prof. Guglielmo Scovazzi at Duke University for many helpful discussions. We thank Prof. Luca Ded\`e at Politecnico di Milano, Prof. Hector Gomez at Purdue University, and Mr. Fei Xu at the Iowa State University for discussions on FSI problems. We also want to thank members of the Cardiovascular Biomechanics Computational Lab at Stanford University for general discussions on biomedical modeling.

\appendix
\section{Fine-scale approximation}
\label{appd:VMS_derivation_perturbation_series}
Physically sound assumptions \cite{Bazilevs2007a,Scovazzi2004} suggest that one may represent $\bsfy^{\prime}$ by a perturbation series:
\begin{align}
\label{eq:appdB_infinite_series}
\bsfy^{\prime} = \sum_{m=1}^{\infty}\varepsilon^m \bsfy^{\prime}_m,
\end{align}
with $\varepsilon = \|\textbf{Res}\left(\bar{\bsfy}\right)\|_{\bsfV^{\prime *}}$. Using this expansion, we may rewrite \eqref{eq:VMS_fine_eqn_taylor} as
\begin{align}
\label{eq:appendix_b_taylor_functional_expansion}
\sum_{k=1}^{n} \frac{1}{k!} D^k_{\bsfy}\bsfB\left(\bsfw^{\prime}, \bar{\bsfy} \right)\left[\sum_{m=1}^{\infty}\varepsilon^m \bsfy^{\prime}_m , \cdots, \sum_{m=1}^{\infty}\varepsilon^m \bsfy^{\prime}_m \right] + o\left( \|\bsfy^{\prime}\|_{\bsfV}^n \right) = \varepsilon \hat{\mathsf R}(\bar{\bsfy})[\bsfw^{\prime}],
\end{align}
wherein
\begin{align*}
\hat{\mathsf R}(\bar{\bsfy}) = \frac{\textbf{Res}\left(\bar{\bsfy}\right)}{\|\textbf{Res}\left(\bar{\bsfy}\right)\|_{\bsfV^{\prime *}} }.
\end{align*}
Here we have to introduce our first assumptions to handle with the residual term $o\left( \|\bsfy^{\prime}\|_{\bsfV}^n \right)$ in the above formulation. We assume that $\|\bsfy^{\prime}\|_{\bsfV} < 1$ such that the residual term $ o\left( \|\bsfy^{\prime}\|_{\bsfV}^n \right) $ is negligible. Grouping the coefficients of the powers of $\varepsilon$ leads to a recurrence formula
\begin{align*}
& D_{\bsfy}\bsfB\left(\bsfw^{\prime} , \bar{\bsfy} \right)[\bsfy^{\prime}_1] = \hat{\mathsf R}(\bar{\bsfy})[\bsfw^{\prime}], \\
& D_{\bsfy}\bsfB\left(\bsfw^{\prime} , \bar{\bsfy} \right)[\bsfy^{\prime}_2] + \frac12 D_{\bsfy}^2\bsfB\left(\bsfw^{\prime} , \bar{\bsfy} \right)[\bsfy^{\prime}_1, \bsfy^{\prime}_1] = 0, \\
& D_{\bsfy}\bsfB\left(\bsfw^{\prime} , \bar{\bsfy} \right)[\bsfy^{\prime}_3] + \frac12 D_{\bsfy}^2\bsfB\left(\bsfw^{\prime} , \bar{\bsfy} \right)[\bsfy^{\prime}_1, \bsfy^{\prime}_2] + \frac12 D_{\bsfy}^2\bsfB\left(\bsfw^{\prime} , \bar{\bsfy} \right)[\bsfy^{\prime}_2, \bsfy^{\prime}_1] + \frac{1}{6} D_{\bsfy}^3\bsfB\left(\bsfw^{\prime} , \bar{\bsfy} \right)[\bsfy^{\prime}_1, \bsfy^{\prime}_1, \bsfy^{\prime}_1] = 0, \\
& \cdots
\end{align*}
The above recurrence formula leads to a sequence of coupled linear variational problems
\begin{align}
\label{eq:appendix_b_perturbation_series_eqn_1}
& D_{\bsfy}\bsfB\left(\bsfw^{\prime} , \bar{\bsfy} \right)[\bsfy^{\prime}_1] = \hat{\mathsf R}(\bar{\bsfy})[\bsfw^{\prime}], \displaybreak[2] \\
& D_{\bsfy}\bsfB\left(\bsfw^{\prime} , \bar{\bsfy} \right)[\bsfy^{\prime}_2] = - \frac12 D_{\bsfy}^2\bsfB\left(\bsfw^{\prime} , \bar{\bsfy} \right)[\bsfy^{\prime}_1, \bsfy^{\prime}_1], \displaybreak[2] \\
\label{eq:appendix_b_perturbation_series_eqn_k}
& D_{\bsfy}\bsfB\left(\bsfw^{\prime} , \bar{\bsfy} \right)[\bsfy^{\prime}_k] = - \sum_{l=2}^{k} \left\lbrace \sum_{\sum_{s}j_s = l}\frac{1}{l!} D_{\bsfy}^l\bsfB\left(\bsfw^{\prime},\bar{\bsfy} \right)[\underbrace{\bsfy^{\prime}_{j_{1}}, \cdots , \bsfy^{\prime}_{j_{l}}}_\text{l copies}] \right\rbrace, \quad \mbox{ for } 3 \leq k < n.
\end{align}
Solving the above system of equations gives the series expansion of $\bsfy^{\prime}$. In the above equations, the left-hand side is the same bilinear operator $D_{\bsfy}\bsfB\left(\bsfw^{\prime} , \bar{\bsfy} \right)[\cdot]$. We assume that there exists a fine-scale Green's operator $\bsfG^{\prime} : \bsfV^{\prime *} \rightarrow \bsfV^{\prime}$ such that
\begin{align*}
D_{\bsfy}\bsfB\left(\bsfw^{\prime}, \bar{\bsfy} \right)\left[ \bsfG^{\prime}\left( \mathfrak G \right) \right] = \mathfrak G [\bsfw^{\prime}] \qquad \mbox{ for }  \forall \mathfrak G \in \bsfV^{\prime *}.
\end{align*}  
With $\bsfG^{\prime}$ defined, we may formally solve the equations \eqref{eq:appendix_b_perturbation_series_eqn_1}-\eqref{eq:appendix_b_perturbation_series_eqn_k} as
\begin{align*}
\bsfy^{\prime}_m = \bsfG^{\prime}\left( \mathfrak G_m \right) \quad \mbox{ for } 1 \leq m < n,
\end{align*}
wherein
\begin{align*}
& \mathfrak G_1 := \hat{\mathsf R}(\bar{\bsfy})[\bsfw^{\prime}], \displaybreak[2] \\
& \mathfrak G_2 := - \frac12 D_{\bsfy}^2\bsfB\left(\bsfw^{\prime} , \bar{\bsfy} \right)[\bsfy^{\prime}_1, \bsfy^{\prime}_1], \displaybreak[2] \\
& \mathfrak G_k := - \sum_{l=2}^{k} \left\lbrace \sum_{\sum_{s}j_s = l}\frac{1}{l!} D_{\bsfy}^l\bsfB\left(\bsfw^{\prime},\bar{\bsfy} \right)[\underbrace{\bsfy^{\prime}_{j_{1}}, \cdots , \bsfy^{\prime}_{j_{l}}}_\text{l copies}] \right\rbrace, \quad \mbox{ for } 3 \leq k < n.
\end{align*}
Although $D_{\bsfy}\bsfB\left(\bsfw^{\prime} , \bar{\bsfy} \right)[\cdot]$ is a linear operator, it still a challenging task to derive an analytic form of the Green's operator.  Therefore, one more approximation needs to be introduced. The most simple but effective approach is to approximate $\bsfG$ by
\begin{align}
\label{eq:appendix_b_algebraic_G}
\bsfG^{\prime}(\mathfrak G) \approx \tilde{\bsfG}^{\prime}(\mathfrak G) = -\bm \tau \mathfrak G,
\end{align} 
wherein $\bm \tau \in \mathbb R^{7\times 7}$ \cite{Bazilevs2007a}. With this formula, we can solve the equations \eqref{eq:appendix_b_perturbation_series_eqn_1}-\eqref{eq:appendix_b_perturbation_series_eqn_k} and obtain the perturbation series $\bsfy^{\prime}_m = \bsfG^{\prime}\left( \mathfrak G_m \right) \approx \tilde{\bsfG}^{\prime}(\mathfrak G_m) = -\bm \tau \mathfrak G_m$. Therefore, we have $\bsfy^{\prime} = \sum_{m=1}^{\infty}\varepsilon^m \bsfy^{\prime}_m \approx -\bm\tau \sum_{m=1}^{\infty}\varepsilon^m \mathfrak G_m$. The infinite series $\bsfy^{\prime}$ needs to be truncated to numerically calculate $\bsfy^{\prime}$. It was shown that $\bsfy^{\prime}_1$ carries most of the subgrid energy \cite{Scovazzi2004}. This suggests that one may truncate the perturbation series as $\bsfy^{\prime} \approx \varepsilon \bsfy^{\prime}_1$. Then the approximation of the fine-scale component can be explicitly written as
\begin{align*}
\bsfy^{\prime} = \mathscr F^{\prime}\left(\bar{\bsfy}, \textbf{Res}\left(\bar{\bsfy} \right) \right)  \approx \tilde{\mathscr F}^{\prime}\left(\bar{\bsfy}, \textbf{Res}\left(\bar{\bsfy} \right) \right) := \varepsilon \bsfy^{\prime}_1 =  -\bm \tau \textbf{Res}\left(\bar{\bsfy} \right).
\end{align*}
In the above, the fine-scale component $\bsfy^{\prime}$ is systematically approximated. This approximation formula is utilized to construct the VMS formulations in this work.

\section{Proof of Proposition \ref{prop:Rk_2_zero}}
\label{appd:proof_prop_Rk_2_zero}
\begin{proof}
We first show that $\bar{\boldsymbol{\mathrm R}}_{k,(2)} = \bm 0$. Given a pair of predictors $\bm U_{n+1,(0)}$ and $\bm V_{n+1,(0)}$, we can write $\bar{\boldsymbol{\mathrm R}}_{k,(1)}$ explicitly as
\begin{align*}
\bar{\boldsymbol{\mathrm R}}_{k,(1)} = \frac{\alpha_m}{\gamma \Delta t_n} \left( \bm U_{n+1,(0)} - \bm U_n \right) + \left( 1 - \frac{\alpha_m}{\gamma} \right) \dot{\bm U}_n - \alpha_f \bm V_{n+1,(0)} - (1-\alpha_f) \bm V_n.
\end{align*}
With $\Delta \dot{\bm V}_{n+1,(1)}$, one can obtain $\Delta \bm V_{n+1,(1)}$ by
\begin{align*}
\Delta \bm V_{n+1,(1)} = \gamma \Delta t_n \Delta \dot{\bm V}_{n+1,(1)},
\end{align*}
due to \eqref{eq:gen_alpha_def_Y_n_plus_1}. Using \eqref{eq:newton_linear_system_smaller_eqn_2} and the above relation, one has
\begin{align}
\label{eq:appendixA_dudot_np1_1}
\Delta \dot{\bm U}_{n+1,(1)} = \frac{\alpha_f}{\alpha_m} \Delta \bm V_{n+1,(1)} - \frac{1}{\alpha_m}\bar{\boldsymbol{\mathrm R}}_{k,(1)}.
\end{align}
Then one has
\begin{align*}
\bar{\boldsymbol{\mathrm R}}_{k,(2)} &= \frac{\alpha_m}{\gamma \Delta t_n} \left( \bm U_{n+1,(1)} - \bm U_n \right) + \left( 1 - \frac{\alpha_m}{\gamma} \right) \dot{\bm U}_n - \alpha_f \bm V_{n+1,(1)} - (1-\alpha_f) \bm V_n \displaybreak[3]\\
&= \frac{\alpha_m}{\gamma \Delta t_n} \left( \bm U_{n+1,(0)} - \bm U_n \right) + \frac{\alpha_m}{\gamma \Delta t_n} \Delta \bm U_{n+1,(1)} + \left( 1 - \frac{\alpha_m}{\gamma} \right) \dot{\bm U}_n  - \alpha_f \bm V_{n+1,(1)} - (1-\alpha_f) \bm V_n \displaybreak[3] \\
&= \frac{\alpha_m}{\gamma \Delta t_n} \left( \bm U_{n+1,(0)} - \bm U_n \right) + \alpha_m \Delta \dot{\bm U}_{n+1,(1)} + \left( 1 - \frac{\alpha_m}{\gamma} \right) \dot{\bm U}_n - \alpha_f \bm V_{n+1,(1)} - (1-\alpha_f) \bm V_n \displaybreak[3] \\
&= \frac{\alpha_m}{\gamma \Delta t_n} \left( \bm U_{n+1,(0)} - \bm U_n \right) + \alpha_f \Delta \bm V_{n+1,(1)} - \bar{\boldsymbol{\mathrm R}}_{k,(1)} + \left( 1 - \frac{\alpha_m}{\gamma} \right) \dot{\bm U}_n  - \alpha_f \bm V_{n+1,(1)} - (1-\alpha_f) \bm V_n \displaybreak[3] \\
&= \alpha_f \Delta \bm V_{n+1,(1)} + \alpha_f \bm V_{n+1,(0)} - \alpha_f \bm V_{n+1,(1)} \displaybreak[3] \\
&= \bm 0.
\end{align*}
Now we just need to show that $\bar{\boldsymbol{\mathrm R}}_{k,(i+1)} = \bm 0$ if $\bar{\boldsymbol{\mathrm R}}_{k,(i)} = \bm 0$. Using $\bar{\boldsymbol{\mathrm R}}_{k,(i)} = \bm 0$, the update formula \eqref{eq:newton_linear_system_smaller_eqn_2} can be simplified as
\begin{align}
\label{eq:appendixA_simple_Ri}
\Delta \dot{\bm U}_{n+1,(i)} = \frac{\alpha_f \gamma \Delta t_n}{\alpha_m} \Delta \dot{\bm V}_{n+1,(i)} - \frac{1}{\alpha_m}\bar{\boldsymbol{\mathrm R}}_{k,(i)} = \frac{\alpha_f \gamma \Delta t_n}{\alpha_m} \Delta \dot{\bm V}_{n+1,(i)} = \frac{\alpha_f}{\alpha_m} \Delta \bm V_{n+1,(i)}.
\end{align}
Now we can expand $\bar{\boldsymbol{\mathrm R}}_{k,(i+1)}$ as
\begin{align*}
\bar{\boldsymbol{\mathrm R}}_{k,(i+1)} &= \frac{\alpha_m}{\gamma \Delta t_n} \left( \bm U_{n+1,(i)} - \bm U_n \right) + \left( 1 - \frac{\alpha_m}{\gamma} \right) \dot{\bm U}_n - \alpha_f \bm V_{n+1,(i)} - (1-\alpha_f) \bm V_n \\
&= \frac{\alpha_m}{\gamma \Delta t_n} \left( \bm U_{n+1,(i-1)} - \bm U_n \right) + \left( 1 - \frac{\alpha_m}{\gamma} \right) \dot{\bm U}_n - \alpha_f \bm V_{n+1,(i-1)} - (1-\alpha_f) \bm V_n \\
& \hspace{5mm} + \frac{\alpha_m}{\gamma \Delta t_n} \Delta \bm U_{n+1,(i)} - \alpha_f \Delta \bm V_{n+1,(i)} \\
&= \bar{\boldsymbol{\mathrm R}}_{k,(i)} + \frac{\alpha_m}{\gamma \Delta t_n} \Delta \bm U_{n+1,(i)} - \alpha_f \Delta \bm V_{n+1,(i)} \\
&= \alpha_m \Delta \dot{\bm U}_{n+1,(i)} - \alpha_f \Delta \bm V_{n+1,(i)} \\
&= \bm 0.
\end{align*}
By mathematical induction, we have $\bar{\boldsymbol{\mathrm R}}_{k,(i)} = \bm 0$ for $i \geq 2$.
\end{proof}
\begin{remark}
In the above proof, we do not require $\Delta \dot{\bm V}_{n+1,(1)}$ to be a solution of the linear system \eqref{eq:newton_linear_system_smaller_eqn_1}. Hence, in practice, one can set $\bar{\boldsymbol{\mathrm R}}_{k,(i)} = \bm 0$ in \eqref{eq:newton_linear_system_smaller_eqn_1} for all $i \geq 1$. This leads to inconsistent updates of $\Delta \dot{P}_{n+1,(1)}$ and $\Delta \dot{\bm V}_{n+1,(1)}$ only in the first Newton-Raphson iteration, and significantly simplify the implementation. In our numerical experience, this choice does not deteriorate the convergence of the predictor multi-corrector algorithm.
\end{remark}

\bibliographystyle{plain}
\bibliography{FSIbib}

\end{document}